\documentclass[letterpaper,twocolumn]{article}
\usepackage[margin=.8in]{geometry}
\usepackage{xcolor}
\usepackage[T1]{fontenc}
\usepackage{times}
\usepackage{hyperref}
\usepackage{natbib}
\usepackage{amsmath}
\usepackage{amsthm}
\usepackage{amssymb}
\usepackage{graphicx}
\usepackage{mleftright}
\usepackage{stmaryrd}
\usepackage{amsthm}
\theoremstyle{definition}
\newtheorem{theorem}{Theorem}
\newtheorem{definition}{Definition}

\newtheorem{corollary}{Corollary}
\usepackage{float}
\usepackage{authblk}
\usepackage{listings}
\DeclareMathOperator{\expect}{E}
\DeclareMathOperator*{\diag}{diag}
\DeclareMathOperator{\argmin}{argmin}
\DeclareMathOperator{\argmax}{argmax}

\title{ApproxED: Approximate exploitability descent via learned best responses}

\author[1]{Carlos Martin}
\author[1,2,3,4]{Tuomas Sandholm}
\affil[ ]{\{cgmartin, sandholm\}@cs.cmu.edu}
\affil[1]{Carnegie Mellon University}
\affil[2]{Strategy Robot, Inc.}
\affil[3]{Optimized Markets, Inc.}
\affil[4]{Strategic Machine, Inc.}
\date{}

\begin{document}
\maketitle
\begin{abstract}
There has been substantial progress on finding game-theoretic equilibria.
Most of that work has focused on games with finite, discrete action spaces.
However, many games involving space, time, money, and other fine-grained quantities have continuous action spaces (or are best modeled as having such).
We study the problem of finding an approximate Nash equilibrium of games with continuous action sets.
The standard measure of closeness to Nash equilibrium is exploitability, which measures how much players can benefit from unilaterally changing their strategy.
We propose two new methods that minimize an approximation of exploitability with respect to the strategy profile.
The first method uses a learned best-response function, which takes the current strategy profile as input and outputs candidate best responses for each player.
The strategy profile and best-response functions are trained simultaneously, with the former trying to minimize exploitability while the latter tries to maximize it.
The second method maintains an ensemble of candidate best responses for each player.
In each iteration, the best-performing elements of each ensemble are used to update the current strategy profile.
The strategy profile and ensembles are simultaneously trained to minimize and maximize the approximate exploitability, respectively.
We evaluate our methods on various continuous games and GAN training, showing that they outperform prior methods.
\end{abstract}

\section{Introduction}

Most work concerning equilibrium computation has focused on games with finite, discrete action spaces.
However, many games involving space, time, money, \emph{etc.} have continuous action spaces.
Examples include continuous resource allocation games~\citep{Ganzfried_2021}, security games in continuous spaces~\citep{Kamra_2017, Kamra_2018, Kamra_2019}, network games~\citep{Ghosh_2019}, military simulations and wargames~\citep{Marchesi20:Learning}, and video games~\citep{Berner19:Dota,Vinyals19:Grandmaster}.
Moreover, even if the action space is discrete, it can be fine-grained enough to treat as continuous for the purpose of computational efficiency~\citep{Borel38:Traite,Chen06:Mathematics,Ganzfried10:Computing}.

As the goal, we use the standard solution concept of a \emph{Nash equilibrium (NE)}, that is, a strategy profile for which each strategy is a best response to the other players' strategies.
The main measure of closeness to NE is \emph{exploitability}, which measures how much players can benefit from unilaterally changing their strategy.
Typically, we seek an exact NE, that is, a strategy profile for which exploitability is zero.
As some prior work in the literature, we can try to search for NE by performing gradient descent on exploitability, since it is non-negative and its zero set is precisely the set of NE.
However, evaluating exploitability requires computing best responses to the current strategy profile, which is itself a nontrivial problem in many games.

We present two new methods that minimize an approximation of the exploitability with respect to the strategy profile.
For the first method, we compute this approximation using \emph{learned best-response functions}, which take the strategy profile as input and return predicted best responses.
We train the strategy profile and best-response functions simultaneously, with the former trying to minimize exploitability while the latter try to maximize it.
The second method maintains and updates a set of \emph{best-response ensembles} for each player.
Each ensemble maintains multiple candidate best responses to the current strategy profile for that player.
In each iteration, the best-performing element of each ensemble is used to update the current strategy profile.
The strategy profile and best-response ensembles are simultaneously trained to minimize and maximize the approximate exploitability, respectively.
Our experiments on various continuous games show that our techniques outperform prior approaches.

The rest of the paper is structured as follows.
In \S\ref{sec:formulation}, we introduce terminology and formulate the problem.
In \S\ref{sec:related_work}, we present related work.
In \S\ref{sec:method}, we present our methods.
In \S\ref{sec:experiments}, we present our experiments.
Finally, in \S\ref{sec:conclusion}, we present conclusions and discuss directions for future research.
In the appendix, we present additional figures, our theoretical analysis, additional related work, and code.

\section{Problem formulation}
\label{sec:formulation}

We use the following notation.
Hereafter, \(\triangle \mathcal{X}\) is the set of probability distributions on a space \(\mathcal{X}\), \([n] = \{0, \ldots, n - 1\}\) is the set of natural numbers less than a natural number \(n \in \mathbb{N}\), \(|\mathcal{X}|\) is the cardinality of a set \(\mathcal{X}\), \(\mathbb{S}_n\) is the unit \(n\)-sphere, and \( \llbracket \phi \rrbracket\) is an Iverson bracket (1 if \(\phi\) and 0 otherwise).

A \emph{game} is a tuple \((\mathcal{I}, \mathcal{X}, u)\) where \(\mathcal{I}\) is a set of players, \(\mathcal{X}_i\) is a strategy set for player \(i\), and \(u_i : \prod_i \mathcal{X}_i \to \mathbb{R}\) is a utility function for player \(i\).
A strategy profile \(x \in \prod_i \mathcal{X}_i\) is an assignment of a strategy to each player.
A game is zero-sum if \(\sum_{i \in \mathcal{I}} u_i = 0\).
Given a strategy profile \(x\), Player \(i\)'s regret is \(R_i(x) = \sup_{y_i \in \mathcal{X}_i} u_i(y_i, x_{-i}) - u_i(x)\), where \(x_{-i}\) denotes the other players' strategies.
It is the highest utility Player \(i\) could gain from unilaterally changing its strategy.
A strategy profile \(x\) is an \(\varepsilon\)-equilibrium if \(\sup_{i \in \mathcal{I}} R_i(x) \leq \varepsilon\).
A 0-equilibrium is called a \emph{Nash equilibrium (NE)}.
In an NE, each player's strategy is a best response to the other players' strategies, that is, \(u_i(x) \geq u_i(y_i, x_{-i})\) for all \(i \in \mathcal{I}\) and \(y_i \in \mathcal{X}_i\).
Conditions for existence and uniqueness of NE can be found in the appendix.

The standard measure of closeness to NE is exploitability, also known as NashConv~\citep{Lanctot17:Unified, Lockhart_2019, Walton_2021, Timbers_2022}.
It is defined as \(\Phi = \sum_{i \in \mathcal{I}} R_i\).
(In a two-player zero-sum game, \(\Phi\) reduces to the so-called \textit{duality gap} \citep{Grnarova_2021}.)
It is non-negative everywhere and zero precisely at NE.
Thus finding an NE is equivalent to minimizing exploitability~\citep{Lockhart_2019}.

Let \(\phi(x, y) = \sum_{i \in \mathcal{I}} \mleft( u_i(y_i, x_{-i}) - u_i(x) \mright)\) be the \emph{Nikaido-Isoda (NI)} function \citep{Nikaido_1955,Flam_1996,Flam_2008,Hou_2018}.
Since \(\Phi(x) = \sup_y \phi(x, y)\), finding an NE is equivalent to solving the min-max problem \(\inf_x \sup_y \phi(x, y)\).
Some prior work has used this function to search for NE \citep{Berridge_1970,Uryasev_1994,Krawczyk_2000,Krawczyk_2005,Flam_2008,Gurkan_2009,vonHeusinger_2009a,vonHeusinger_2009b,Qu_2013,Hou_2018,Raghunathan_2019,Tsaknakis_2021}.

\section{Related work}
\label{sec:related_work}

In this section we review the prior methods for solving continuous games, which we use as baselines in our experiments.
They can be characterized by the \textit{ordinary differential equations (ODEs)} shown in Table \ref{tab:odes}.
Here, \(v = \diag \nabla u\) is the \emph{simultaneous gradient}.
Each component \(v_i = \nabla_i u_i\) is the gradient of a player's utility with respect to their strategy.
It is therefore a vector field on the space of strategy profiles.
The fact that this vector field need not be conservative (that is, the gradient of some potential), like it is in ordinary gradient descent, is the main source of difficulties for applying standard gradient-based optimization methods, since trajectories can cycle around fixed points rather than converging to them.
Additionally, \(J = \nabla v\) is the Jacobian of the vector field \(v\), \(J^\top\) is its transpose, \(J_a = \tfrac{1}{2}(J - J^\top)\) is its antisymmetric part, and \(J_o\) is its off-diagonal part (replacing its diagonal with zeroes).
Dots indicate derivatives with respect to time, \(\gamma > 0\) is a hyperparameter, and \(v|_y\) denotes \(v\) evaluated at \(y\) rather than \(x\).
The actual optimization is done by discretizing each ODE in time.
For example, SG is discretized as \(x_{i+1} = x_i + \eta v_i\) and OP is discretized as \(x_{i+1} = x_i + \eta v_i + \gamma (v_i - v_{i-1})\), where \(\eta > 0\) is a stepsize and \(v_j = v(x_j)\).

\begin{table}
    \centering
    \begin{tabular}{ll}
        \hline
        SG & \(v\) \\
        EG & \(v|_{x + \gamma v}\) \\
        OP & \((I + \gamma \tfrac{\mathrm{d}}{\mathrm{d}t}) v = v + \gamma \dot{v}\) \\
        CO & \((I - \gamma J^\top) v = v - \gamma \nabla \tfrac{1}{2} \|v\|^2\) \\
        SGA & \((I - \gamma J_a^\top) v\) \\
        SLA & \((I + \gamma J) v\) \\
        LA & \((I + \gamma J_o) v\) \\
        LOLA & \((I + \gamma J_o) v - \gamma \diag J_o^\top \nabla u\) \\
        LSS & \((I + J^\top J^{-1}) v\) \\
        PCGD & \((I - \gamma J_o)^{-1} v\) \\
        ED & \(-\nabla_x \sup_y \phi(x, y) = -\nabla_x \Phi(x)\) \\
        GNI & \(-\nabla_x \phi(x, x + \gamma v)\) \\
        \hline
    \end{tabular}
    \caption{Value of \(\dot{x}\) in each method's ODE.}
    \label{tab:odes}
\end{table}

\emph{Simultaneous gradients (SG)} maximizes each player's utility independently, as if the other players are fixed.
\emph{Extragradient (EG)} \citep{Korpelevich_1976} takes a step in the direction of the simultaneous gradient and uses the simultaneous gradient at that new point to take a step from the original point.
\citet{golowich2020last} proved a tight last-iterate convergence guarantee for EG.
\emph{Optimistic gradient (OP)} \citep{Popov1980,Daskalakis_2017,Hsieh_2019} uses past gradients to predict future gradients and update according to the latter.
\emph{Consensus optimization (CO)} \citep{Mescheder_2017} penalizes the magnitude of the simultaneous gradient, encouraging ``consensus'' between players that attracts them to fixed points.
\emph{Symplectic gradient adjustment (SGA)} \citep{Balduzzi_2018} (also known as Crossing-the-Curl \citep{gemp2018global}) reduces the rotational component of game dynamics by using the antisymmetric part of the Jacobian.
\emph{Lookahead (LA)} \citep{Zhang_2010} excludes the diagonal components of the Jacobian.
Each player predicts the behaviour of other players after a step of naive learning, but assumes this step will occur independently of the current optimisation.
In \emph{symmetric lookahead (SLA)} \citep{Letcher_2018b}, instead of best-responding to opponents' learning, each player responds to \emph{all} players learning, including themselves.
It is a linearized version of EG \citep[Lemma 1.35]{Domingo_2019}.
In \emph{learning with opponent-learning awareness (LOLA)} \citep{Foerster_2018}, a learner optimises its utility under one step look-ahead of opponent learning.
Instead of optimizing utility under the current parameters, it optimises utility after the opponent updates its policy with one naive learning step.
\citet{Mazumdar_2019} proposed \emph{local symplectic surgery (LSS)} to find local NE in two-player zero-sum games.
It requires solving a linear system on each timestep, which is prohibitive for high-dimensional parameter spaces.
Hence, its authors propose a two-timescale approximation that updates the strategy profile while simultaneously improving an approximate solution to the linear system.
\emph{Competitive gradient descent (CGD)} \citep{Schaefer_2019} naturally generalizes gradient descent to the two-player setting.
On each iteration, it jumps to the NE of a quadratically-regularized bilinear local approximation of the game.
Its convergence and stability properties are robust to strong interactions between the players without adapting the stepsize.
\emph{Polymatrix competitive gradient descent (PCGD)} \citep{Ma_2021} generalizes CGD to more than two players.
It jumps to the NE of a quadratically-regularized local polymatrix approximation of the game.
The series expansion of PCGD to zeroth and first order in \(\gamma\) yields SG and LA \citep[Proposition 4.4]{Willi_2022}, respectively,  since \((I - \gamma M)^{-1} = I + \gamma M + \gamma^2 M^2 + \ldots\) for sufficiently small \(\gamma\).
CGD and PGCD require solving a linear system of equations on each iteration, which is prohibitive for high-dimensional parameter spaces \citep[p. 10]{Ma_2021}.
\emph{Exploitability descent (ED)} \citep{Lockhart_2019} directly minimizes exploitability, and converges to approximate equilibria in two-player zero-sum extensive-form games.
However, it requires computing best responses \(y\) on each iteration, which is inefficient and/or intractable in general games.
\emph{Gradient-based Nikaido-Isoda (GNI)} \citep{Raghunathan_2019} minimizes a local approximation of exploitability that uses local best responses \(y = x + \gamma v\).
\citet{goktas2022exploitability} recast the exploitability-minimization problem as a min-max optimization problem and obtain polynomial-time first-order methods for computing variational equilibria in convex-concave cumulative regret pseudo-games with jointly convex constraints.
They present two algorithms called \emph{extragradient descent ascent (EDA)} and \emph{augmented descent ascent (ADA)}.
We benchmark against EDA but not ADA because, unlike the other baselines, it requires \emph{multiple} substeps of gradient ascent \emph{per timestep} to approximate a best response.
(As the authors note, ``ADA's accuracy depends on the accuracy of the best-response found''.)
The methods presented in this section have been analyzed in prior work \citep{Balduzzi_2018, Letcher_2018, Letcher_2019, Mertikopoulos_2019, Grnarova_2019, Mazumdar_2019, Hsieh_2021, Willi_2022}.
Due to space constraints, we describe additional related research in the appendix.

\section{Proposed methods}
\label{sec:method}

We are given a utility function \(u\) and our goal is to find an NE.
Since the exploitability function \(\Phi(x) = \sup_y \phi(x, y)\) is non-negative everywhere, and zero precisely at NE, we reformulate the problem of finding an NE as \emph{finding a global minimum of the exploitability function}.
That is, we wish to solve the min-max optimization problem \(\inf_x \sup_y \phi(x, y)\).
This is equivalent to finding a minimally-exploitable strategy for a two-player zero-sum meta-game with utility function \(\phi\).

To find a minimum, we could try performing gradient decent on \(\Phi\), like ED.\footnote{For ease of exposition, we assume that utility functions are differentiable. If they are not differentiable, we can replace any gradient with a \emph{pseudogradient}, which is the gradient of a smoothened version of the function (\emph{e.g.}, the function convolved with a narrow Gaussian).
An unbiased estimator for this pseudogradient can be obtained by evaluating the function at randomly-sampled perturbed points and using their values to approximate directional derivatives along those directions \citep{Duchi_2015, Nesterov_2017, Shamir_2017, Salimans_2017, Berahas_2022, metz2021gradients}.}
However, the latter requires best-response oracles.
To solve this problem, we can try to perform gradient descent on \(x\) and \(y\) simultaneously: \(\dot{x} = \nabla_x \phi(x, y) \), \(\dot{y} = -\nabla_y \phi(x, y)\).
Unfortunately, this approach can fail even in simple games.
For example, consider the simple bilinear game with \(u(x, y) = x y\).
The unique Nash equilibrium is at the origin.
However, simultaneous gradient descent fails to converge to it, and instead cycles around it indefinitely.
The essence of this cycling problem is that Player 2 has to ``relearn'' a good response to Player 1 every time the Player 1's strategy switches sign.
This is a general problem for games.
``Small'' changes in other players' strategies can cause ``large'' (discontinuous) changes in a player's best response.
When such changes occur, players have to ``relearn'' how to respond to the other players' strategies.
We propose two methods to tackle this problem.
These are described in the next two subsections, respectively.

\subsection{Best-response functions}
We reformulate the problem as minimizing \(\phi(x, b^*(x))\) with respect to \(x\), where \(b^* : \mathcal{X} \to \mathcal{Y}\) is a \emph{function} that satisfies \(b^*(x) \in \argmax_{y \in \mathcal{Y}} \phi(x, y)\).
Since \(b^*\) is a \emph{function}, it can map different strategies for Player 1 to different strategies for Player 2.
Thus it can \emph{immediately} adapt to Player 1's strategy and avoid the cycling problem.

More precisely, suppose \(\mathcal{Y}\) is compact and \(\phi : \mathcal{X} \times \mathcal{Y} \to \mathbb{R}\) is continuous in its second argument.
Let \(x \in \mathcal{X}\).
By the extreme value theorem, a continuous real-valued function on a non-empty compact set attains its extrema.
Therefore, there exists \(y \in \mathcal{Y}\) such that \(\phi(x, y) = \sup_{y \in \mathcal{Y}} \phi(x, y)\).
Since this is true for every \(x \in \mathcal{X}\), there exists a function \(b^* : \mathcal{X} \to \mathcal{Y}\) such that, for every \(x \in \mathcal{X}\), \(\phi(x, b^*(x)) = \sup_{y \in \mathcal{Y}} \phi(x, y)\).
That is, \(b^*\) is a best-response function.

Even when $\mathcal{Y}$ is not compact and \(\phi\) does not attain its extrema, one can define a best-response \emph{value} for any \(x \in \mathcal{X}\) as \(\sup_{y \in \mathcal{X}} \phi(x, y)\), provided the latter exists.
In that case, we have the following.
Let \(\varepsilon > 0\) and \(x \in \mathcal{X}\).
Any function gets arbitrarily close to its supremum (continuity is not required).
Therefore, there exists a \(y \in \mathcal{Y}\) such that \(\phi(x, y) + \varepsilon \geq \sup_{y \in \mathcal{Y}} \phi(x, y)\).
Therefore, there exists a function $b_\varepsilon : \mathcal{X} \to \mathcal{Y}$ such that, for every $x \in \mathcal{X}$, \(\phi(x, b_\varepsilon(x)) + \varepsilon \geq \sup_{y \in \mathcal{Y}} \phi(x, y)\).
That is \(b_\varepsilon\) is an \(\varepsilon\)-\emph{approximate} best-response function.

To find \(x\) and \(b = b^*\) simultaneously, we can perform simultaneous gradient ascent: \(\dot{x} = \nabla_x \phi(x, b(x))\), \(\dot{b} = -\nabla_b \phi(x, b(x))\), where \(\nabla_x \phi(x, b(x))\) is a total (not partial) derivative.
That is, the best response function tries to \emph{increase} the exploitability while the strategy profile tries to \emph{decrease} it.
Since \(b\) is a function, Player 1's changing behavior poses no fundamental hindrance to it learning good responses and ``saving'' them for later use if Player 1's behavior changes.
It could even learn a good approximation to the true best-response function, leaving Player 1 to face a simple standard optimization problem.

If \(\mathcal{X}\) is infinite and \(\mathcal{Y}\) is nontrivial, \(\mathcal{X} \to \mathcal{Y}\) has infinite dimension.
To represent and optimize \(b\) in practice, we need a finite-dimensional parameterization of (a subset of) this function space.
More precisely, if \(b\) is parameterized by \(w\) and is (approximately) surjective onto \(\mathcal{Y}\), then \(\inf_{x \in \mathcal{X}} \Phi(x) = \inf_{x \in \mathcal{X}} \sup_{y \in \mathcal{Y}} \phi(x, y) \approx \inf_{x \in \mathcal{X}} \sup_{w \in \mathcal{W}} \phi(x, b(w, x))\).
Inspired by this idea, we propose jointly optimizing \(x\) and \(w\) according to the following ODE system: \(\dot{x} = -\nabla_x \phi(x, b(w, x))\), \(\dot{w} = +\nabla_w \phi(x, b(w, x))\).
We call our method \emph{approximate exploitability descent with learned best-response functions (ApproxED-BRF)}.

The best-response function can take on many possible forms.
One possibility is to use a neural network.
Neural networks are a universal class of function approximators and have a powerful ability to generalize well across inputs \citep{cybenko1989approximation,hornik1989multilayer,hornik1991approximation,leshno1993multilayer,pinkus1999approximation}.
Neural networks are also, by far, the most popular function approximators used in game solving.
Therefore, we use this approach in our experiments.
On one hand, in games where each player's strategy is simply a vector of probabilities, our network takes as input a strategy profile and outputs another strategy profile where each player's strategy is the player's approximate best response.
(We actually represent each player's strategy by a vector of unconstrained real numbers and then use a softmax---one softmax per information set in the game---to convert the reals to probability distributions over actions.)

On the other hand, we also experiment with settings where each player's strategy is \emph{itself} is represented by a neural network.
In this case, the best-response functions take those networks' parameters as input.
In other words, we adapt the concept of \emph{hypernetworks}~\citep{schmidhuber1992learning,ha2016hypernetworks,lorraine2018stochastic,mackay2019self,bae2020delta} to the game-theoretic context.
We note that, to obtain good performance, the true best response function, which may be discontinuous, need not be represented exactly, but only approximated.
Our experimental results indicate that the approximation yielded by the neural network performs well across a wide class of games.

\subsection{Best-response ensembles}
For our second approach, we reformulate the problem as \(\inf_{x \in \mathcal{X}} \max_{j \in \mathcal{J}} \phi(x, y_j)\) where \(\mathcal{J}\) is a finite set of indices, and \(x \in \mathcal{X}\) and \(y : \mathcal{J} \to \mathcal{Y}\) are trainable parameters.
That is, we use an \emph{ensemble} of \(|\mathcal{J}|\) responses to \(x\), where the \emph{best} response is selected automatically by evaluating \(x\) against each \(y_j\) and taking the one that attains the best value.
Each individual \(y_j\) is a \emph{strategy} for the original game.
Since there are multiple responses in the ensemble, each one can ``focus on'' tackling a particular ``type'' of behavior from \(x\) without having to change drastically when the latter changes.
We can then train \(x\) and \(y\) simultaneously: \(\dot{x} = -\nabla_x \max_{j \in \mathcal{J}} \phi(x, y_j)\), \(\dot{y} = \nabla_{y} \max_{j \in \mathcal{J}} \phi(x, y_j)\).
That is, \(x\) improves against the best \(y_j\), while the best \(y_j\) improves against \(x\).
Ties are broken in indexical lower (lower indices first).
This allows for symmetry breaking if the ensemble elements are initially equal and the game is deterministic.

There is an issue with the aforementioned scheme, however.
If one of the ensemble elements \(y_j\) strictly dominates the others for all encountered \(x\), then the other elements will never be selected under the maximum operator.
Thus they will never have a chance to change, improve performance, and thus contribute.
In that case, the scheme degenerates to ordinary simultaneous gradient ascent.
We observed this degeneracy in some games.

To solve this issue, we introduced the following approach.
To give \emph{all} \(y_j\) some chance to improve, while incentivizing them to ``focus'' on particular types of \(x\) rather than cover all cases, we use a \emph{rank-based weighting} approach.
Specifically, we let \(\dot{y} = \nabla_{y} \operatorname{mix}_{j \in \mathcal{J}} \phi(x, y_j)\) where
\(\operatorname{mix}_{j \in \mathcal{J}} a_j = \frac{1}{|\mathcal{J}|} \sum_{j \in \mathcal{J}} r_j a_j\) and \(r_j \in \{1, \ldots, |\mathcal{J}|\}\) is the ordinal rank of element \(j\).
This makes better elements receive a higher weight.
Thus the best \(y_j\) has the most incentive to adapt against the current \(x\), while others have less incentive, but still some nonetheless.
Since the weight of each ensemble element depends only on the rank or order of values, it is invariant under monotone transformations of the utility function.

Our method is defined by the ODE system
\(\dot{x} = -\nabla_x \sum_{i \in \mathcal{I}} \mleft(\max_{j \in \mathcal{J}} u_i(y_{ij}, x_{-i}) - u_i(x)\mright)\)
and
\(\dot{y} = +\nabla_y \sum_{i \in \mathcal{I}} \mleft(\operatorname*{mix}_{j \in \mathcal{J}} u_i(y_{ij}, x_{-i}) - u_i(x)\mright)\).
Here, \(y = \{y_{ij}\}_{i \in \mathcal{I}, j \in \mathcal{J}}\) is an ensemble of \(|\mathcal{J}|\) responses for each individual player \(i \in \mathcal{I}\).
One can compute the \(u_i(y_{ij}, x_{-i})\) and their gradients in parallel for \(i \in \mathcal{I}, j \in \mathcal{J}\).
Hence, with access to \(O(|\mathcal{I}| |\mathcal{J}|)\) cores, the method can run approximately as fast as standard simultaneous gradient ascent.
Due to this parallelism, the ensemble size can be as big as the amount of memory and number of cores or workers available allows.
We call our method \emph{approximate exploitability descent with learned best-response ensembles (ApproxED-BRE)}.

\section{Experiments}
\label{sec:experiments}

In this section, we present our experiments.
We use a learning rate of \(\eta = 10^{-3}\) and \(\gamma = 10^{-1}\).
For BRF's best-response function, we use a fully-connected network with a hidden layer of size 32, the tanh activation function, and He weight initialization \citep{He_2015}.
We do not try to find the best neural architecture, because this problem comprises an entire field, may be task-specific, and is not the focus of our paper.
Thus our experiments are conservative, in the sense that our technique could perform even better compared to the baselines if engineering effort were spent tuning the neural network.
For BRE, we use ensembles of size 10 for each player.
For each experiment, we ran 64 trials.
In our plots, solid lines show the mean across trials, and bands show its standard error.
For games with stochastic utility functions, we used a batch size of 64.
Each trial was run on one NVIDIA A100 SXM4 40GB GPU on a computer cluster.

Some of the benchmarks are based on normal-form or extensive-form games with finite action sets, and thus finite-dimensional continuous mixed strategies.
While there are algorithms for such games that might have better performance (such as counterfactual regret minimization~\citep{Zinkevich07:Regret} and its fastest new variants~\cite{Farina21:Faster,Brown19:Solving}), these do not readily generalize to general continuous-action games.
Thus we are interested in comparing only to those algorithms which, like ours, \emph{do} generalize to continuous-action games, namely those described in \S\ref{sec:related_work}.

We parameterize mixed strategies on finite action sets (\emph{e.g.}, for normal-form games, or at a particular information set inside an extensive-form game) using logits.
The action probabilities are obtained by applying the softmax function, which ensures they are non-negative and sum to 1.
The utilities are the expected utilities under the resulting mixed strategies.
Therefore, our games games are continuous in the \emph{nonlinear} and \emph{high-dimensional} space of mixed behavioral strategies parameterized by logits, which is the strategy space we optimize over.

If a player must randomize over a \emph{continuous} action set, we use an implicit density model, also called a deep generative model \citep{Ruthotto_2021}.
It samples noise from some base distribution, such as a multivariate standard normal distribution, and feeds it to a neural network, inducing a transformed probability distribution on the output space.
Unlike an explicit parametric distribution on the output space, it can flexibly model a wide class of distributions.
A \textit{generative adversarial network (GAN)} \citep{goodfellow2020generative} is one example of an implicit density model.

Additional figures and tables from our experiments can be found in the appendix.
The games we use are equal to or greater in size and complexity than those used as benchmarks in the papers of the other methods we compare to.

\paragraph{Saddle-point game}
The saddle-point game is a two-player zero-sum game with actions that are real numbers and utility function \(u_1(x, y) = -u_2(x, y) = x y\).
It has a unique NE at the origin.
Figure \ref{fig:saddle} shows performances.
Our methods converge fastest.

\begin{figure}
    \centering
    \includegraphics[width=.7\linewidth]{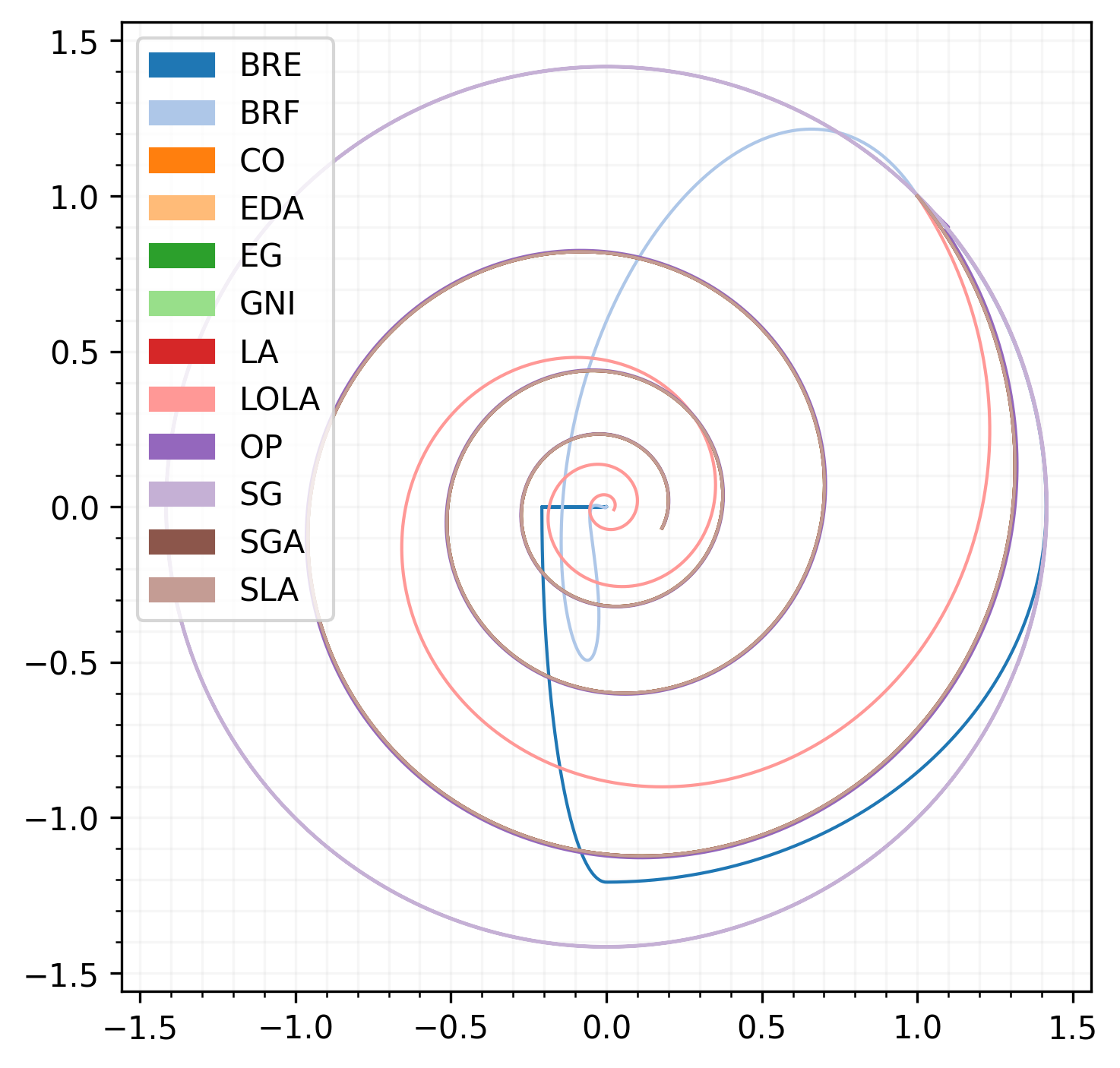}
    \includegraphics[width=.8\linewidth]{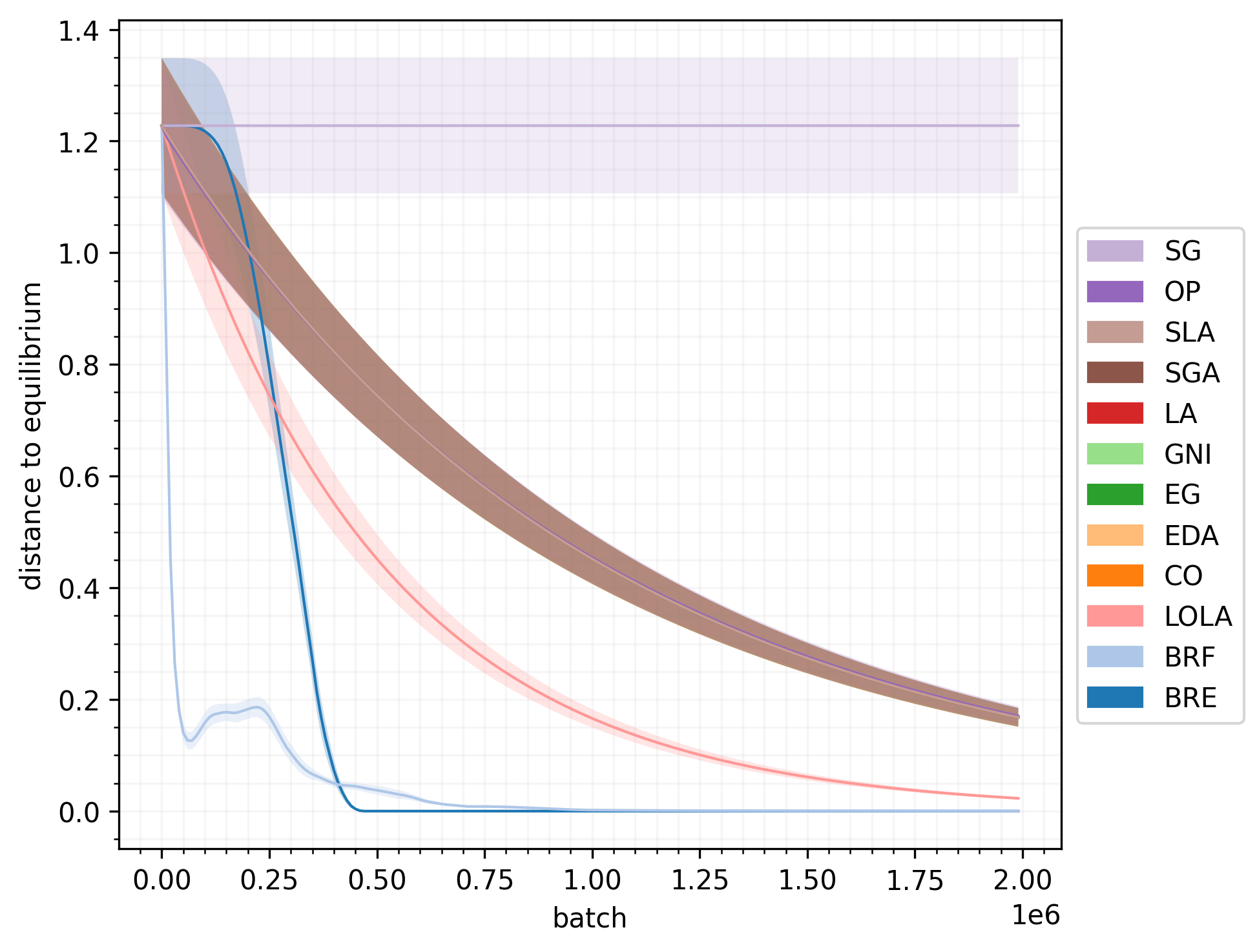}
    \caption{Saddle-point game.
    Top: Trajectories from one trial.
    Bottom: Distance to equilibrium.}
    \label{fig:saddle}
\end{figure}

\paragraph{GAN training}
\emph{Generative adversarial networks (GANs)}~\citep{goodfellow2020generative} are a prominent approach for generative modeling that use deep learning.
A GAN consists of two neural networks: a generator and a discriminator.
The generator maps latent noise to a data sample.
The discriminator maps a data sample to a probability.
The generator learns to generate fake data, while the discriminator learns to distinguish it from real data.
More precisely, the generator \(G\) and discriminator \(D\) play the zero-sum ``game'' \(\min_G \max_D V(D, G)\) where \(V(D, G) = \expect_{x \sim \mathcal{D}} \log D(x) + \expect_{z \sim \mathcal{N}} \log (1 - D(G(z)))\).
Here, \(\mathcal{D}\) is the real data distribution and \(\mathcal{N}\) is a fixed latent noise distribution, usually a multivariate standard Gaussian.
Various GAN variants with different architectures and loss functions exist in the literature.
Surveys can be found in \citet{wang2017generative, pan2019recent, jabbar2021survey, gui2021review}.

GAN training is a very high-dimensional problem, with a highly nontrivial utility function, since the strategies are entire neural network parameters for the generator and discriminator.
We test the equilibrium-finding methods on the following datasets. 
The \emph{ring} dataset consists of a mixture of 8 Gaussians with a standard deviation of 0.1 whose means are equally spaced around a circle of radius 1.
The \emph{grid} dataset consists of a mixture of 9 Gaussians with a standard deviation of 0.1 whose means are laid out in a regular square grid spanning from \(-1\) to \(+1\) in each coordinate.
The \emph{spiral} dataset consists of a noisy Archimedean spiral.
More precisely, we let
\(t \sim \mathcal{U}(0, 1)\),
\(r = \sqrt{t}\),
\(\theta = 2 \pi r n\),
\(x = \mathcal{N}(r \cos \theta, \sigma)\),
and \(y = \mathcal{N}(r \sin \theta, \sigma)\).
Here, \(n\) is the number of turns (we use 2) and \(\sigma\) is the standard deviation of the noise (we use 0.05).
Finally, the \emph{cube} dataset consists of points sampled uniformly from the edges of a cube and perturbed with Gaussian noise of scale 0.05.

In all cases, the generator's latent noise distribution is a standard Gaussian matching the dimension of the dataset.
The generator and discriminator have hidden layers of size 32.
Figures \ref{fig:gan_images_ring}, \ref{fig:gan_images_grid}, \ref{fig:gan_images_spiral}, and \ref{fig:gan_images_cube} show what the resulting data distributions look like after training.

We also test on MNIST \citep{deng2012mnist}, a dataset of 70,000 \(28 \times 28\) grayscale images of handwritten digits from 0 to 9.
The generator and discriminator networks are the same as before, and fully connected, but with the hidden layer size increased to 256 and the noise dimension increased to 32.
Due to the larger network size, we use a smaller learning rate of \(10^{-4}\).
Samples are shown in Figure \ref{fig:gan_images_mnist}.

The Wasserstein distance is a distance between probability distributions on a metric space. 
Intuitively, it is the minimum transportation cost needed to turn one distribution into another, that is, earth-mover's distance.
The \emph{empirical Wasserstein distance (EWD)} estimates the Wasserstein distance between the real data distribution and the data distribution produced by the generator.
It is obtained by sampling 1000 real data points and 1000 fake data points, computing the Euclidean distance matrix between them, solving the linear sum assignment problem, and returning the transportation cost. 
For the linear sum assignment problem, we use the implementation of the Python library scipy~\citep{2020SciPy-NMeth}, which uses a modified Jonker-Volgenant algorithm with no initialization~\citet{crouse2016implementing}.
Figure \ref{fig:gan} shows performances.
Our methods outperform the rest.

\begin{figure}
    \centering
    \includegraphics[width=.8\linewidth]{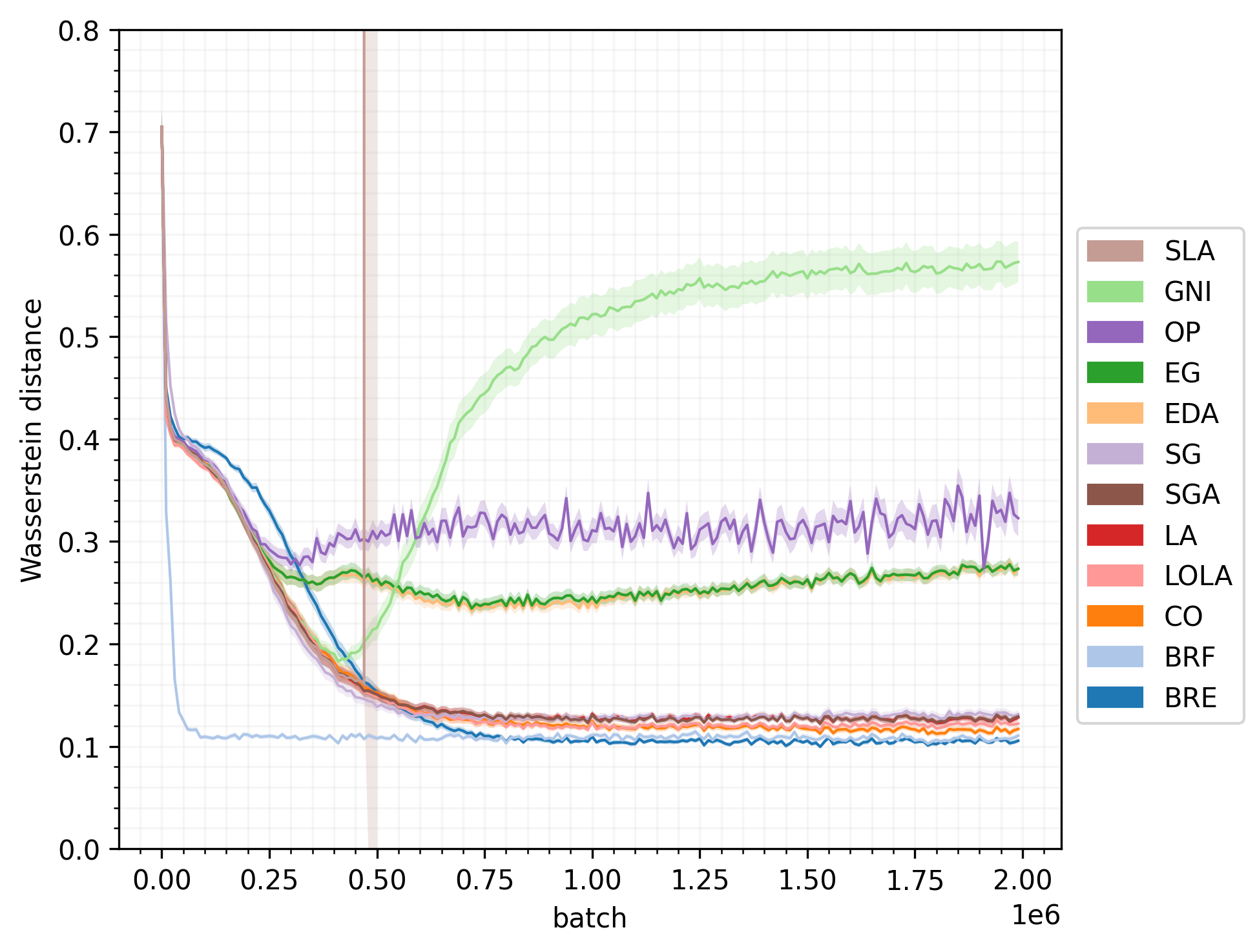}
    \includegraphics[width=.8\linewidth]{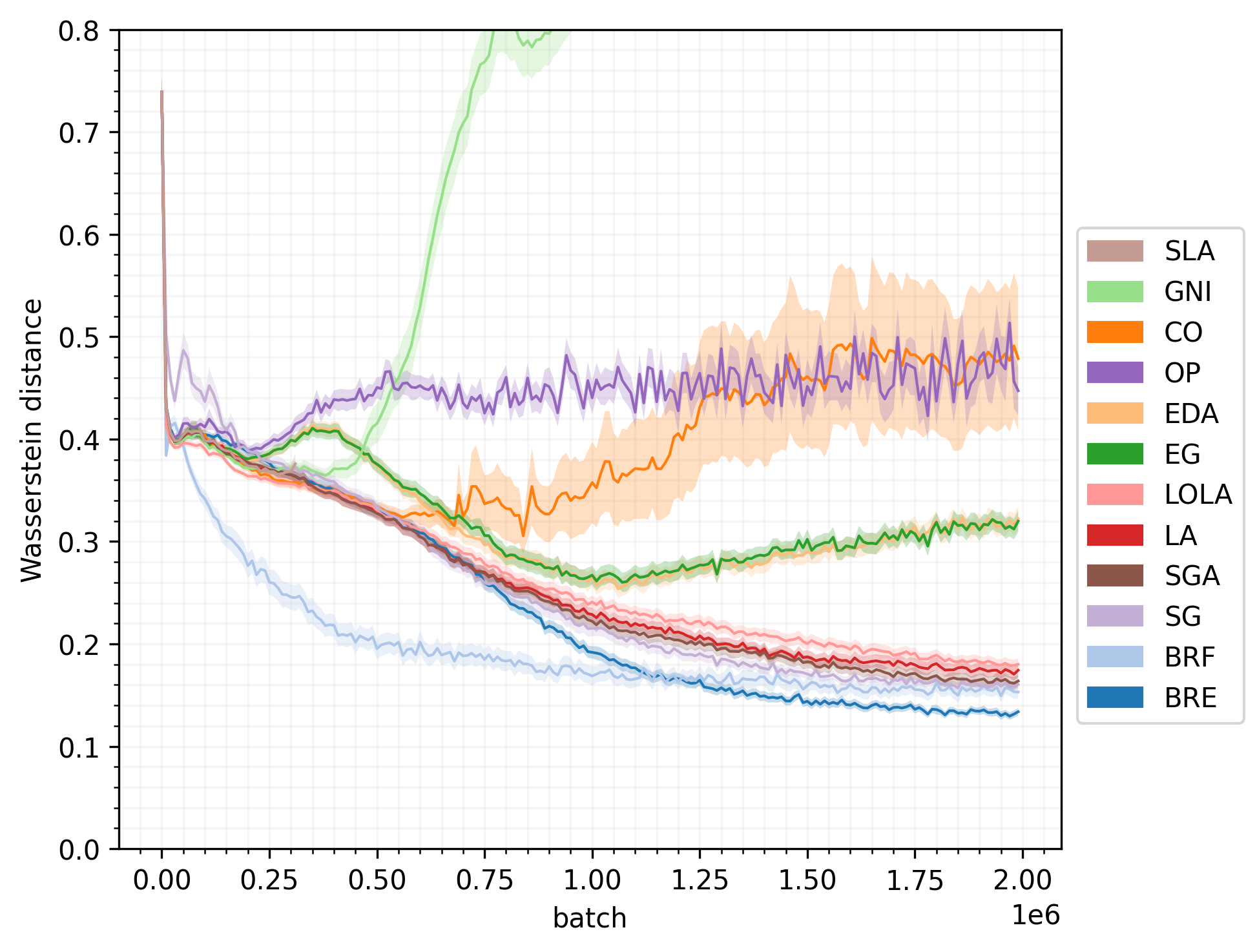}
    \includegraphics[width=.8\linewidth]{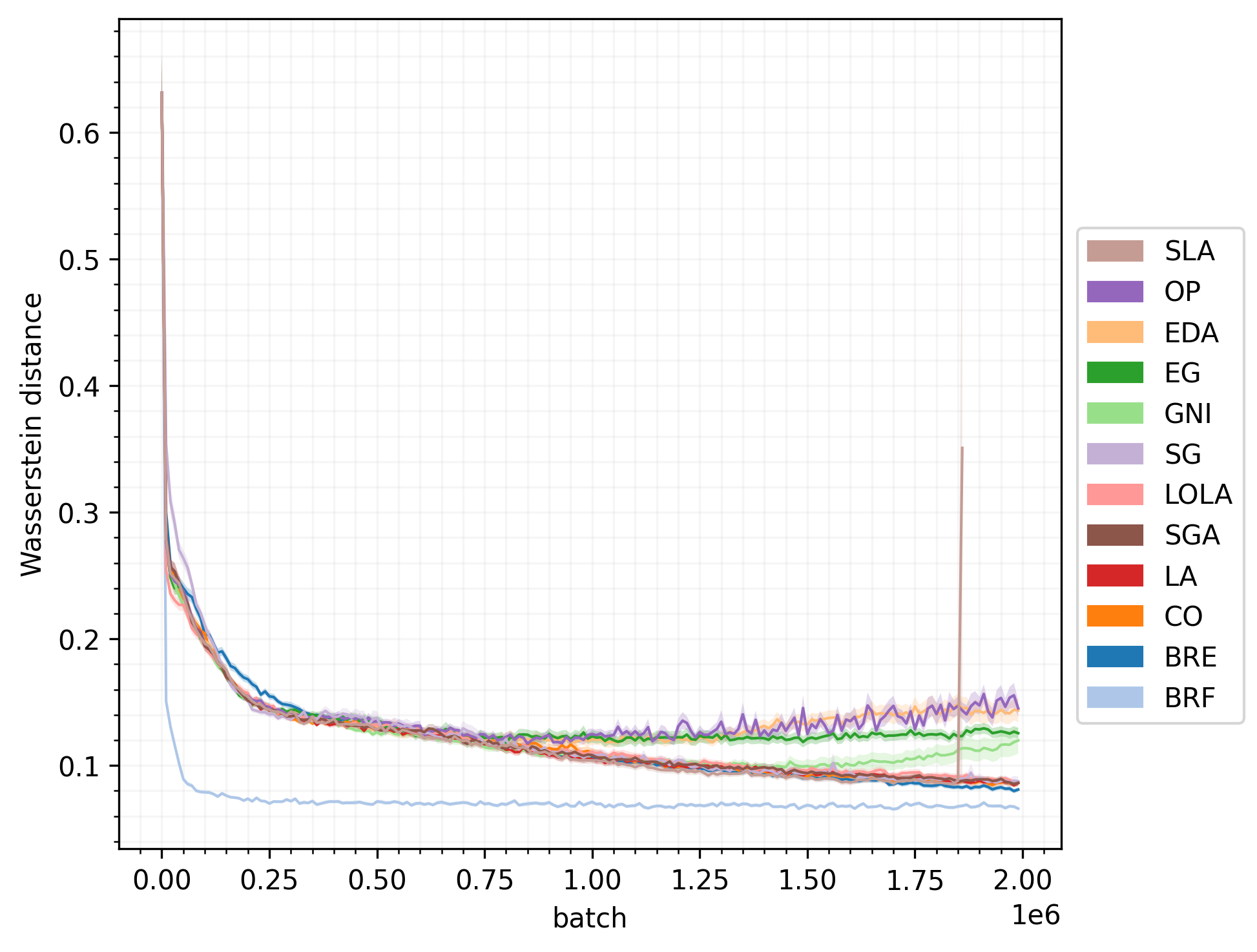}
    \includegraphics[width=.8\linewidth]{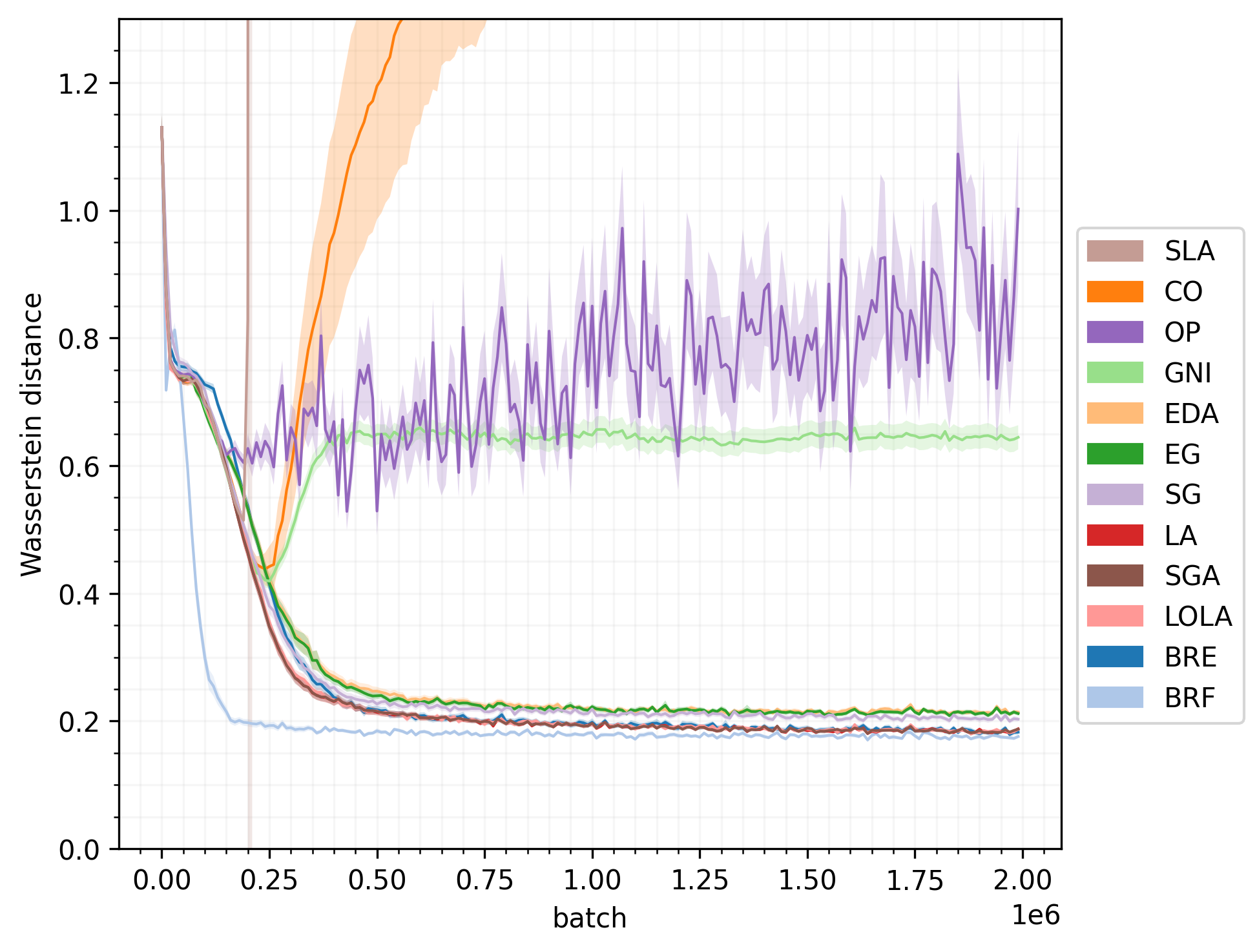}
    \caption{GAN with ring, grid, spiral, and cube datasets.}
    \label{fig:gan}
\end{figure}

\paragraph{Continuous security game}

Security games are used to model defender-adversary interactions in many domains,
including the protection of infrastructure like airports, ports, and flights \citep{Kamra_2018},
as well as the protection of wildlife, fisheries and forests \citep{kar2017trends,sinha2018stackelberg}.
Security games are often modeled with Stackelberg equilibrium as the solution concept, which coincides with NE in zero-sum security games and in certain structured general-sum games~\citep{korzhyk2011stackelberg}.
In practice, many of the aforementioned domains have continuous action spaces.
Such games have been studied by \citet{Kamra_2017,Kamra_2018,Kamra_2019}, among others.
As an example, we use the following game.
Let \(\mathcal{S} = [0, 1]^2\).
The attacker chooses a point \(x \in \mathcal{S}\).
Simultaneously, the defender chooses \(n\) points \(y_i \in \mathcal{S}\).
Let \(d = \inf_{i \in [n]} \|x - y_i\|\) be the distance between the attacker's point and the defender's closest point.
The defender receives a utility of \(\exp(-d^2)\), and the attacker receives \(-\exp(-d^2)\).
Thus the defender seeks to be close to the attacker, while the opposite is true for the attacker.
Figure \ref{fig:security} shows performances.
Our methods perform best.

\begin{figure}
    \centering
    \includegraphics[width=.8\linewidth]{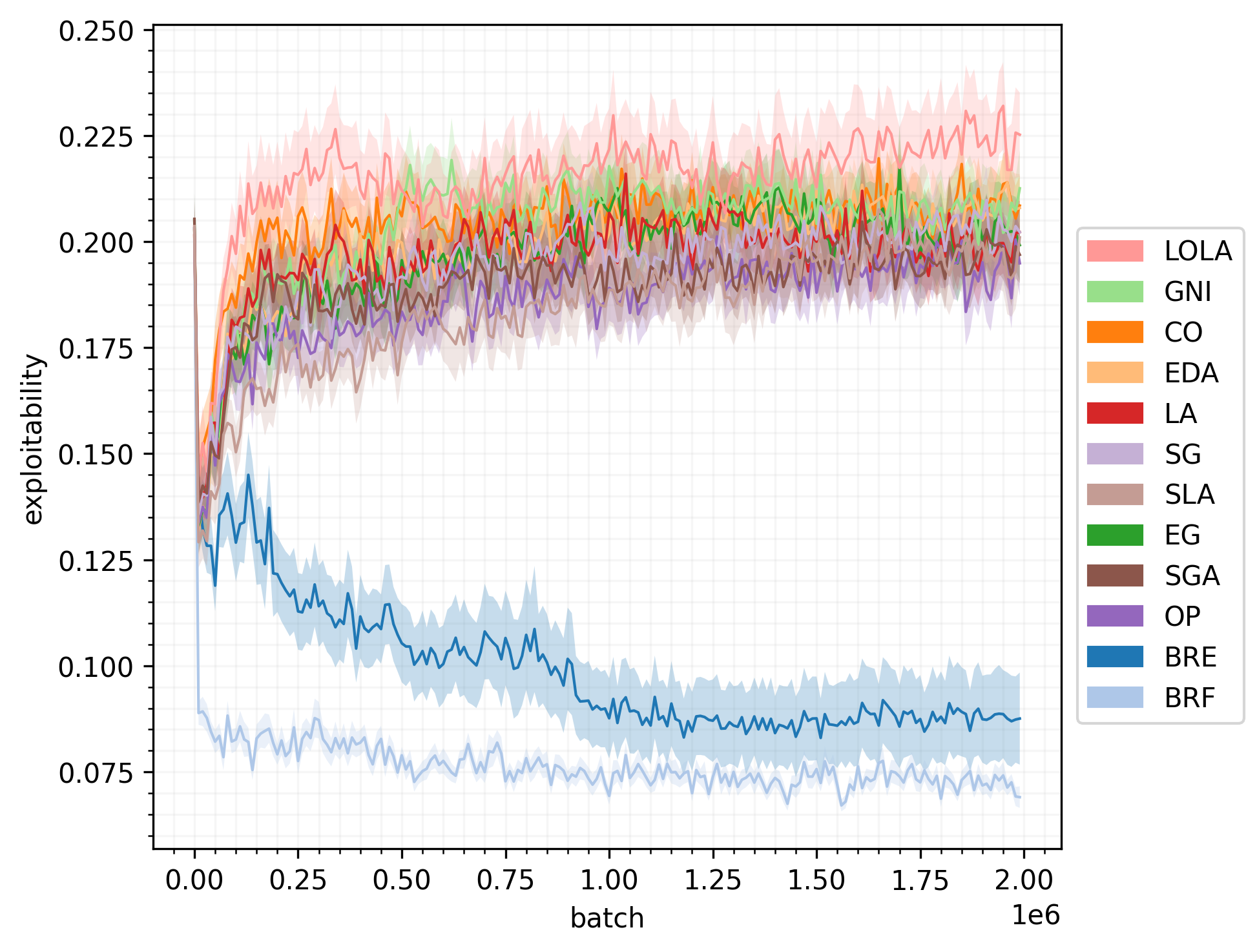}
    \includegraphics[width=.8\linewidth]{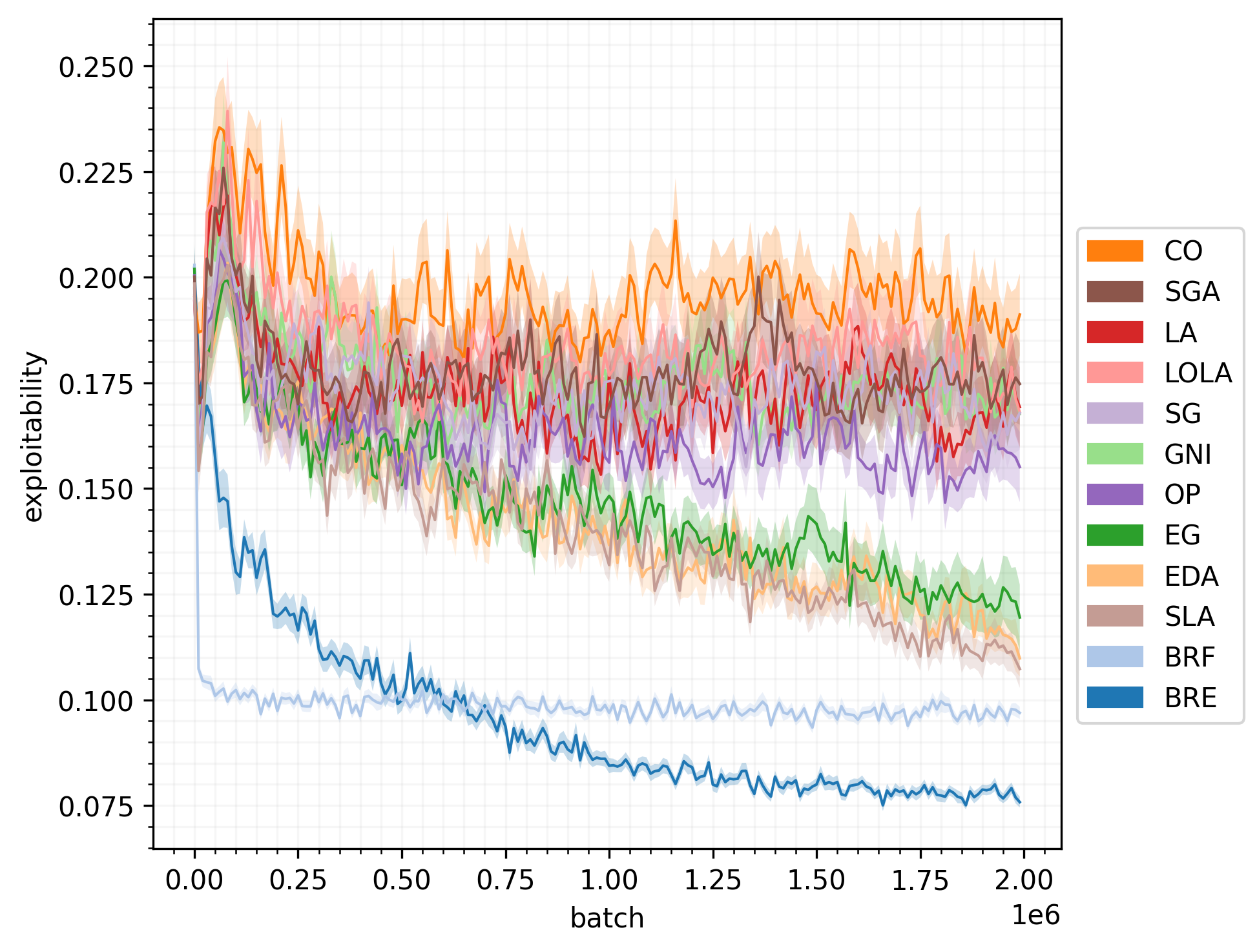}
    \caption{Security game with 1 and 2 points.}
    \label{fig:security}
\end{figure}

\paragraph{Glicksberg--Gross game}
This is a two-player zero-sum normal-form game with continuous action sets \(\mathcal{A}_i = [0, 1]\) and utility function \(u_1(x, y) = -u_2(x, y) = \frac{(1+x)(1+y)(1-xy)}{(1+xy)^2}\).
\citet{Glicksberg_1953} analyzed this game and proved that it has a unique mixed-strategy NE where each player's strategy has a cumulative distribution function of \(F(t) = \frac{4}{\pi} \arctan \sqrt{t}\).
To model mixed strategies, we use the following implicit density model.
We feed a sample from a 1-dimensional standard normal distribution into a fully-connected network with one hidden layer of size 32 and output layer of size 1.
The output is squeezed to the unit interval using the logistic sigmoid function.
Figure \ref{fig:glicksberg_gross} shows performances.
Our methods converge fastest.

\begin{figure}
    \centering
    \includegraphics[width=.8\linewidth]{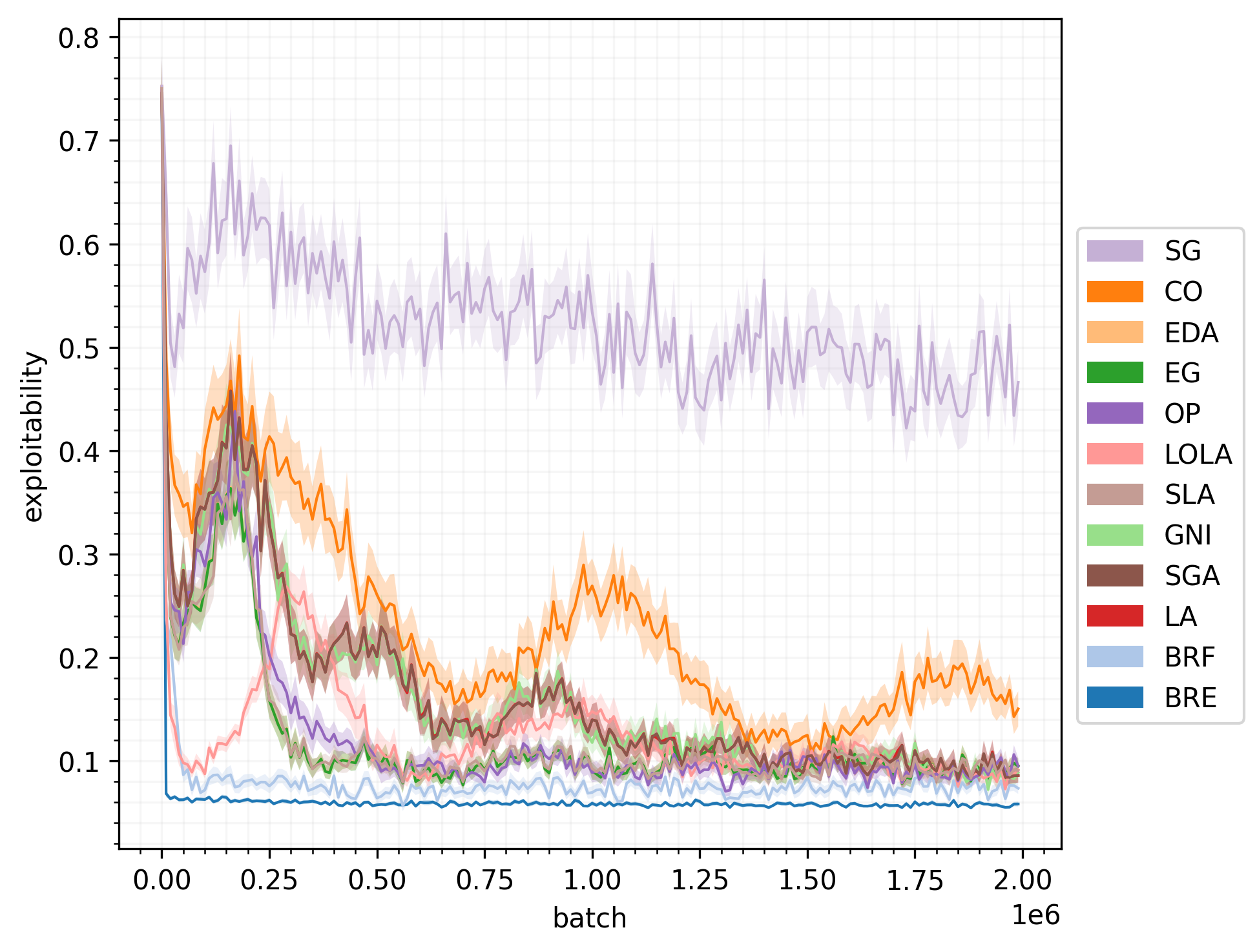}
    \caption{Glicksberg--Gross game.}
    \label{fig:glicksberg_gross}
\end{figure}

\paragraph{Shapley game}
This is a two-player normal-form game with 3 actions per player.
Its utility matrices are presented in the appendix.
It was introduced by~\citet[p. 26]{Shapley64:Some}, and is a classic example of a game for which fictitious play \citep{Brown51:Iterative,berger_2007} diverges.
(Instead, fictitious play cycles through the cells with 1's in them, with ever-increasing lengths of play in each of these cells.)
Figure \ref{fig:shapley} shows performances.
Our methods converge, while the rest diverge.

\begin{figure}
    \centering
    \includegraphics[width=.8\linewidth]{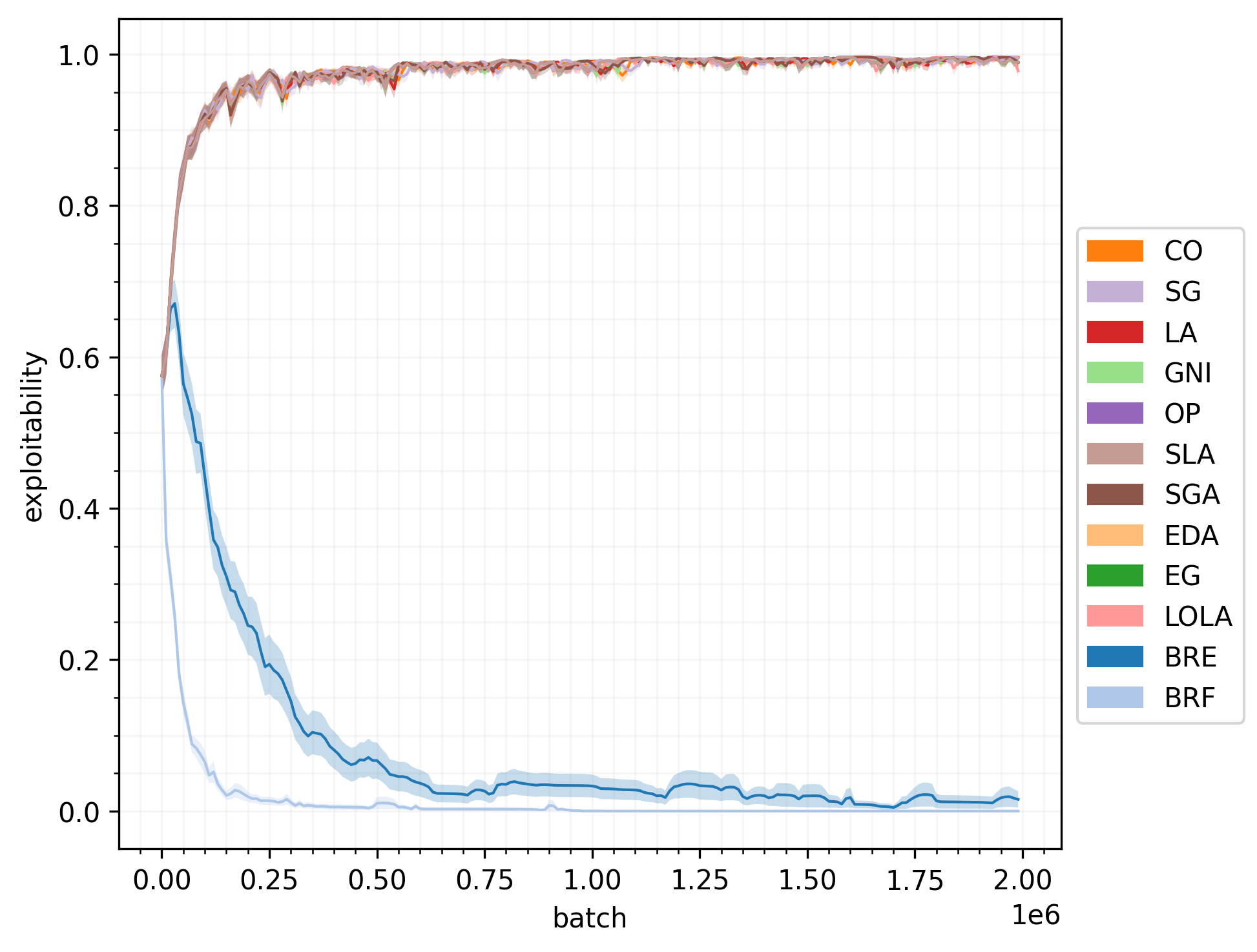}
    \caption{Shapley game.}
    \label{fig:shapley}
\end{figure}

\paragraph{Poker games}
Kuhn poker is a variant of poker introduced by \citet{Kuhn50:Simplified}.
It is a two-player zero-sum imperfect-information game.
A 3-player variant was introduced by \citet{Szafron_2013}, and was one of the largest three-player games to be solved analytically to date.
2-player Kuhn poker has a 12-dimensional strategy space per player (24 in total).
3-player Kuhn poker has a 32-dimensional strategy space per player (96 in total).
Thus the utility function for these games is high-dimensional and nonlinear, making them a good benchmark.
Figure \ref{fig:kuhn} shows performances.
Our methods converge fastest.

\begin{figure}
    \centering
    \includegraphics[width=.8\linewidth]{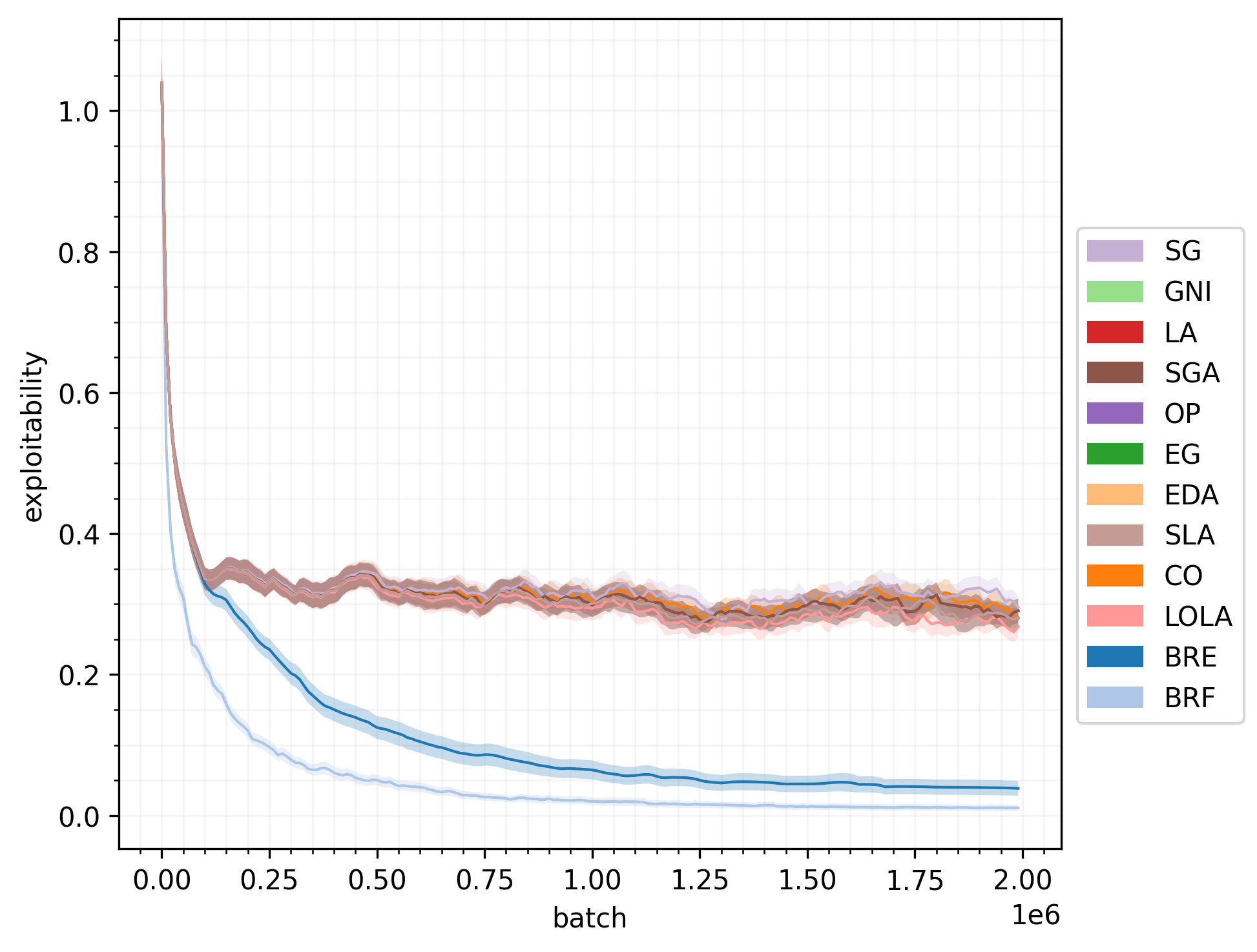}
    \includegraphics[width=.8\linewidth]{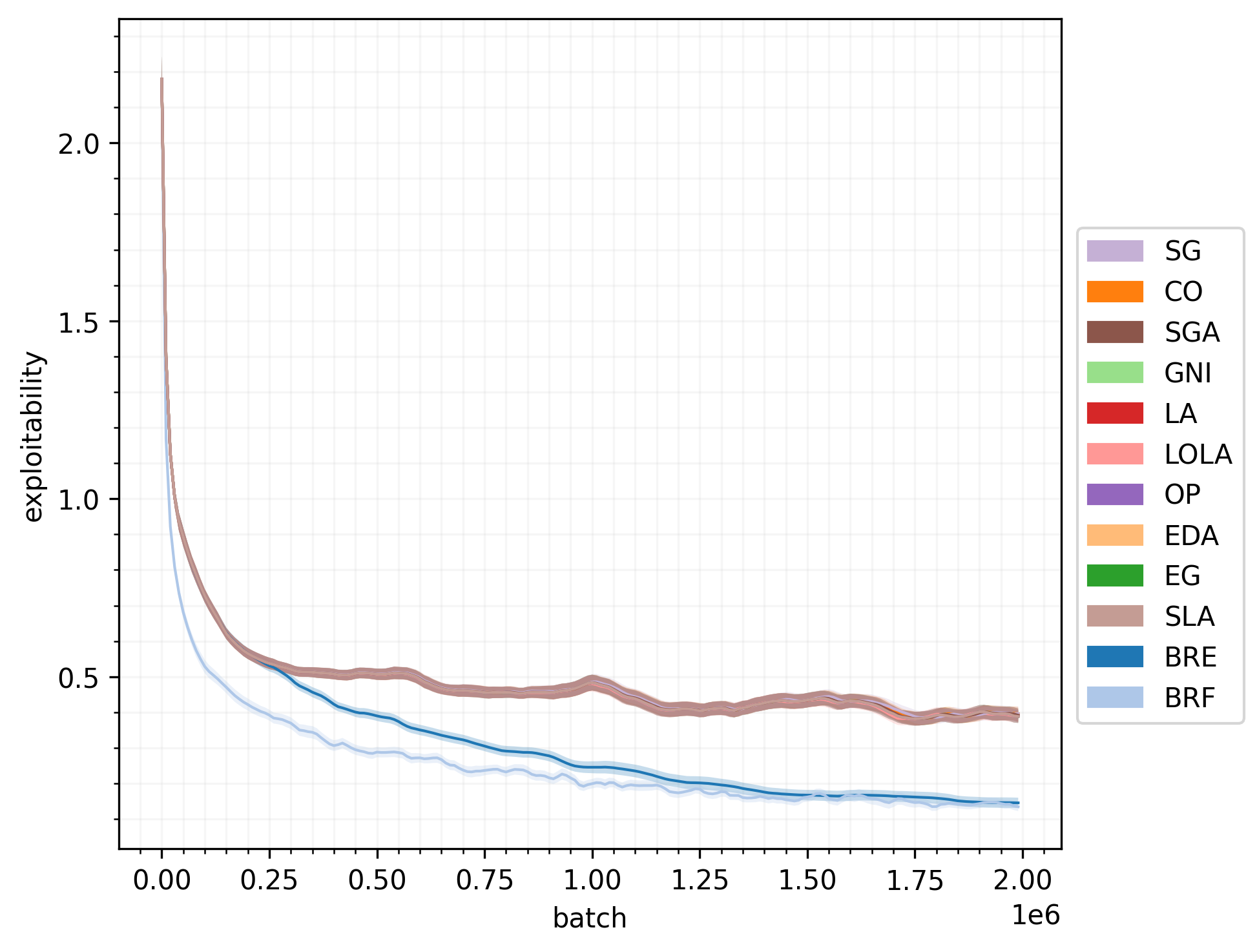}
    \caption{Kuhn poker with 2 and 3 players.}
    \label{fig:kuhn}
\end{figure}

\paragraph{Generalized rock paper scissors}
\emph{Rock paper scissors (RPS)} is a classic two-player zero-sum normal-form game with 3 actions per player.
It has a unique mixed-strategy NE where each player mixes uniformly over its actions.
\citet[p. 7]{cloud2022anticipatory} generalize RPS to \(n\) actions by letting \(u_1(a) = -u_2(a) = \llbracket a_2 - a_1 = 1 \mod n \rrbracket - \llbracket a_1 - a_2 = 1 \mod n \rrbracket\).
Figure \ref{fig:rps} shows performances.
Our methods converge fastest.

\begin{figure}
    \centering
    \includegraphics[width=.8\linewidth]{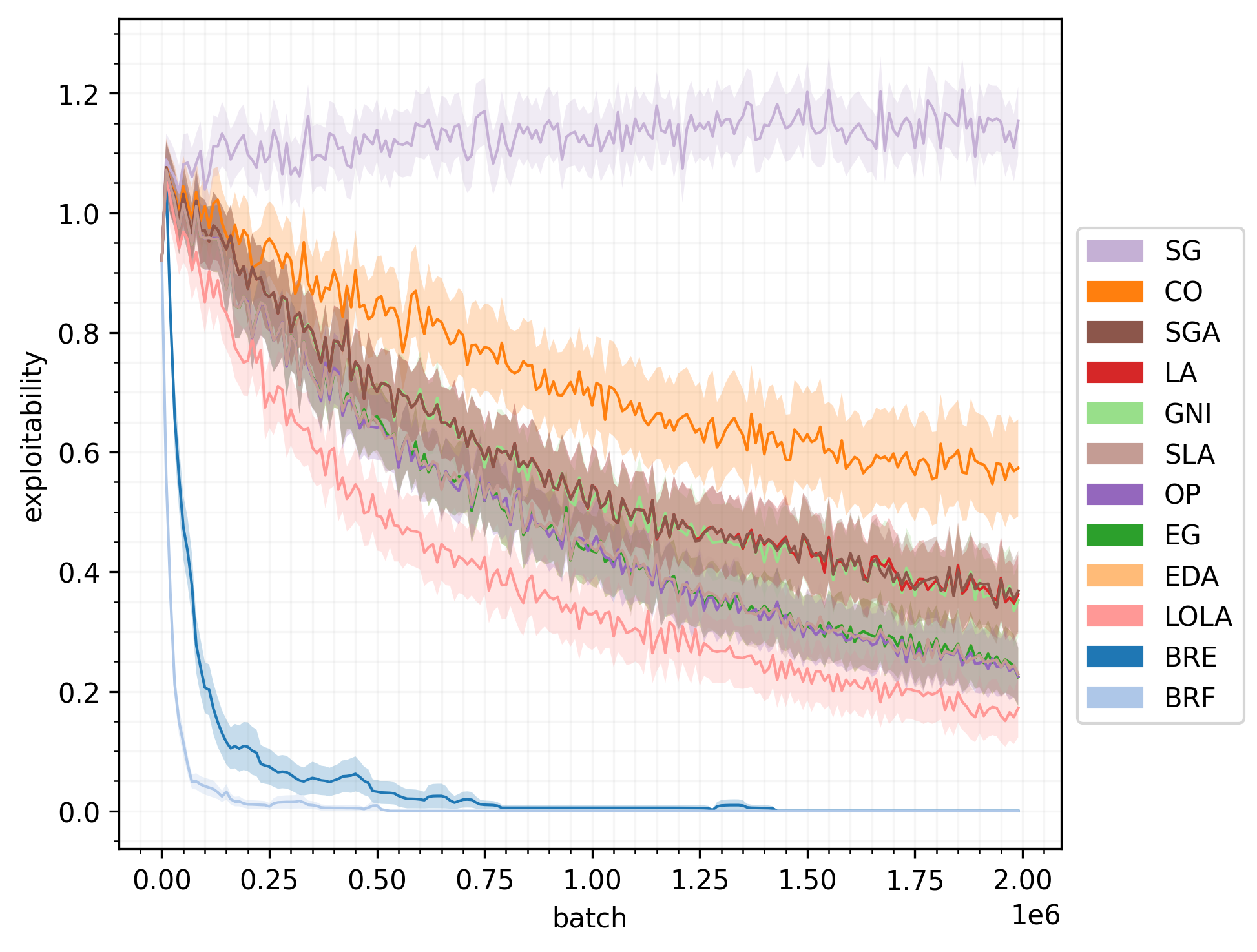}
    \includegraphics[width=.8\linewidth]{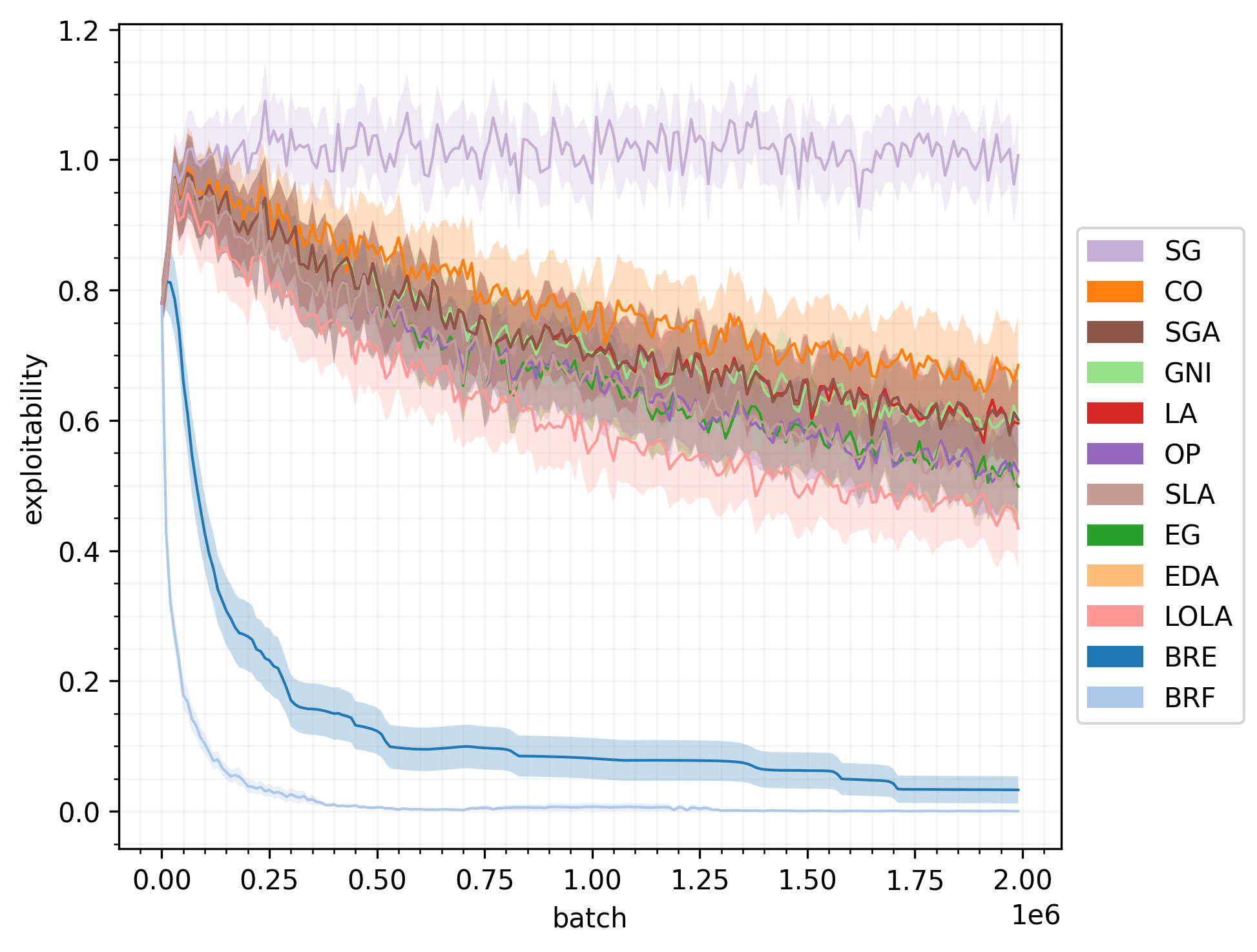}
    \caption{Rock paper scissors with 3 and 4 actions.}
    \label{fig:rps}
\end{figure}

\paragraph{Generalized matching pennies}
\textit{Matching pennies (MP)} is a classic two-player zero-sum normal-form game with 2 actions per player.
It can be generalized to \(n\) players~\citep{Jordan_1993,Leslie_2003} by letting \(u_i : [2]^n \to \mathbb{R}\) where \(u_i(a) = (2 \llbracket a_i = a_{i+1 \bmod n} \rrbracket - 1) (-1)^{\llbracket i = n - 1 \rrbracket}\).
That is, each player seeks to match the next, but the last player seeks to \emph{un}match the first.
Like the 2-player version, it has a unique mixed-strategy NE where each player mixes uniformly over its actions.
The 3-player game's NE is locally unstable in a strong sense \citep{Jordan_1993}.
More precisely, discrete-time fictitious play \citep{Brown51:Iterative} fails to converge, and instead enters a limit cycle asymptotically.
Figure \ref{fig:mp} shows performances.
In the 2-player game, our methods converge fastest.
In the 3-player game, our methods converge, while the rest diverge.

\begin{figure}
    \centering
    \includegraphics[width=.8\linewidth]{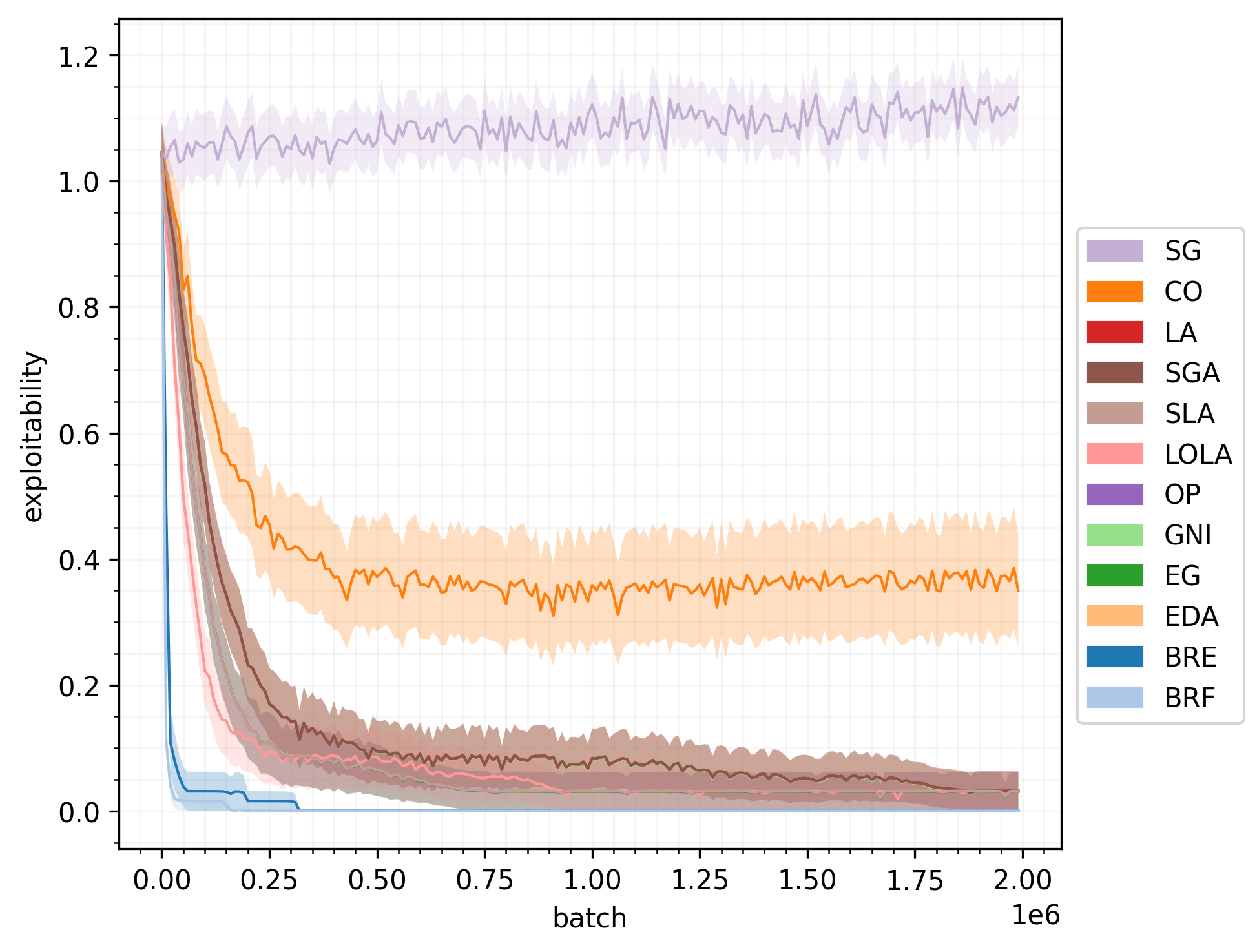}
    \includegraphics[width=.8\linewidth]{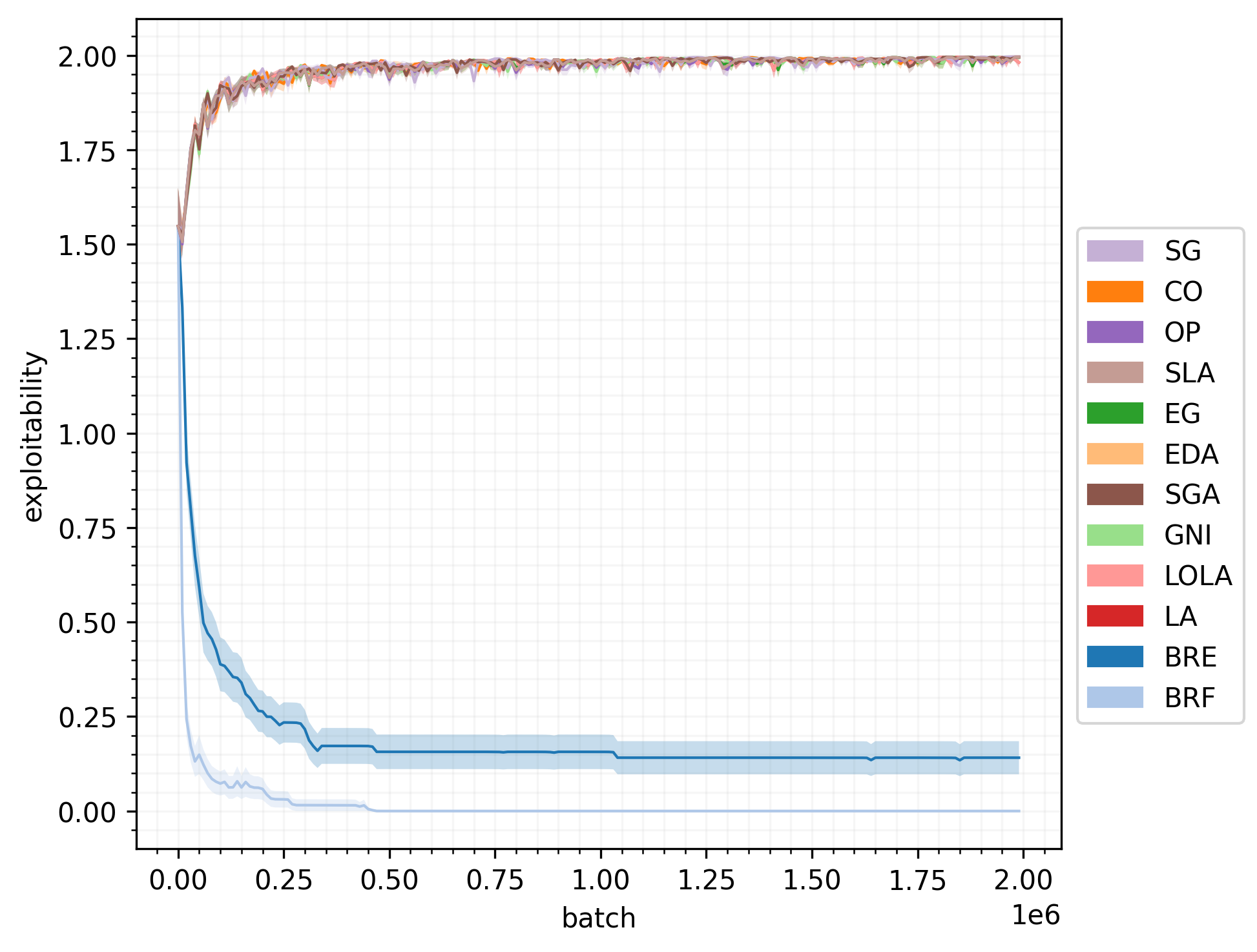}
    \caption{Matching pennies with 2 and 3 players.}
    \label{fig:mp}
\end{figure}

\section{Conclusion}
\label{sec:conclusion}

In this paper, we studied the problem of finding an approximate NE of continuous games.
The main measure of closeness to NE is \emph{exploitability}, which measures how much players can benefit from unilaterally changing their strategy.
We proposed two new methods that minimize an approximation of the exploitability with respect to the strategy profile.
These methods use learned best-response functions and best-response ensembles, respectively.
We evaluated these methods in various continuous games, showing that they outperform prior methods.
Prior equilibrium-finding techniques usually suffer from cycling or divergent behavior. 
By thinking about the equilibrium-finding problem we are trying to solve from first principles, and by formulating a novel solution to it, our paper opens up new possibilities for tackling games where such problematic behavior appears, and significantly reduces the amount of time and resources needed to obtain good approximate equilibria.
In \S\ref{sec:future} of the appendix, we discuss possible extensions of our methods and directions for future research.

\section{Acknowledgements}

This material is based on work supported by the Vannevar Bush Faculty Fellowship ONR N00014-23-1-2876; National Science Foundation grants CCF-1733556, RI-2312342, and RI-1901403; ARO award W911NF2210266; and NIH award A240108S001.

\bibliographystyle{plainnat}
\bibliography{dairefs,references}
\appendix
\section{Additional figures}

In this section, we include additional figures that did not fit in the body of the paper.
Utility tables are shown in Tables \ref{tab:utilities_mp}, \ref{tab:utilities_rps_3a}, \ref{tab:utilities_rps_4a}, and \ref{tab:utilities_shapley}.
Learned distributions under GAN training are illustrated in Figures \ref{fig:gan_images_ring}, \ref{fig:gan_images_grid}, \ref{fig:gan_images_spiral}, \ref{fig:gan_images_cube}, and \ref{fig:gan_images_mnist}.

\begin{table}
    \centering
    \begin{tabular}{|c|cc|}
        \hline
        & H & T \\
        \hline
        H & \(+1\) & \(-1\) \\
        T & \(-1\) & \(+1\) \\
        \hline
    \end{tabular}
    \caption{Utilities for matching pennies (2 players), from the perspective of player 1.}
    \label{tab:utilities_mp}
\end{table}

\begin{table}
    \centering
    \begin{tabular}{|c|ccc|}
        \hline
        & R & P & S \\
        \hline
        R & \(0\) & \(-1\) & \(+1\) \\
        P & \(+1\) & \(0\) & \(-1\) \\
        S & \(-1\) & \(+1\) & \(0\) \\
        \hline
    \end{tabular}
    \caption{Utilities for rock paper scissors (3 actions), from the perspective of player 1.}
    \label{tab:utilities_rps_3a}
\end{table}

\begin{table}
    \centering
    \begin{tabular}{|c|cccc|}
        \hline
        & A & B & C & D \\
        \hline
        A & \(0\) & \(-1\) & \(0\) & \(+1\) \\
        B & \(+1\) & \(0\) & \(-1\) & \(0\) \\
        C & \(0\) & \(+1\) & \(0\) & \(-1\) \\
        D & \(-1\) & \(0\) & \(+1\) & \(0\) \\
        \hline
    \end{tabular}
    \caption{Utilities for rock paper scissors (4 actions), from the perspective of player 1.}
    \label{tab:utilities_rps_4a}
\end{table}

\begin{table}
    \centering
    \begin{tabular}{|c|ccc|}
        \hline
        & A & B & C \\
        \hline
        A & \(1\) & \(0\) & \(0\) \\
        B & \(0\) & \(1\) & \(0\) \\
        C & \(0\) & \(0\) & \(1\) \\
        \hline
    \end{tabular}
    \hspace{1em}
    \begin{tabular}{|c|ccc|}
        \hline
        & A & B & C \\
        \hline
        A & \(0\) & \(1\) & \(0\) \\
        B & \(0\) & \(0\) & \(1\) \\
        C & \(1\) & \(0\) & \(0\) \\
        \hline
    \end{tabular}
    \caption{Utilities for Shapley game (players 1 and 2).}
    \label{tab:utilities_shapley}
\end{table}

\begin{figure}
    \centering
    \includegraphics[width=.4\linewidth]{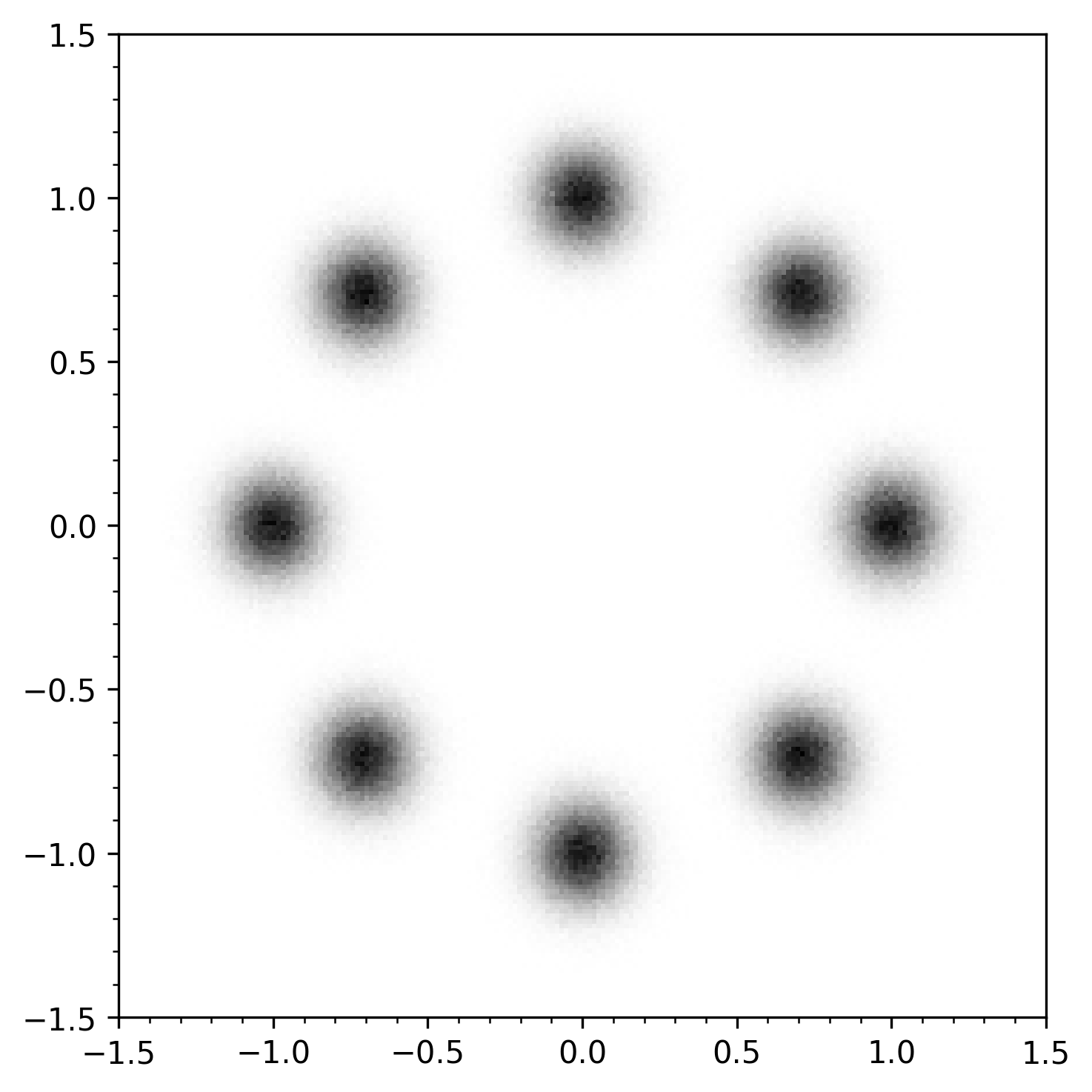}%
    \includegraphics[width=.4\linewidth]{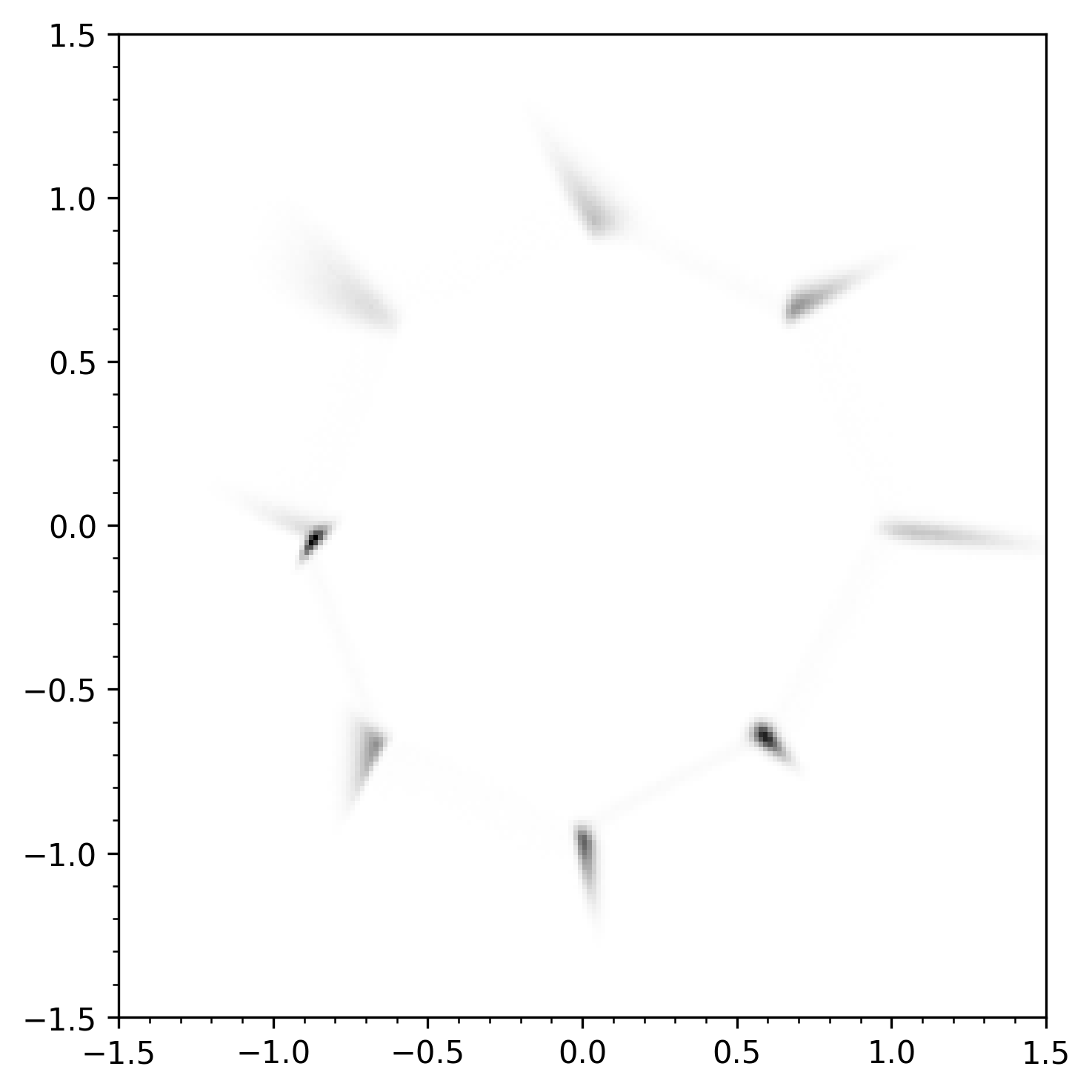}
    \includegraphics[width=.4\linewidth]{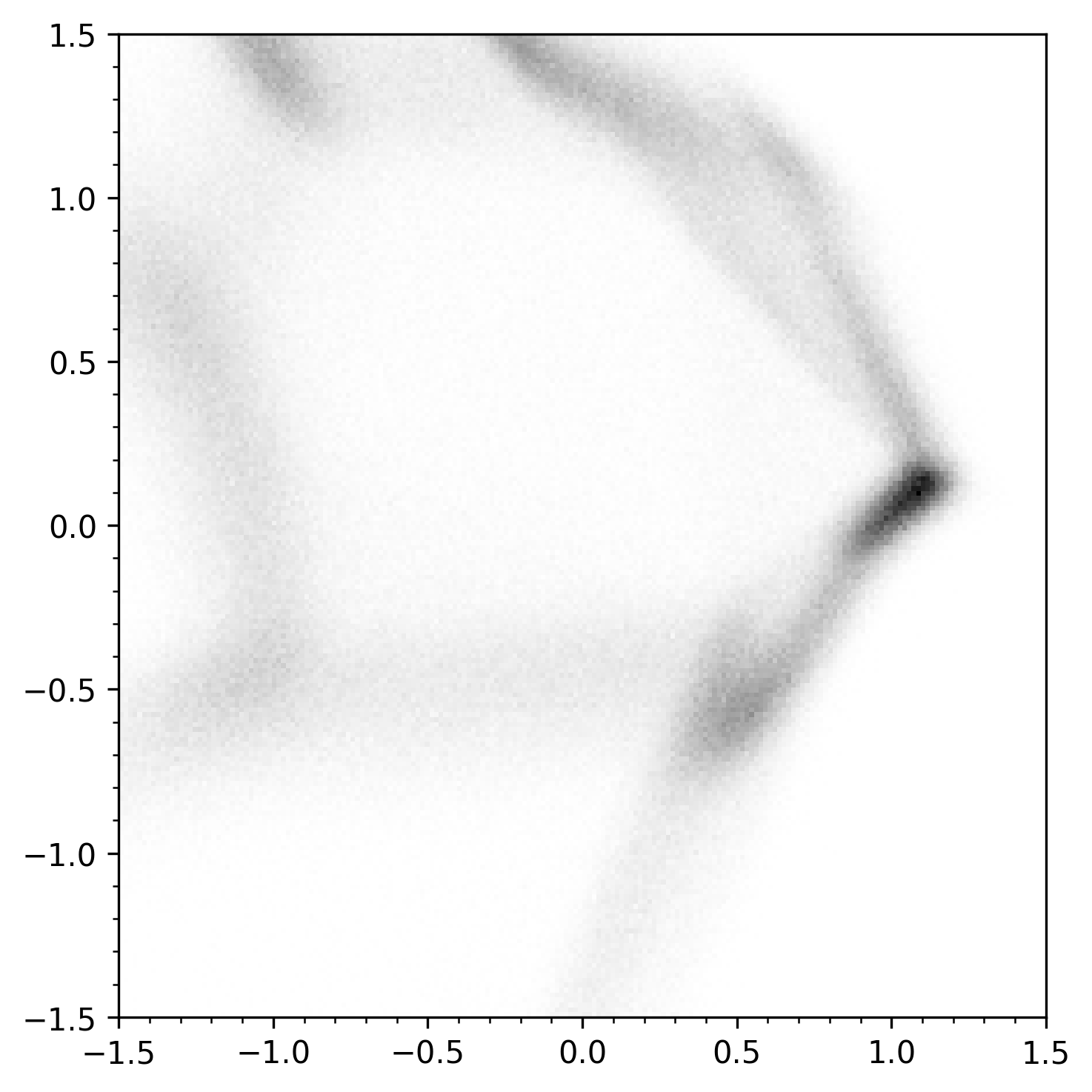}%
    \includegraphics[width=.4\linewidth]{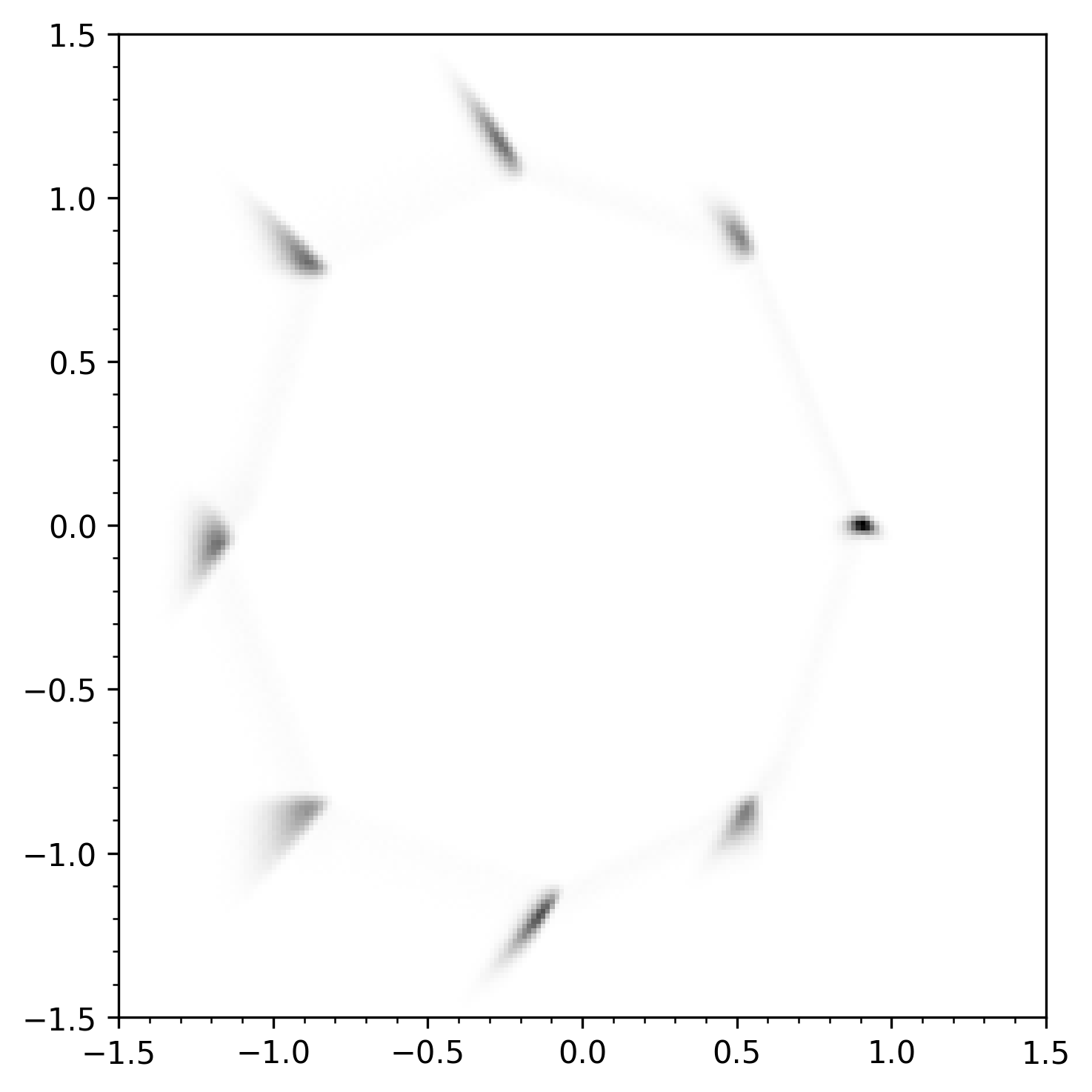}
    \includegraphics[width=.4\linewidth]{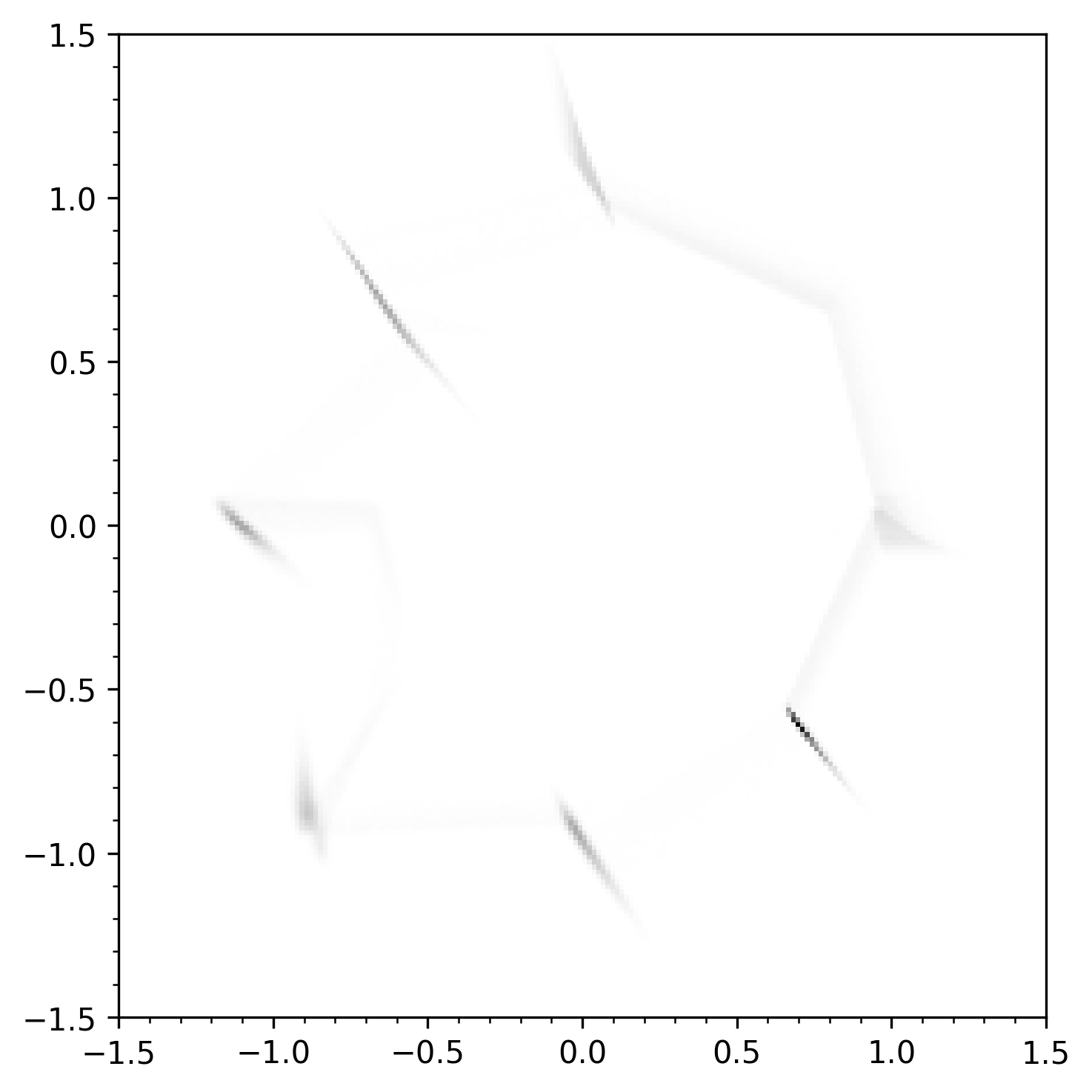}%
    \includegraphics[width=.4\linewidth]{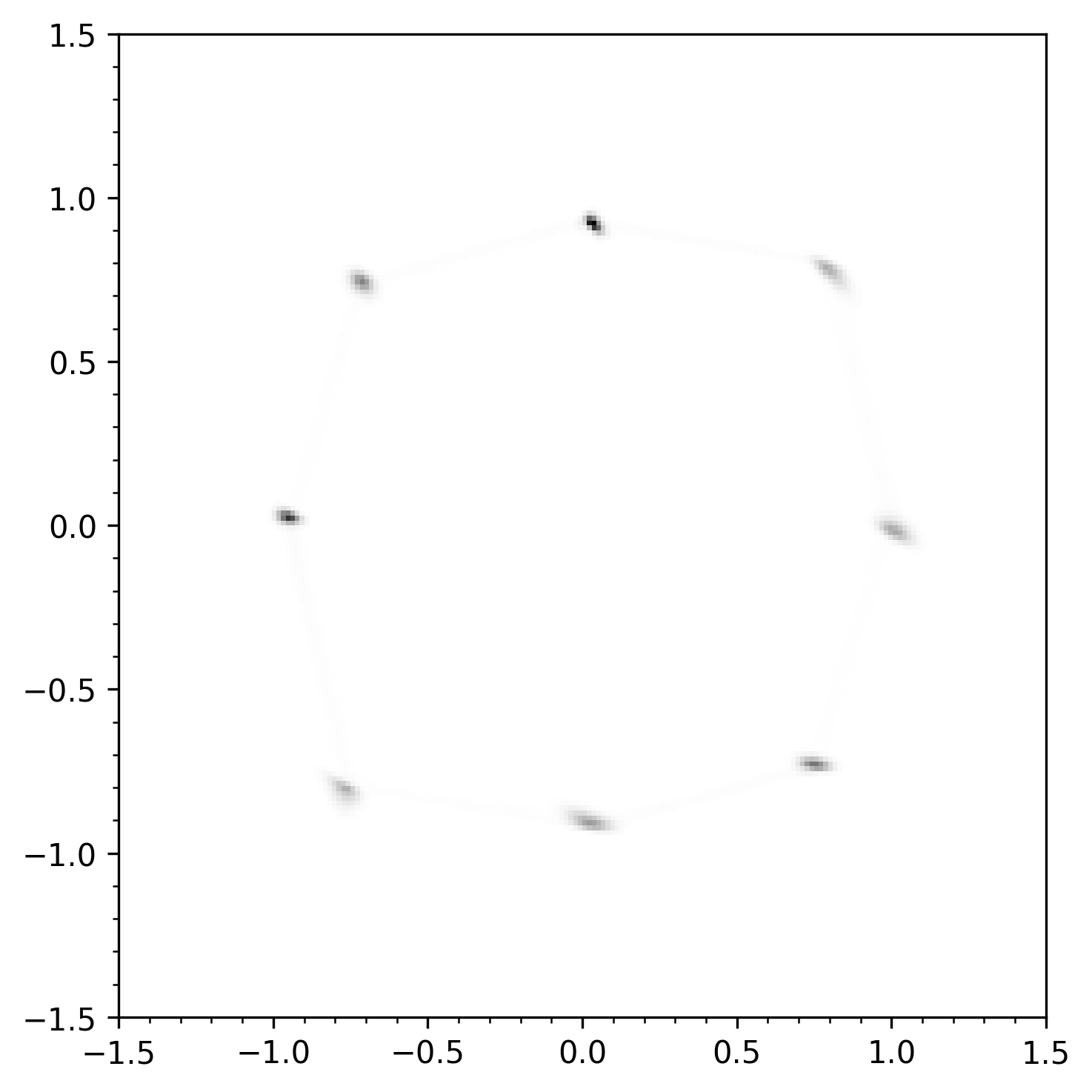}
    \includegraphics[width=.4\linewidth]{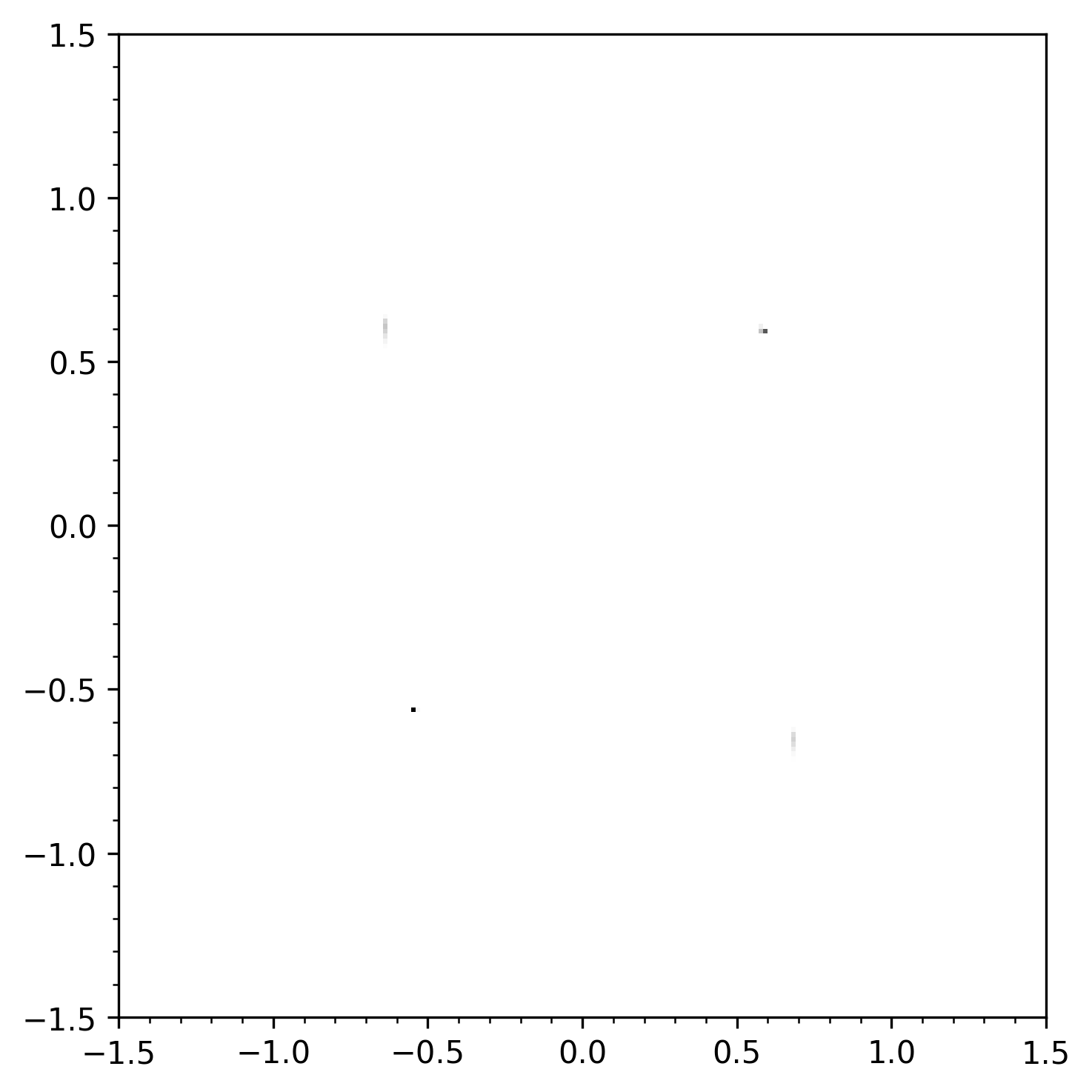}%
    \includegraphics[width=.4\linewidth]{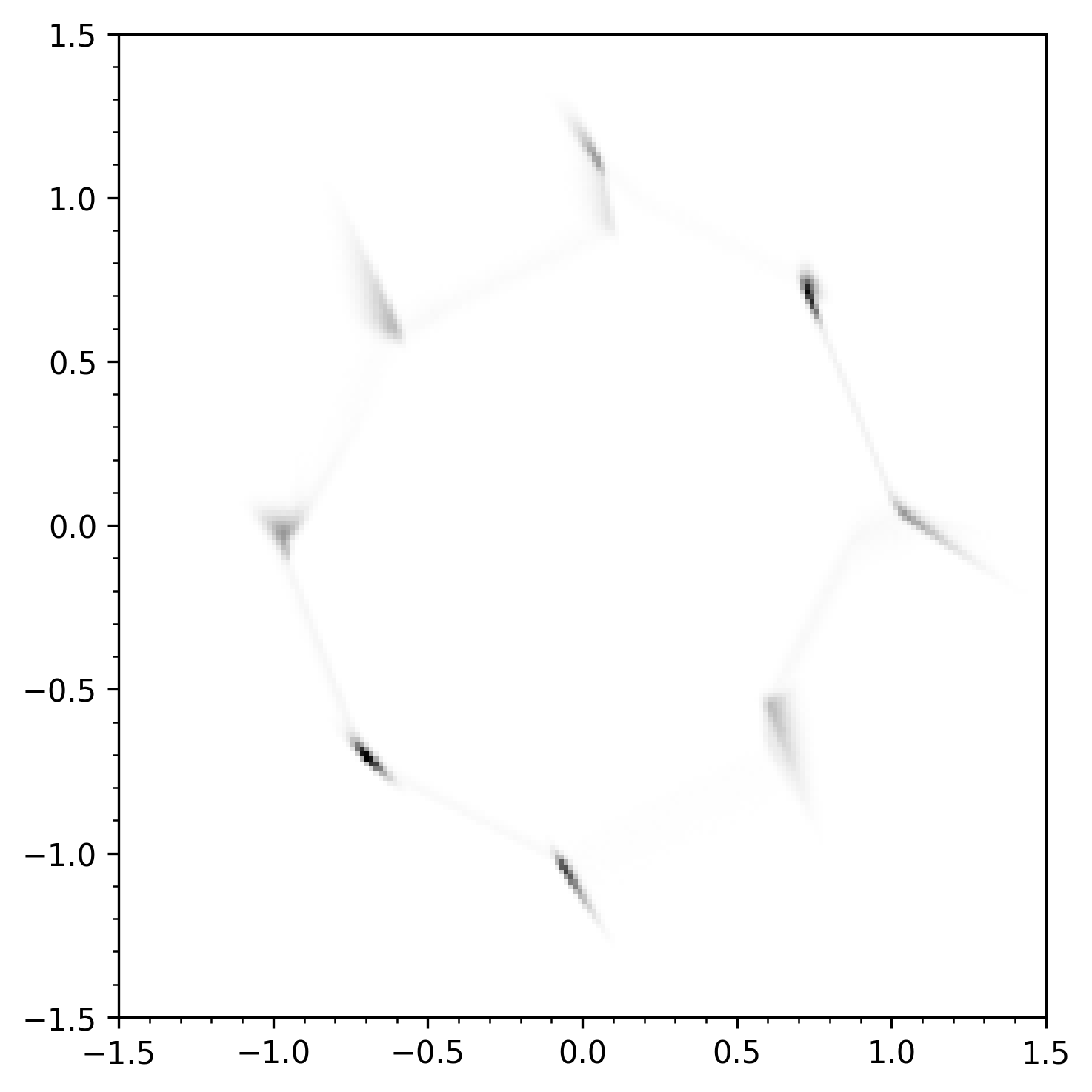}
    \includegraphics[width=.4\linewidth]{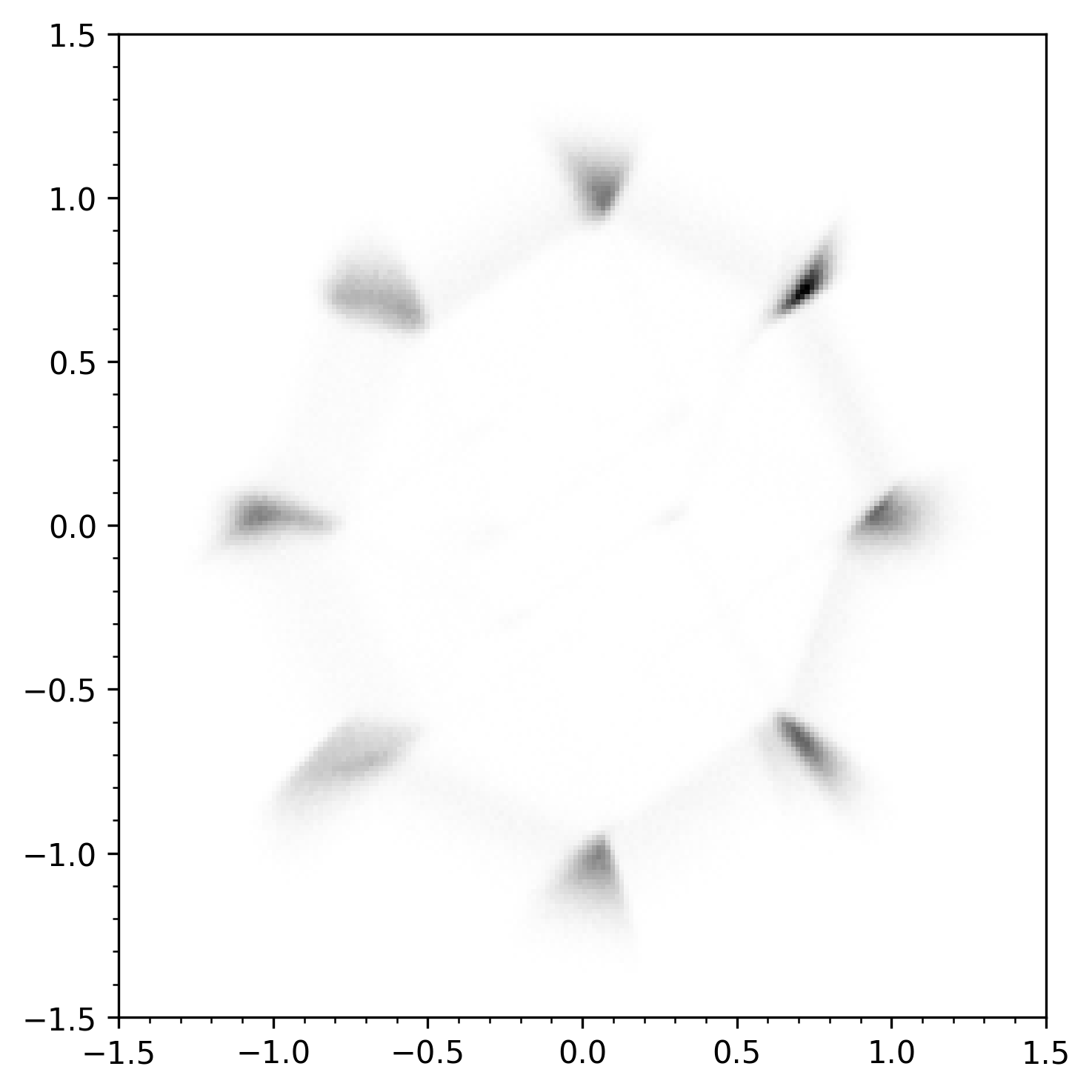}%
    \includegraphics[width=.4\linewidth]{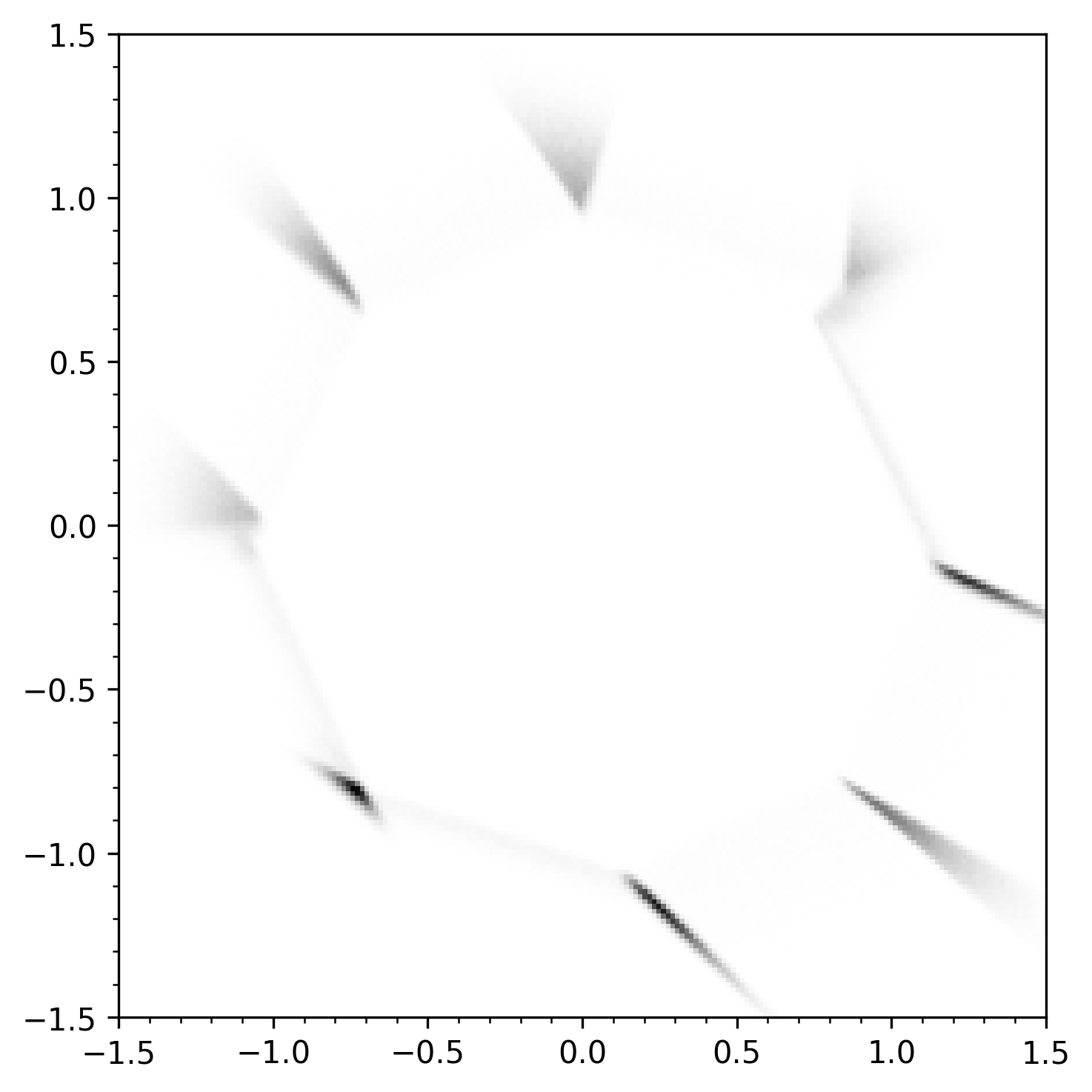}
    \includegraphics[width=.4\linewidth]{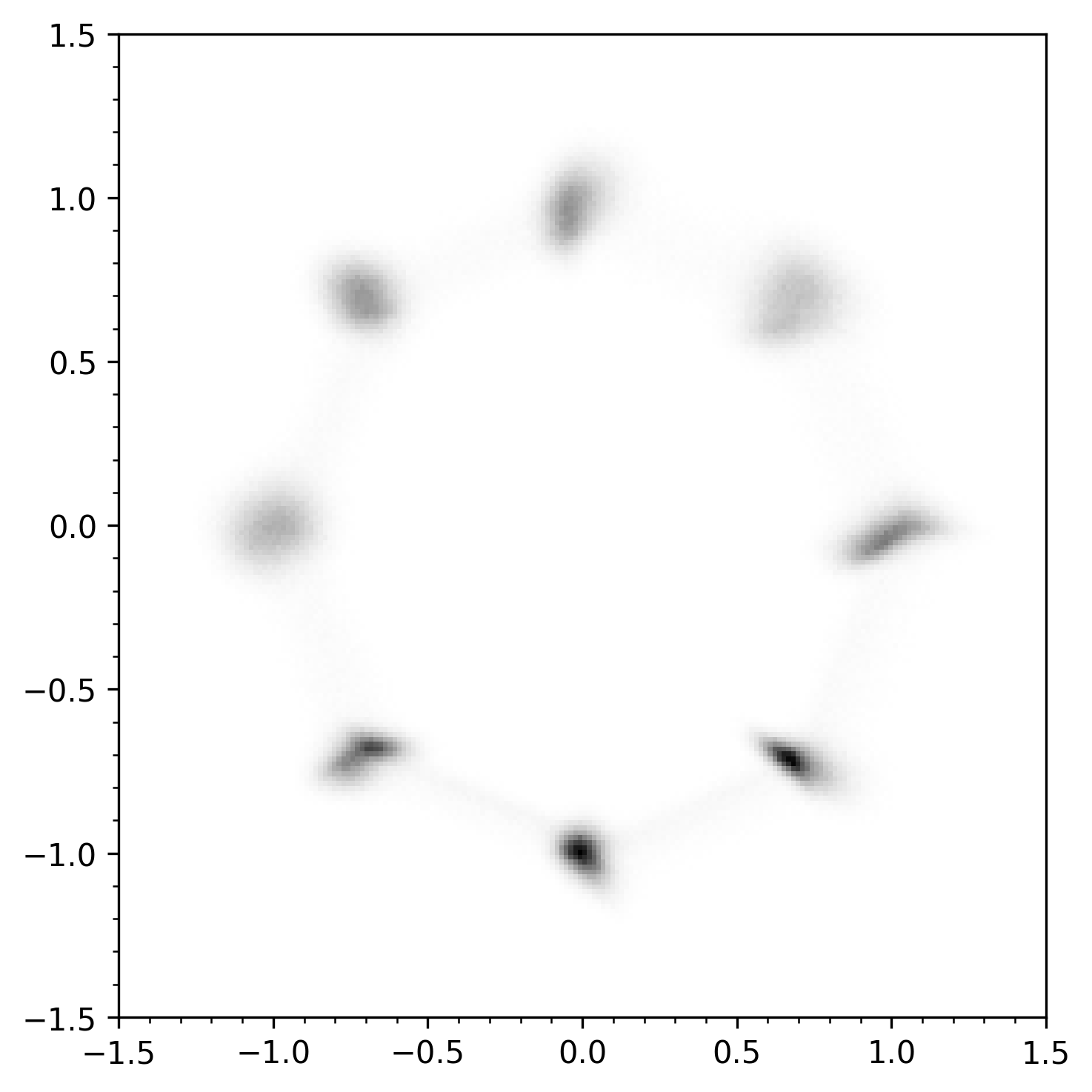}%
    \includegraphics[width=.4\linewidth]{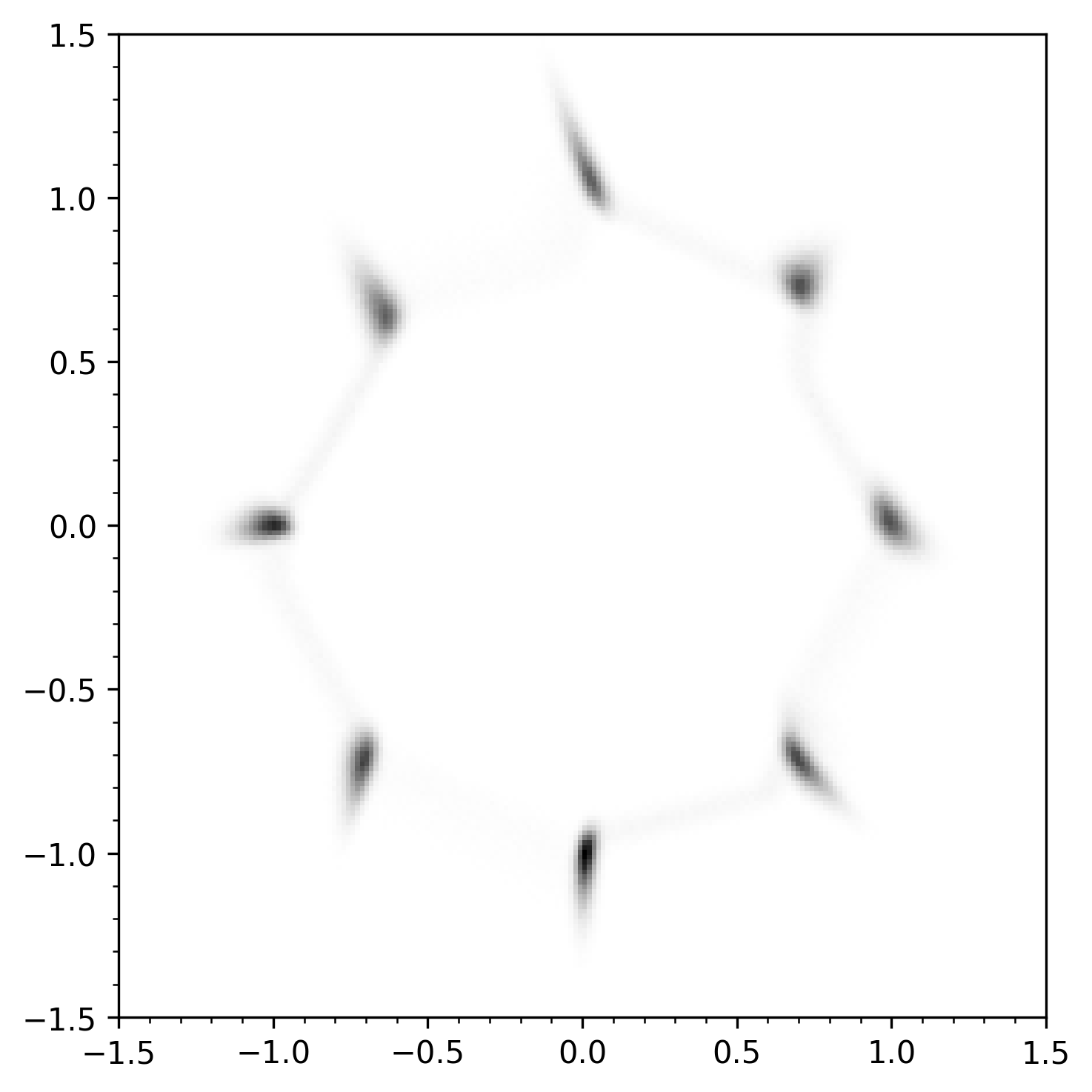}
    \caption{GAN (ring dataset).
    Left to right, top to bottom: Ground truth, SG, OP, EG, CO, SGA, GNI, LA, LOLA, EDA, BRF, BRE. SLA excluded due to divergence.}
    \label{fig:gan_images_ring}
\end{figure}

\begin{figure}
    \centering
    \includegraphics[width=.4\linewidth]{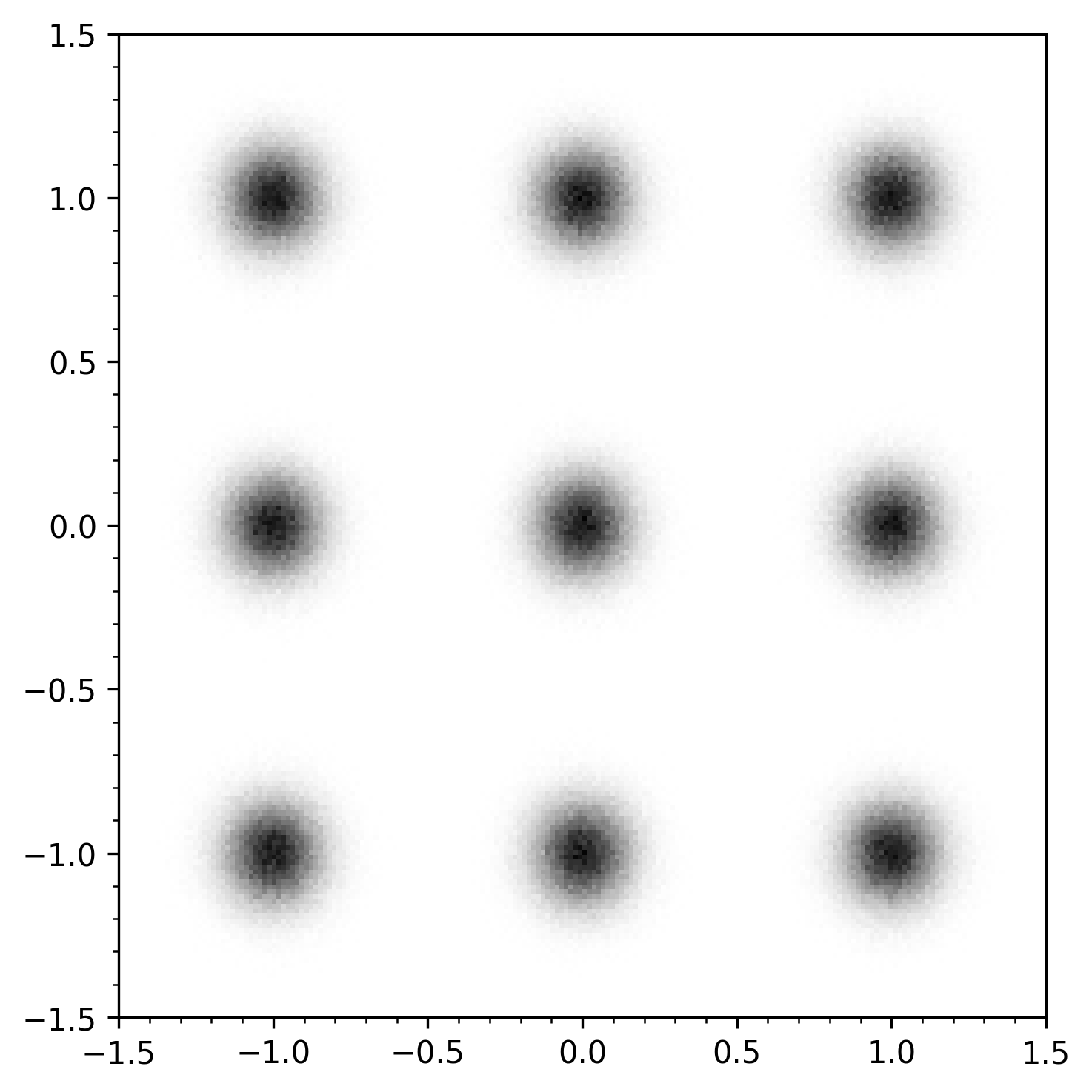}%
    \includegraphics[width=.4\linewidth]{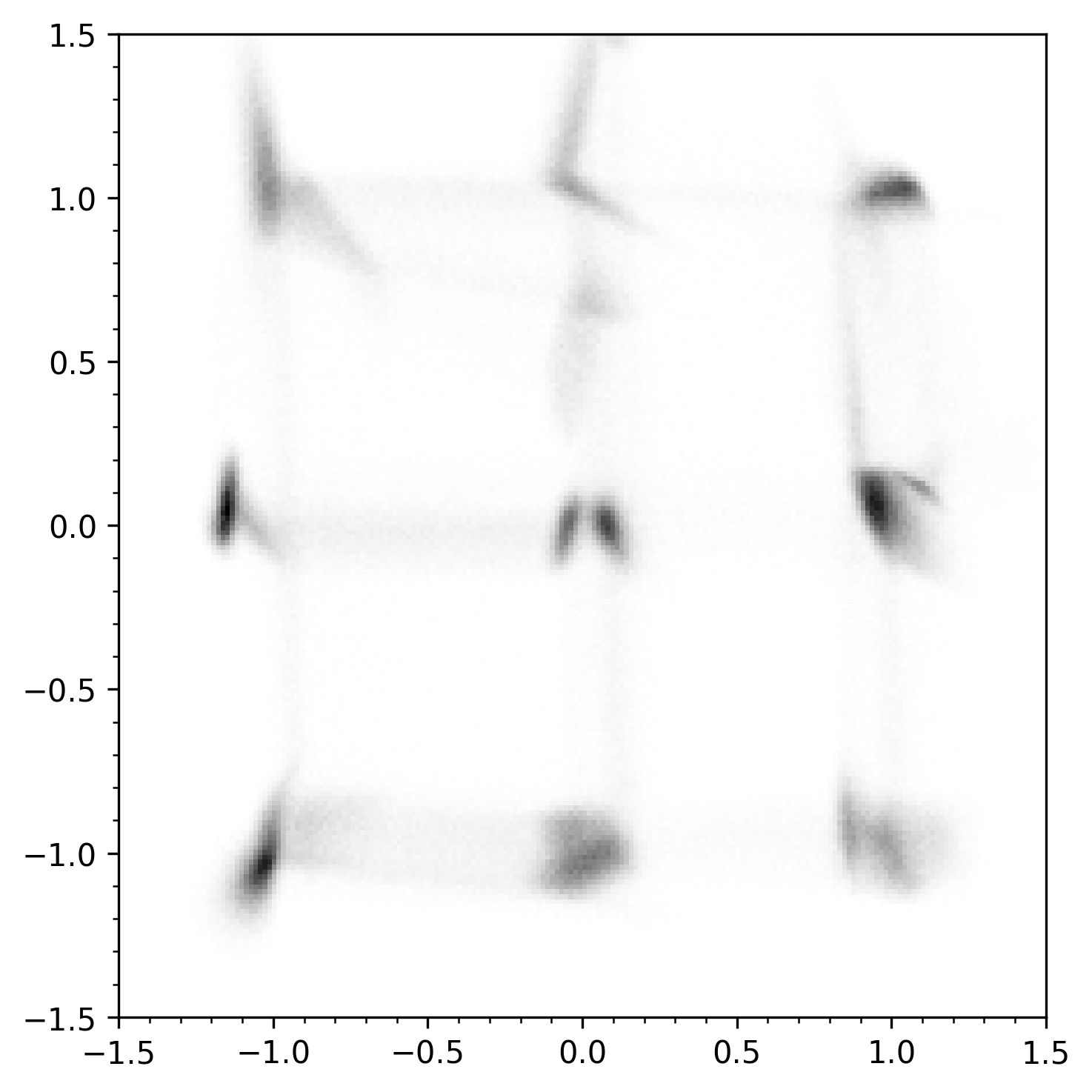}
    \includegraphics[width=.4\linewidth]{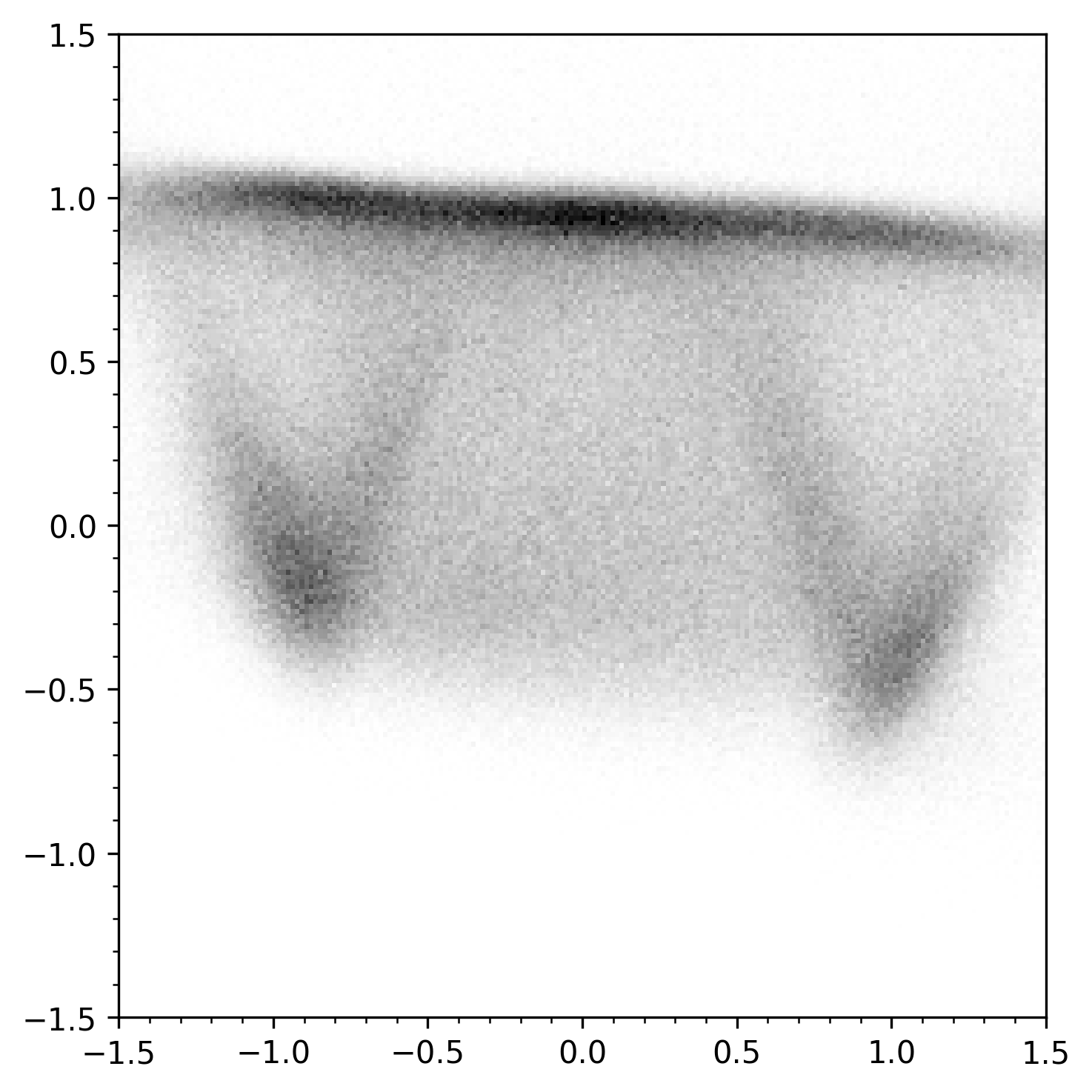}%
    \includegraphics[width=.4\linewidth]{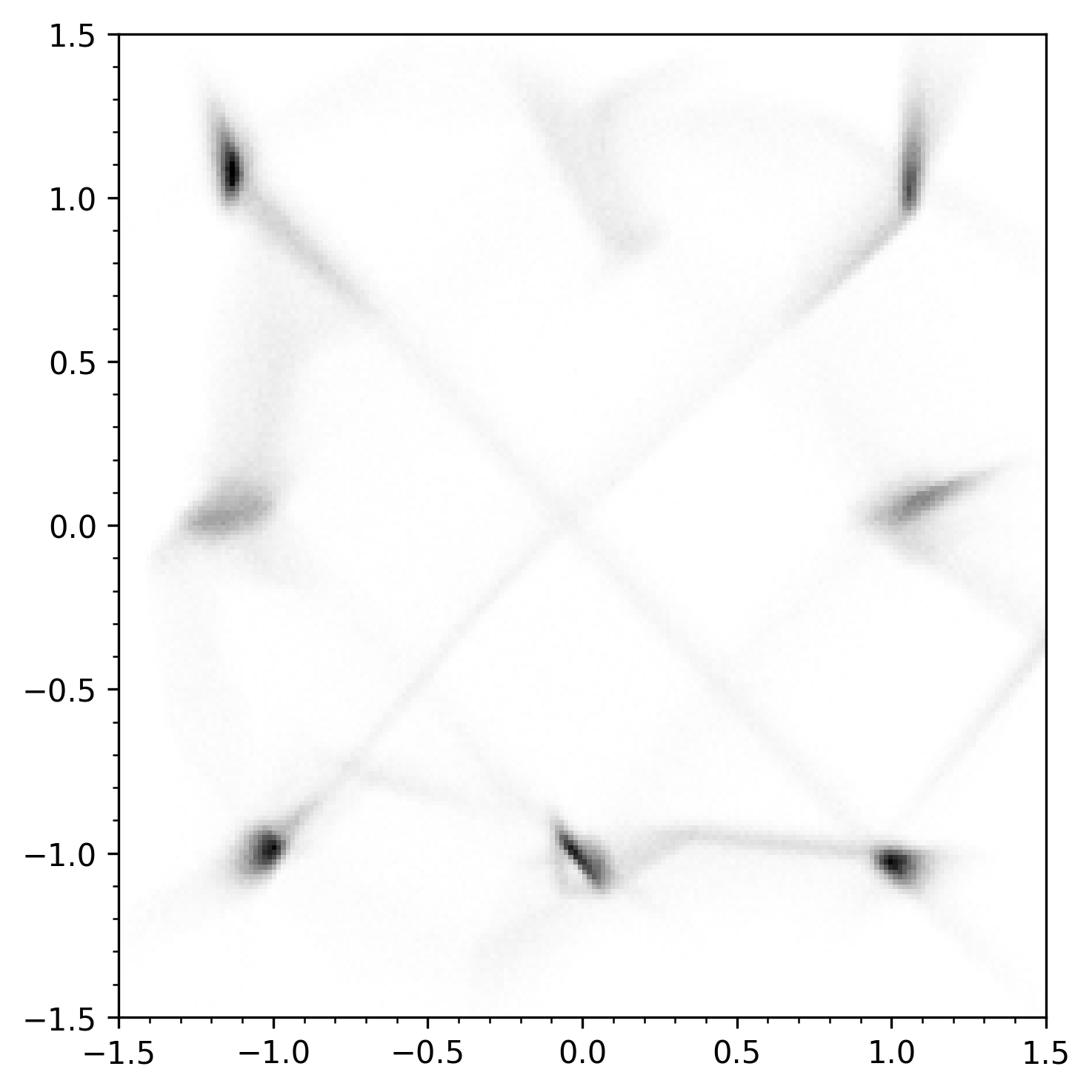}
    \includegraphics[width=.4\linewidth]{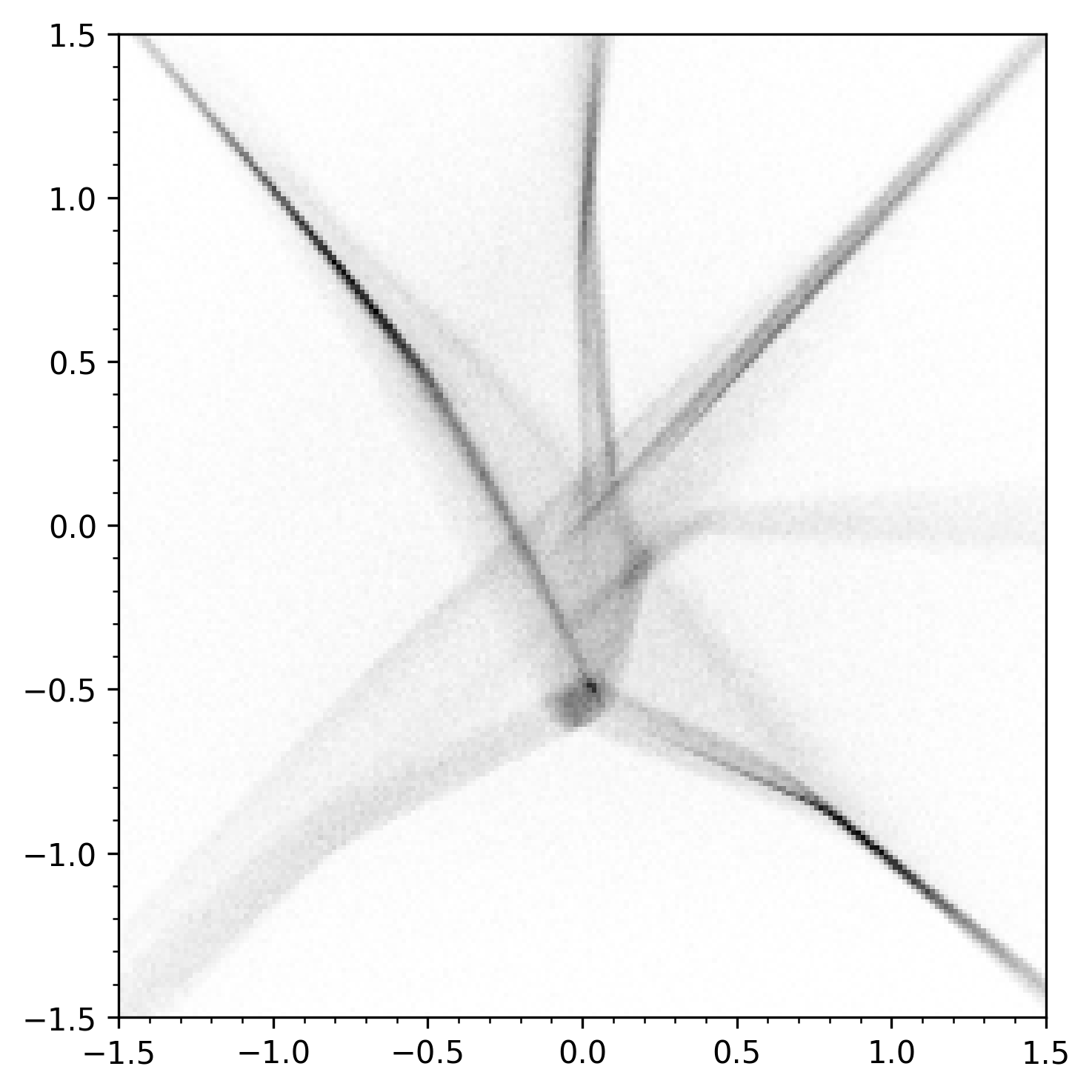}%
    \includegraphics[width=.4\linewidth]{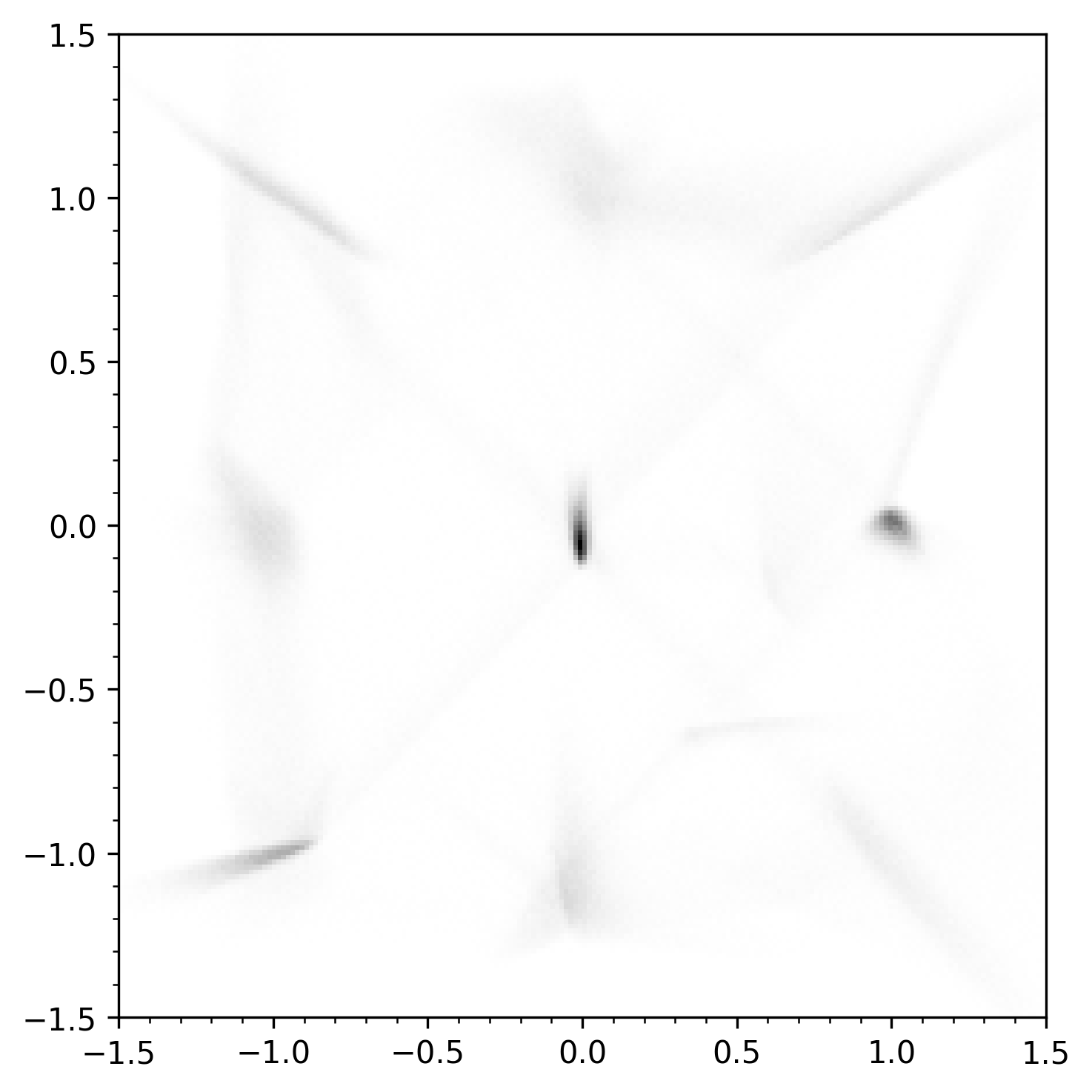}
    \includegraphics[width=.4\linewidth]{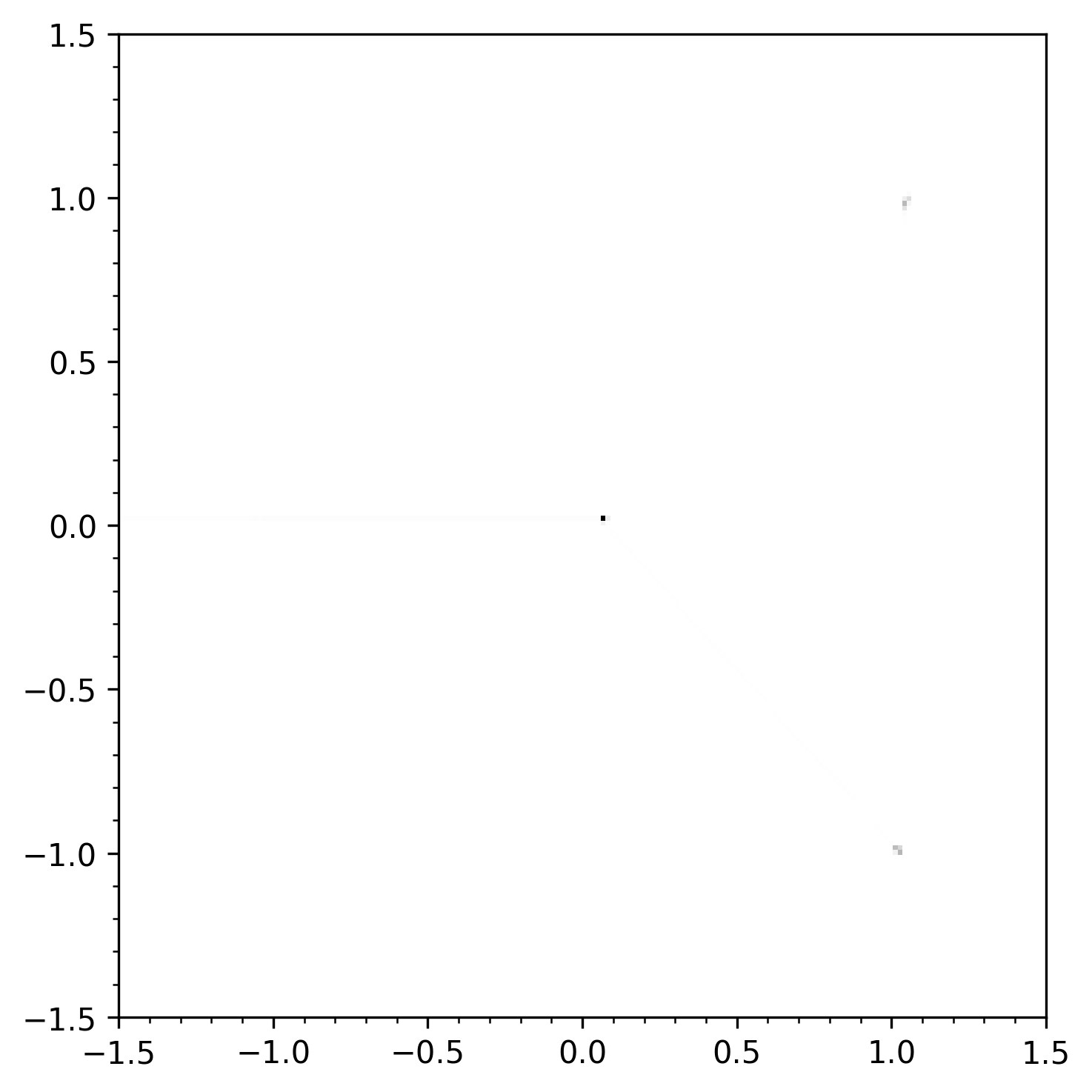}%
    \includegraphics[width=.4\linewidth]{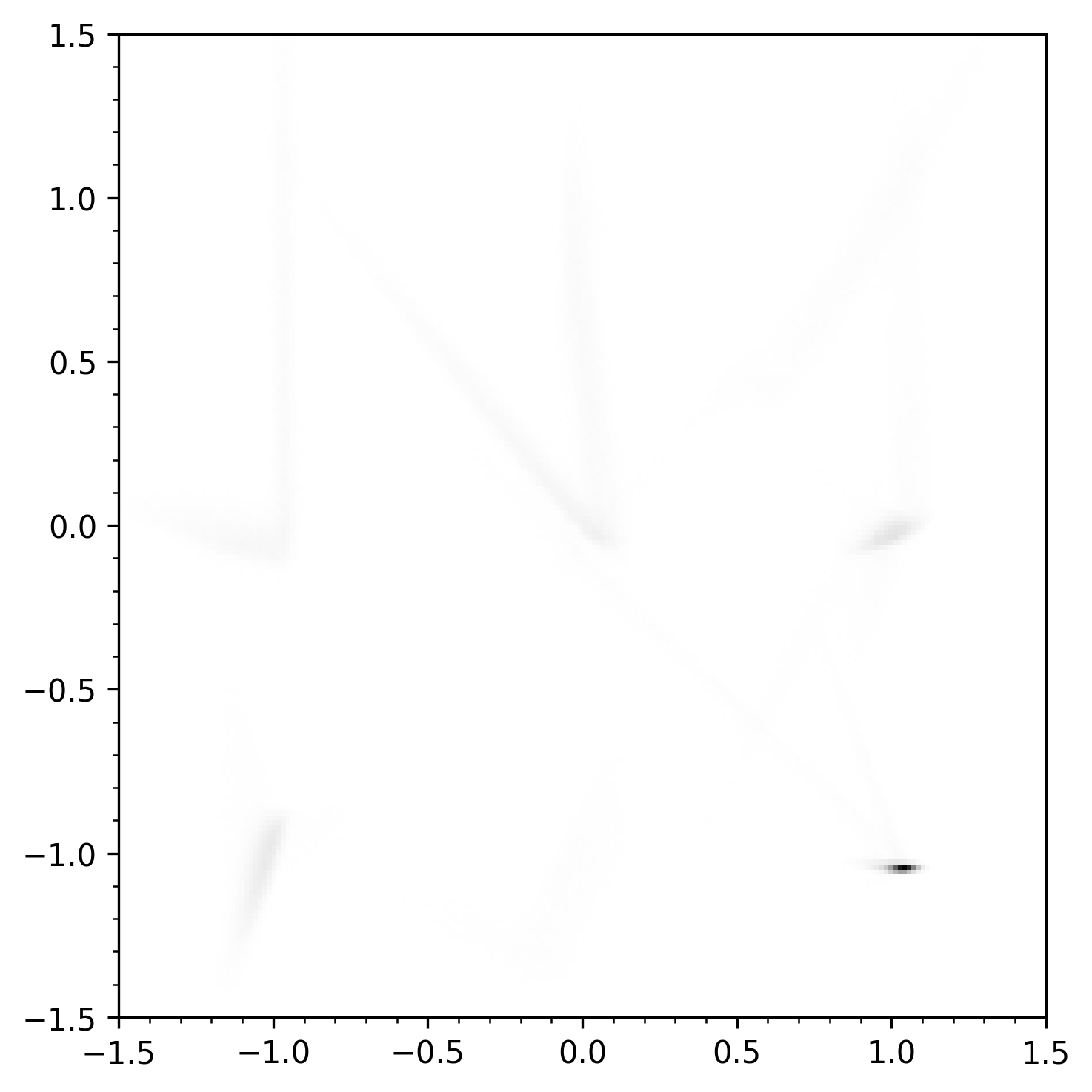}
    \includegraphics[width=.4\linewidth]{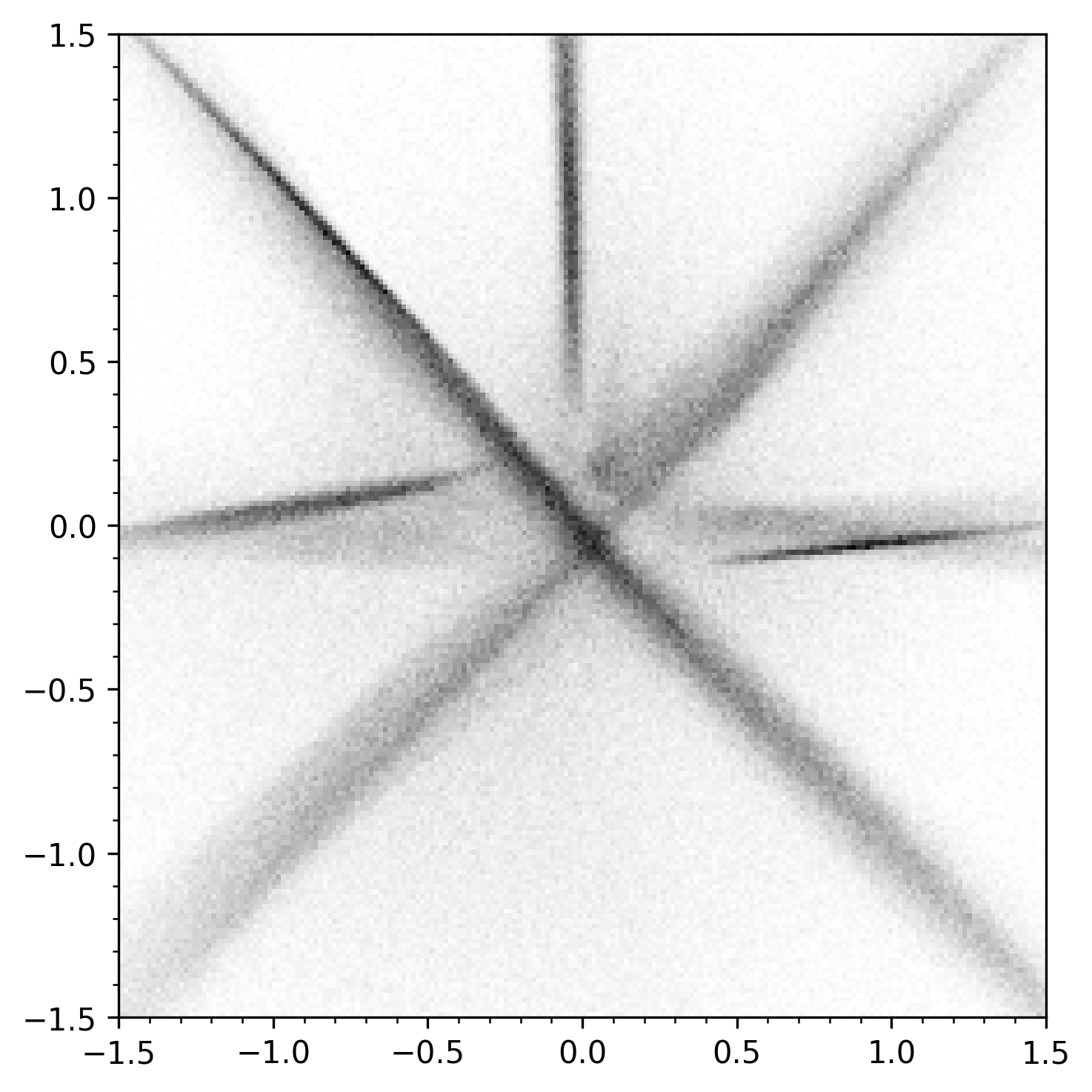}%
    \includegraphics[width=.4\linewidth]{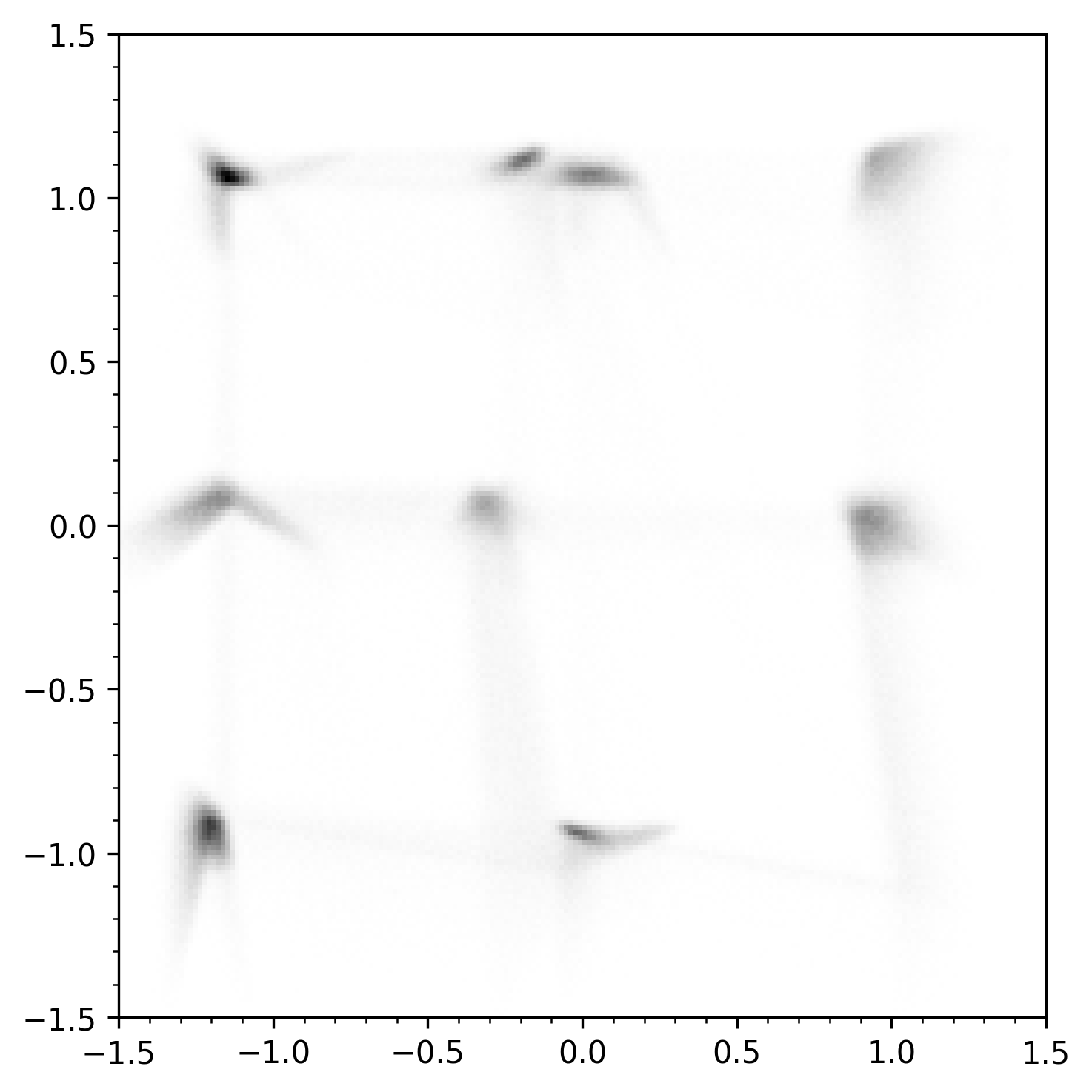}
    \includegraphics[width=.4\linewidth]{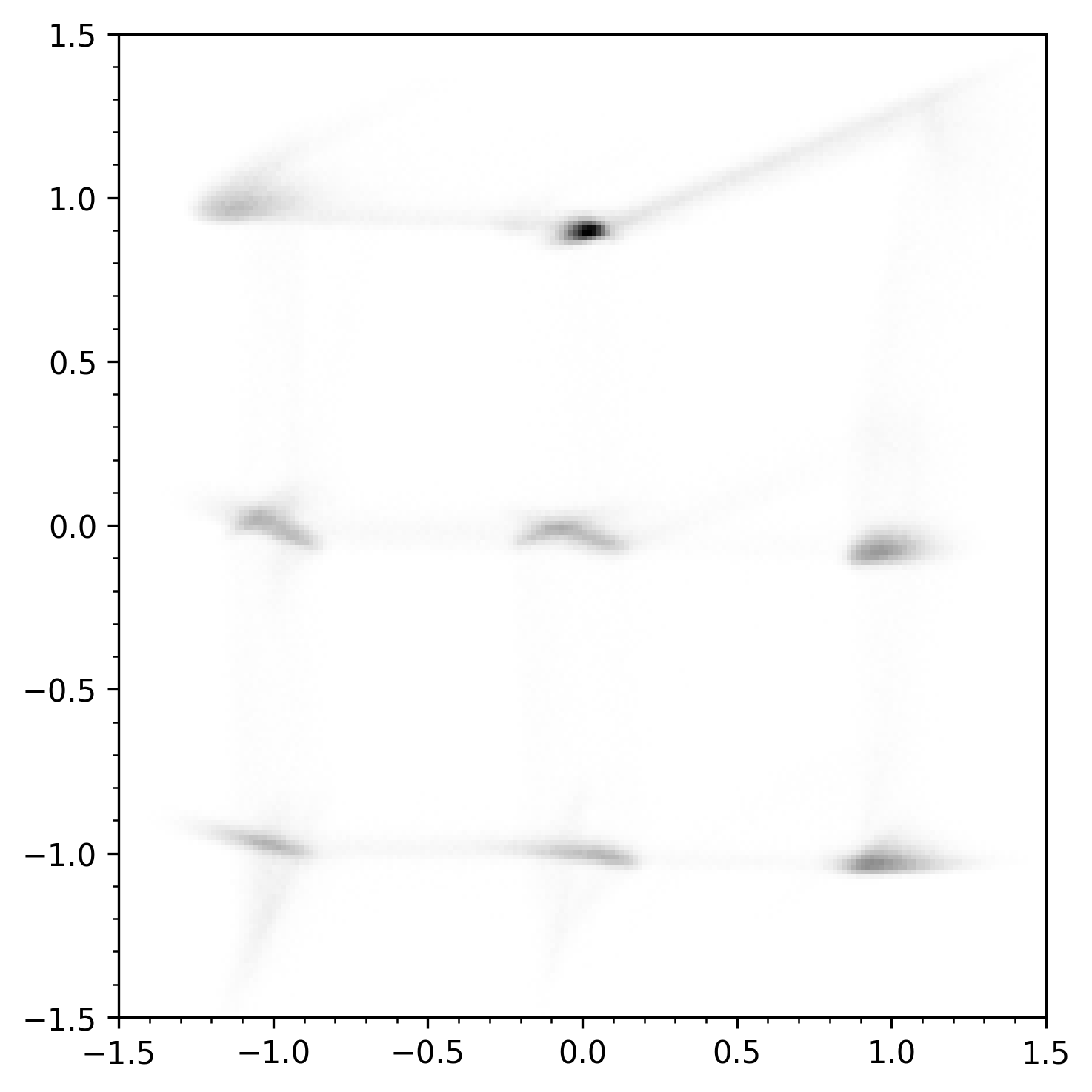}%
    \includegraphics[width=.4\linewidth]{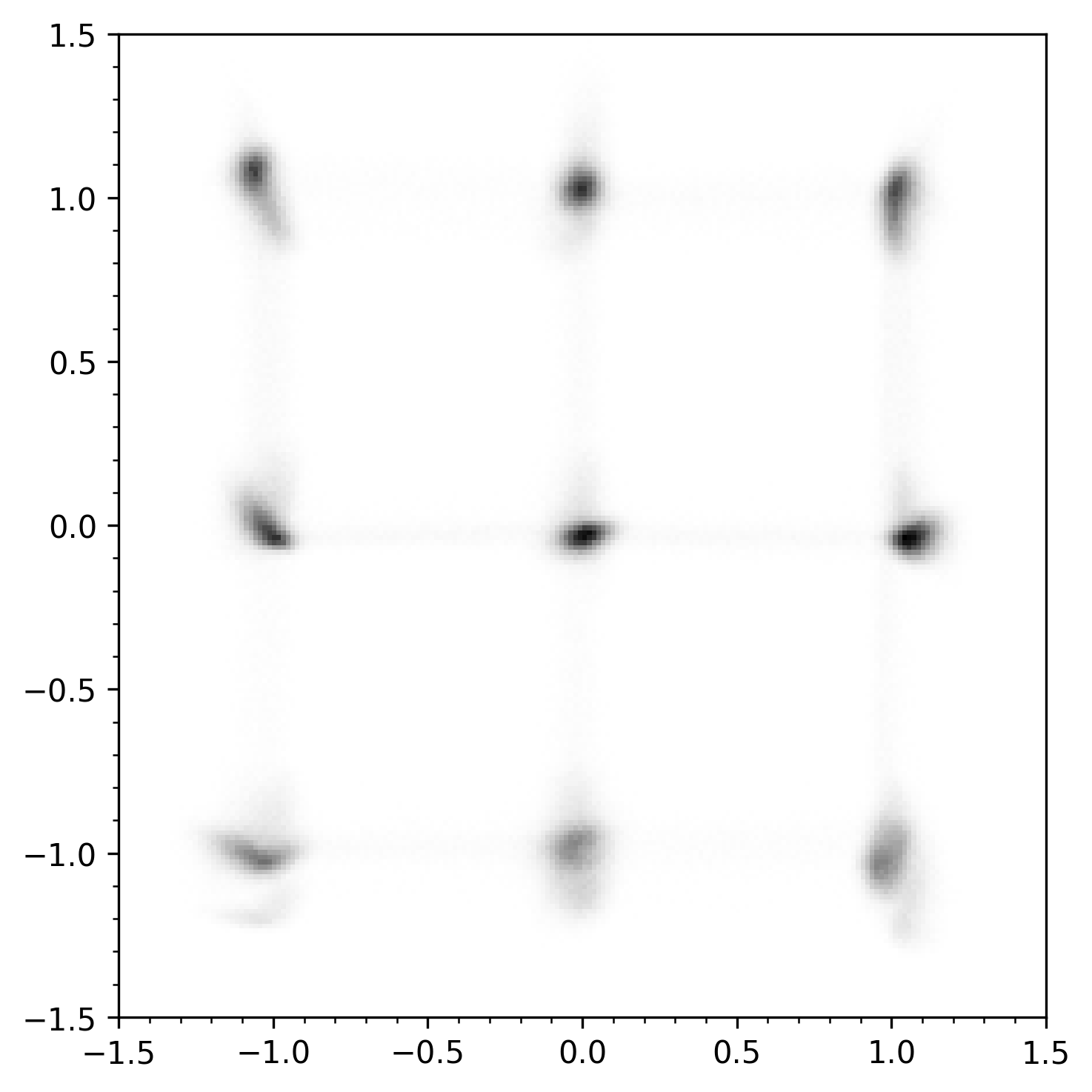}
    \caption{GAN (grid dataset).
    Left to right, top to bottom: Ground truth, SG, OP, EG, CO, SGA, GNI, LA, LOLA, EDA, BRF, BRE. SLA excluded due to divergence.}
    \label{fig:gan_images_grid}
\end{figure}

\begin{figure}
    \centering
    \includegraphics[width=.4\linewidth]{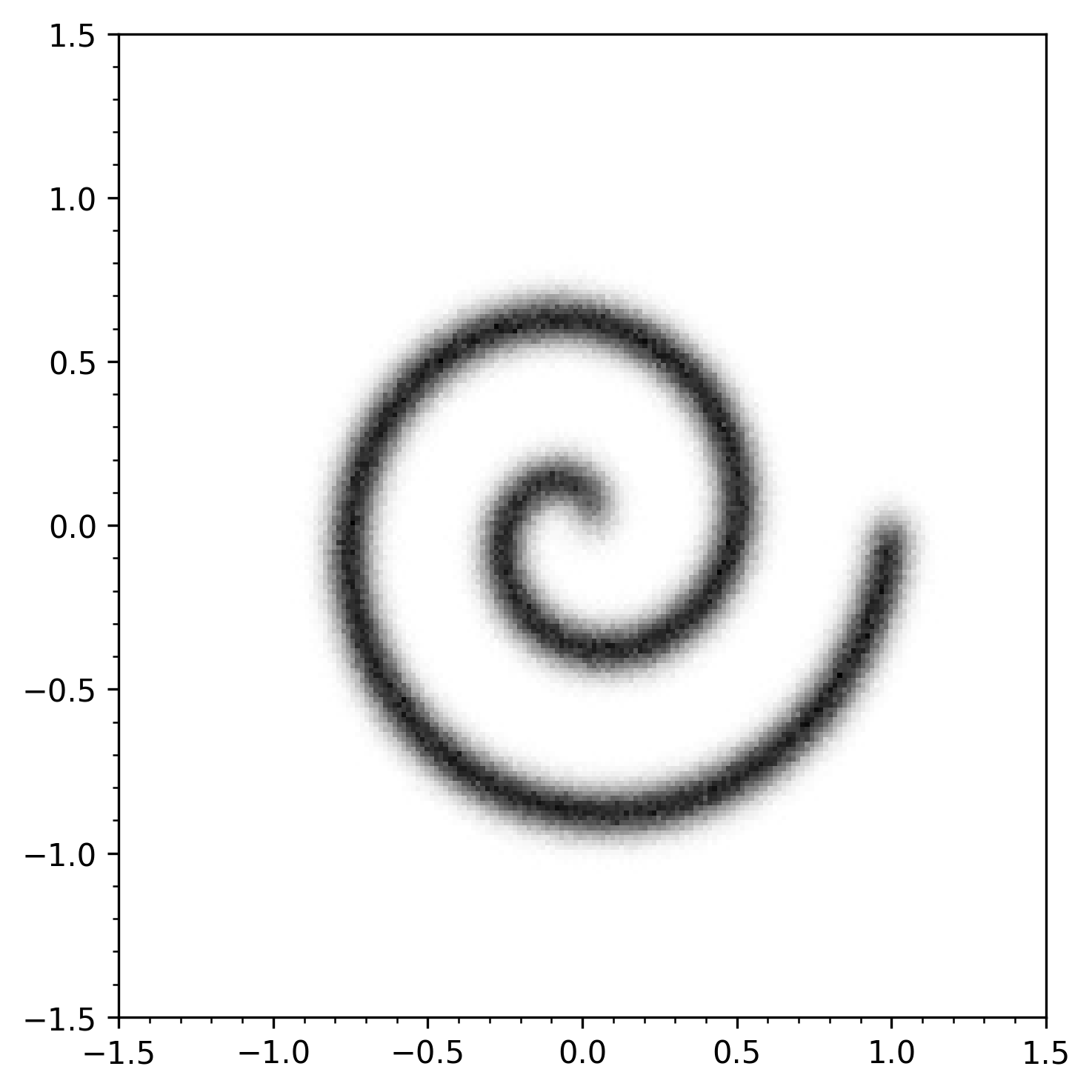}%
    \includegraphics[width=.4\linewidth]{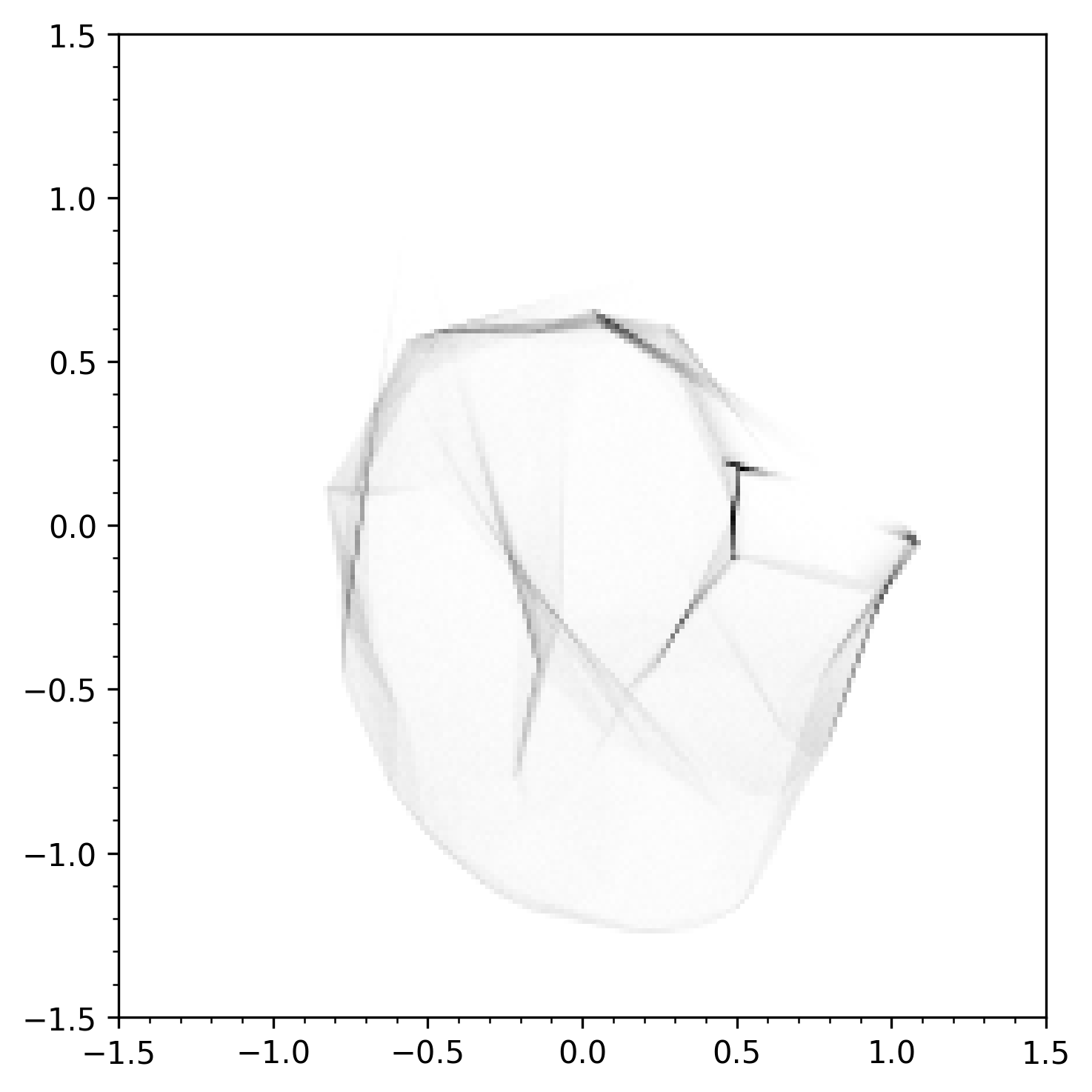}
    \includegraphics[width=.4\linewidth]{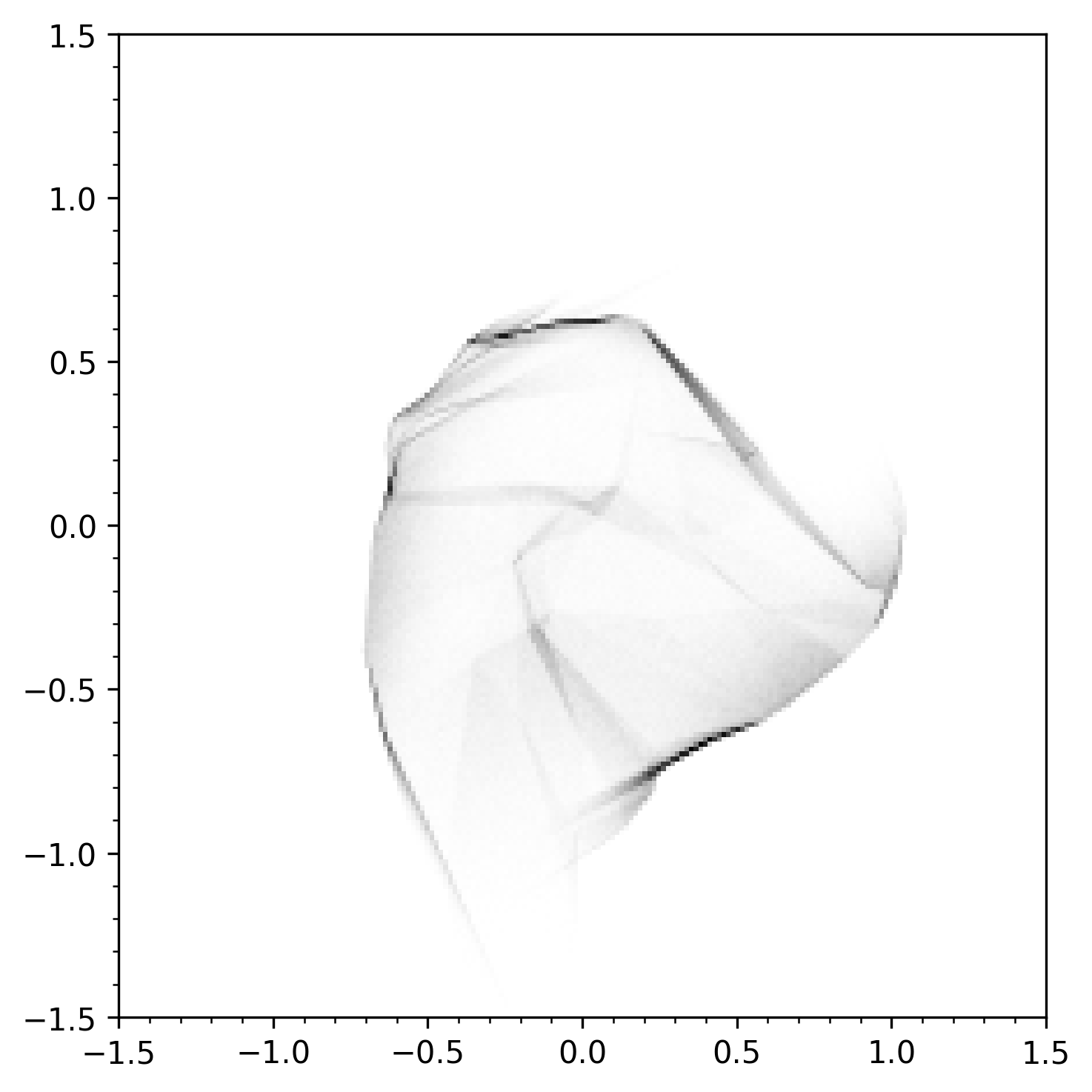}%
    \includegraphics[width=.4\linewidth]{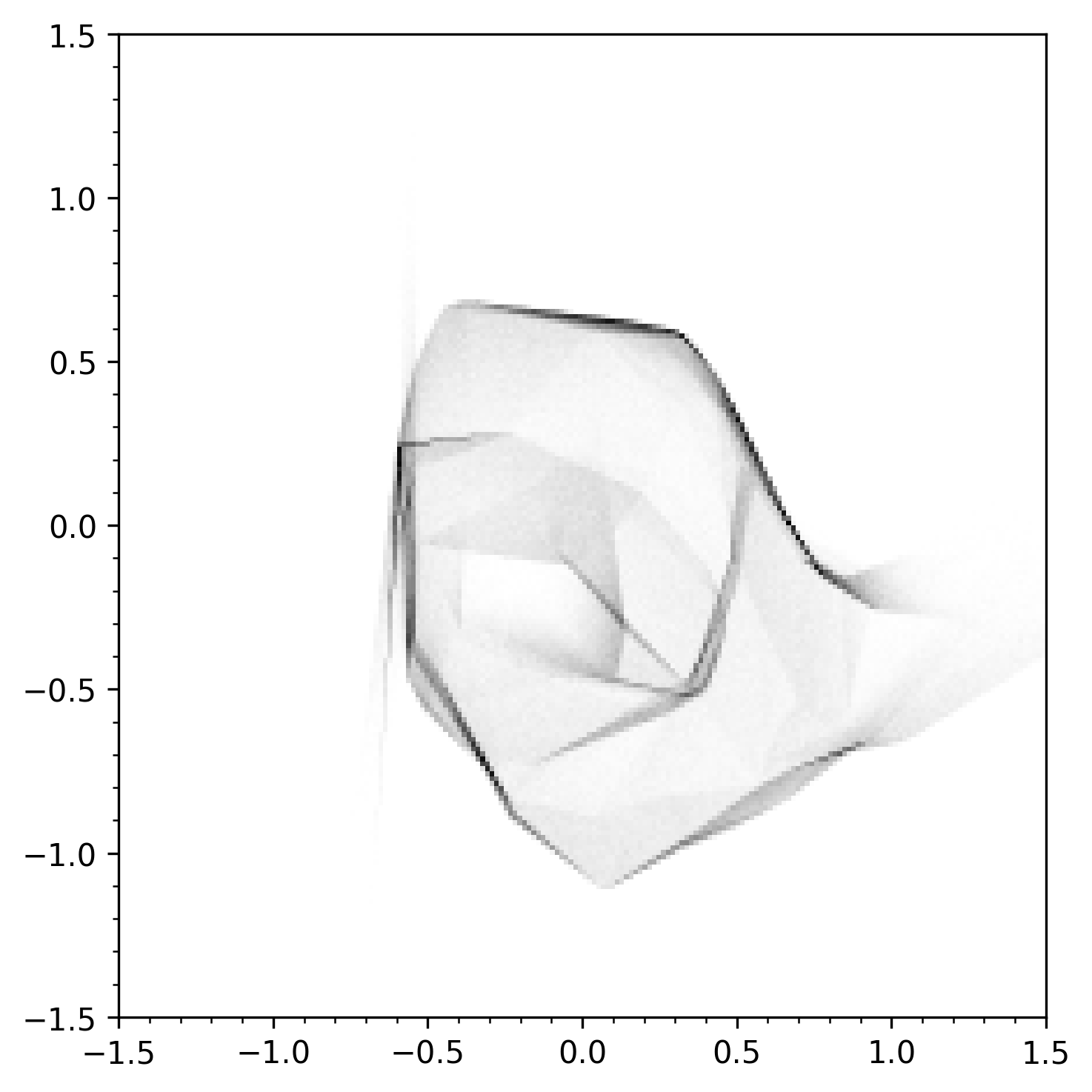}
    \includegraphics[width=.4\linewidth]{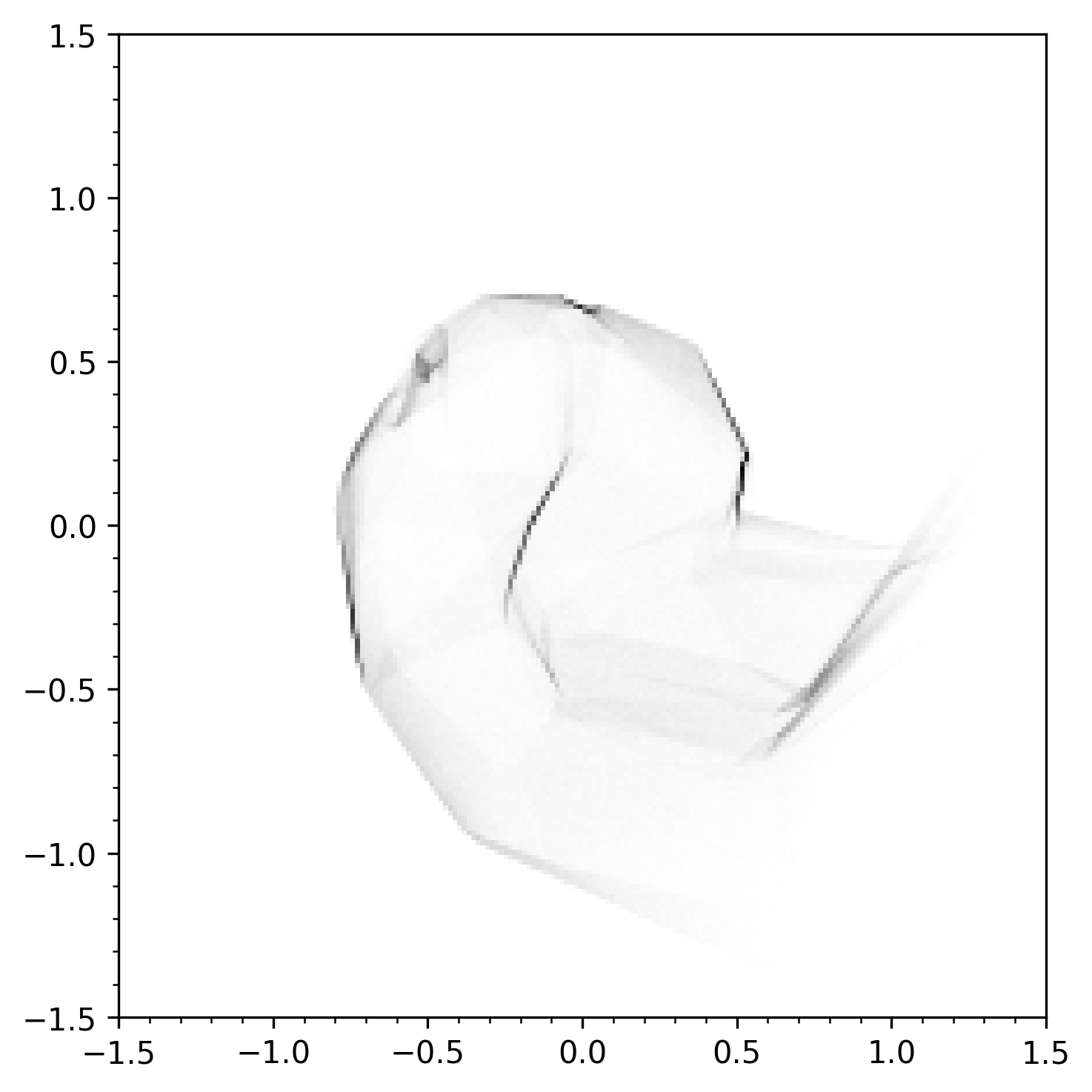}%
    \includegraphics[width=.4\linewidth]{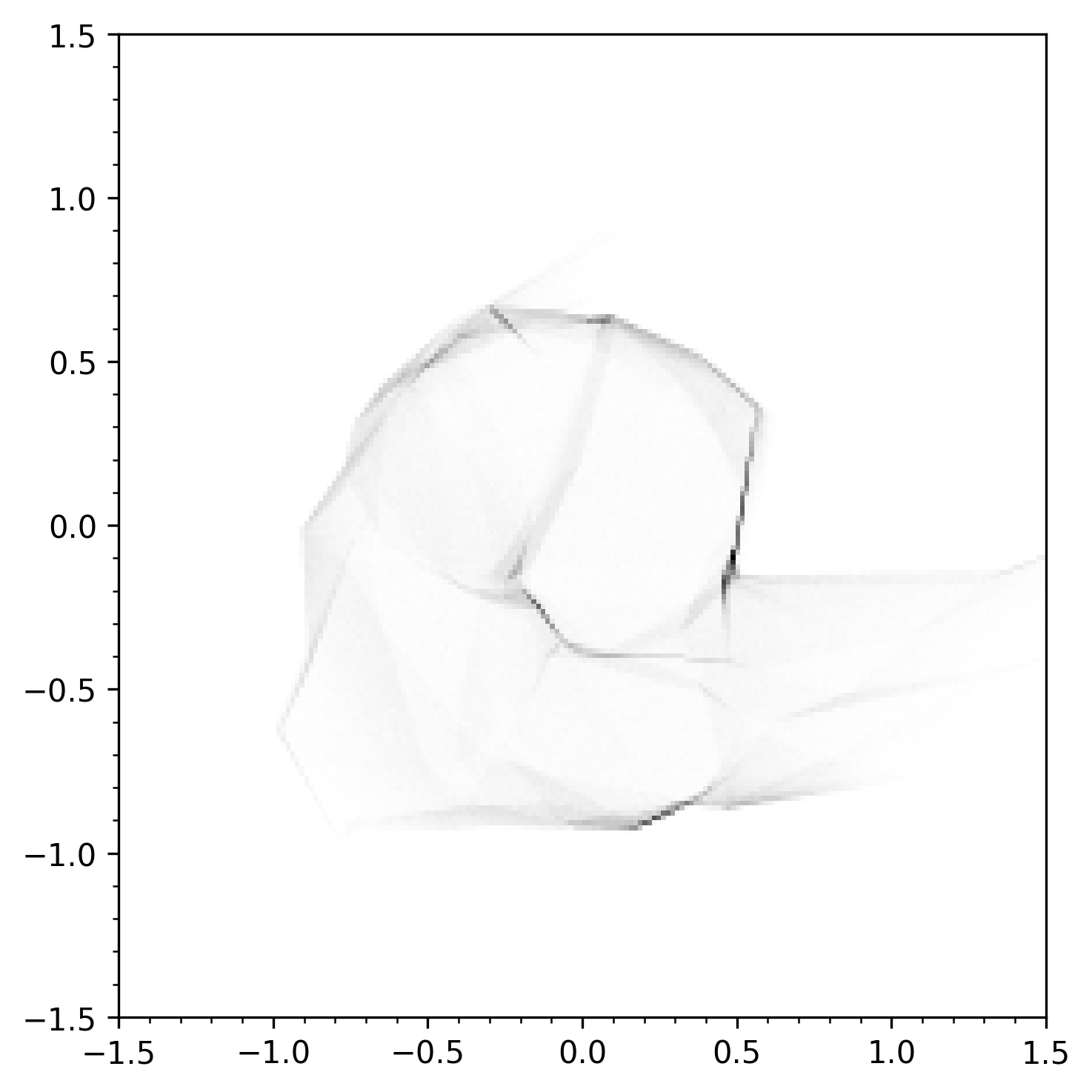}
    \includegraphics[width=.4\linewidth]{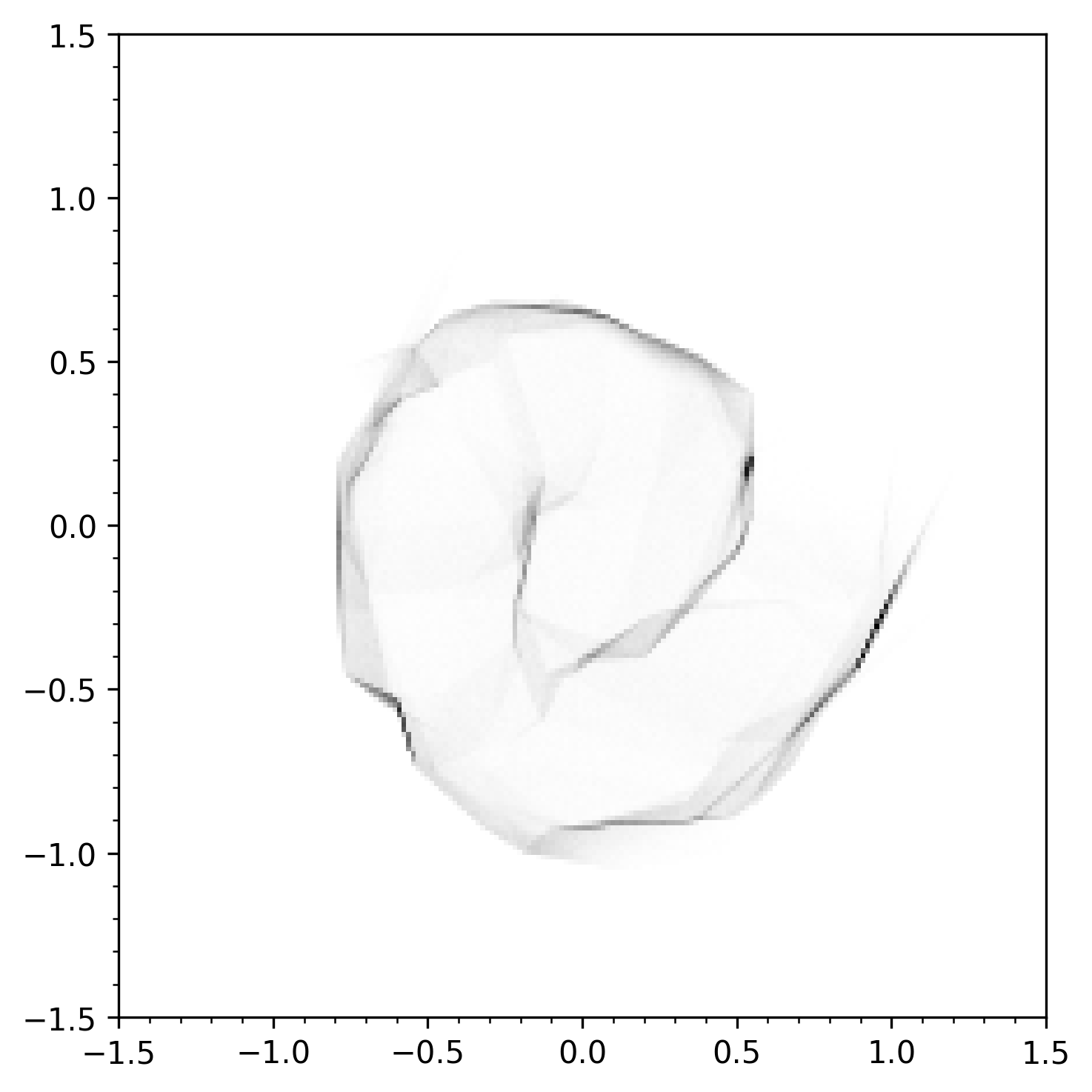}%
    \includegraphics[width=.4\linewidth]{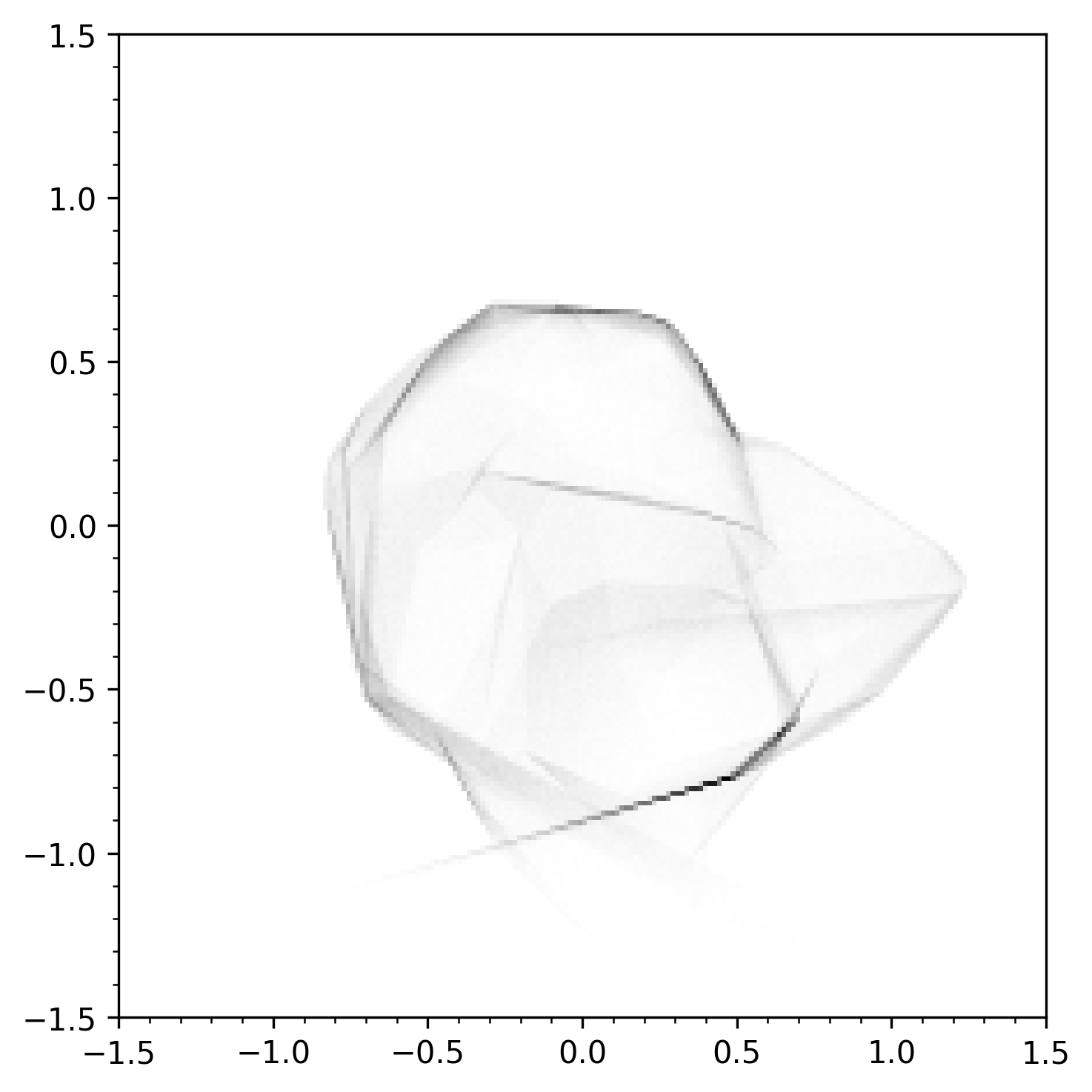}
    \includegraphics[width=.4\linewidth]{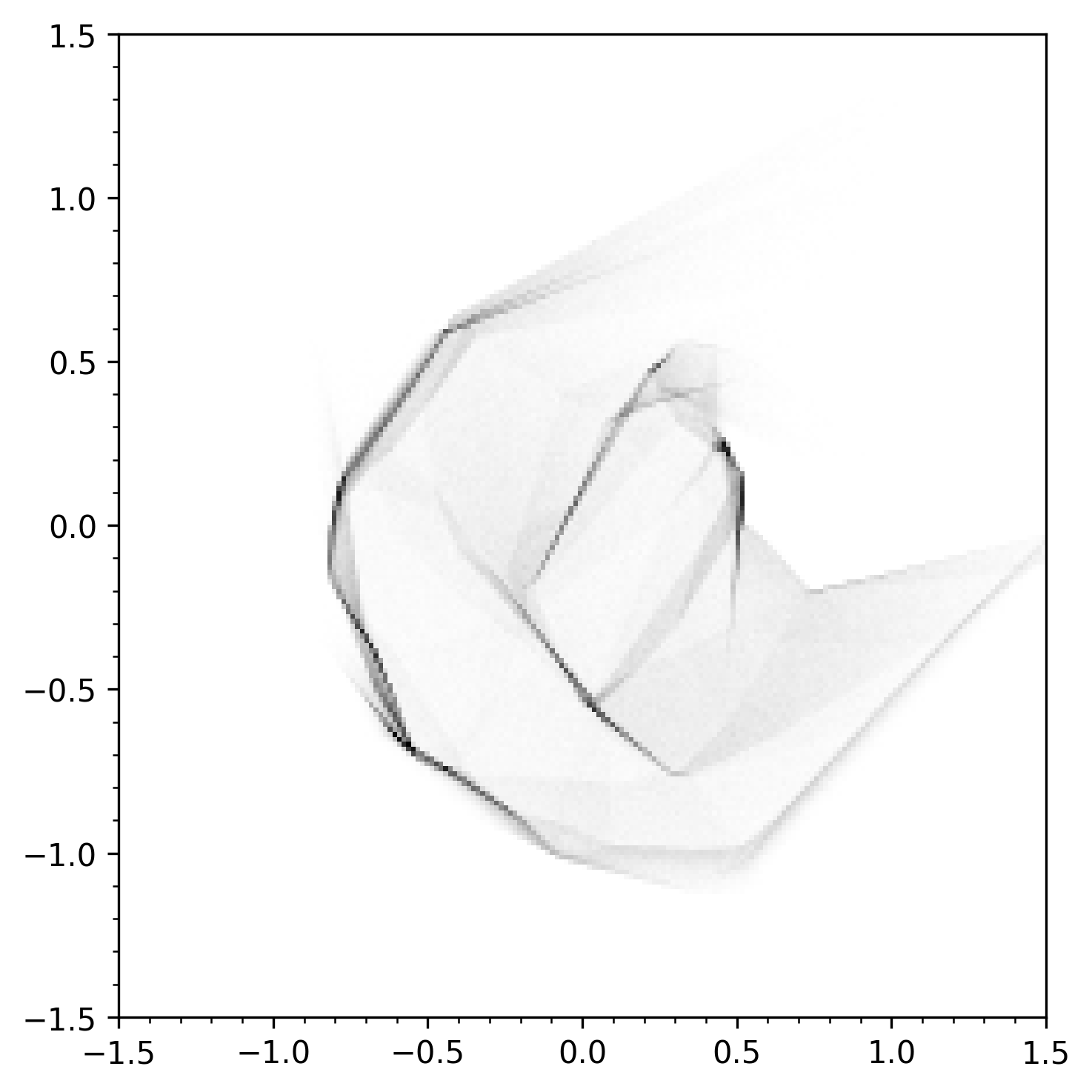}%
    \includegraphics[width=.4\linewidth]{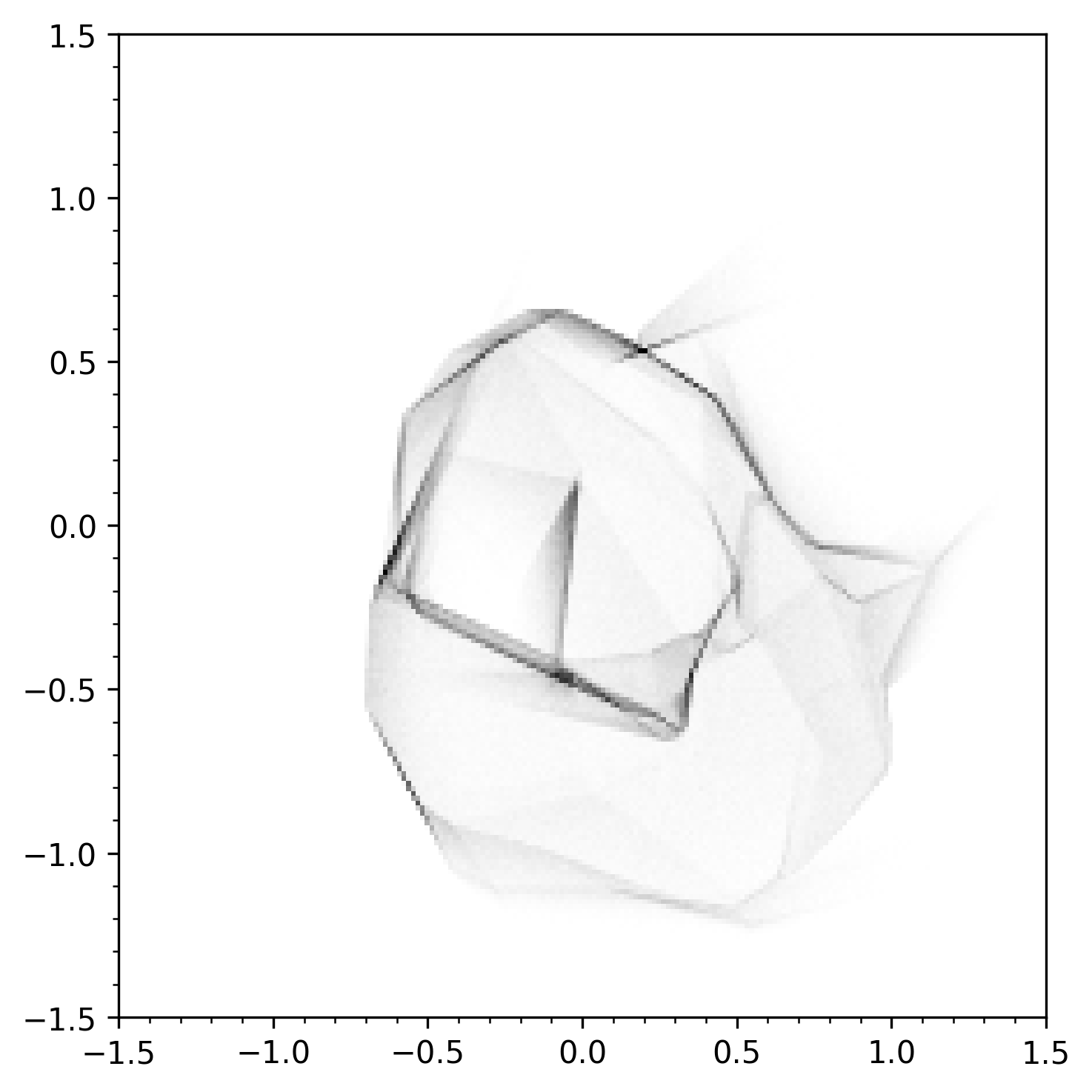}
    \includegraphics[width=.4\linewidth]{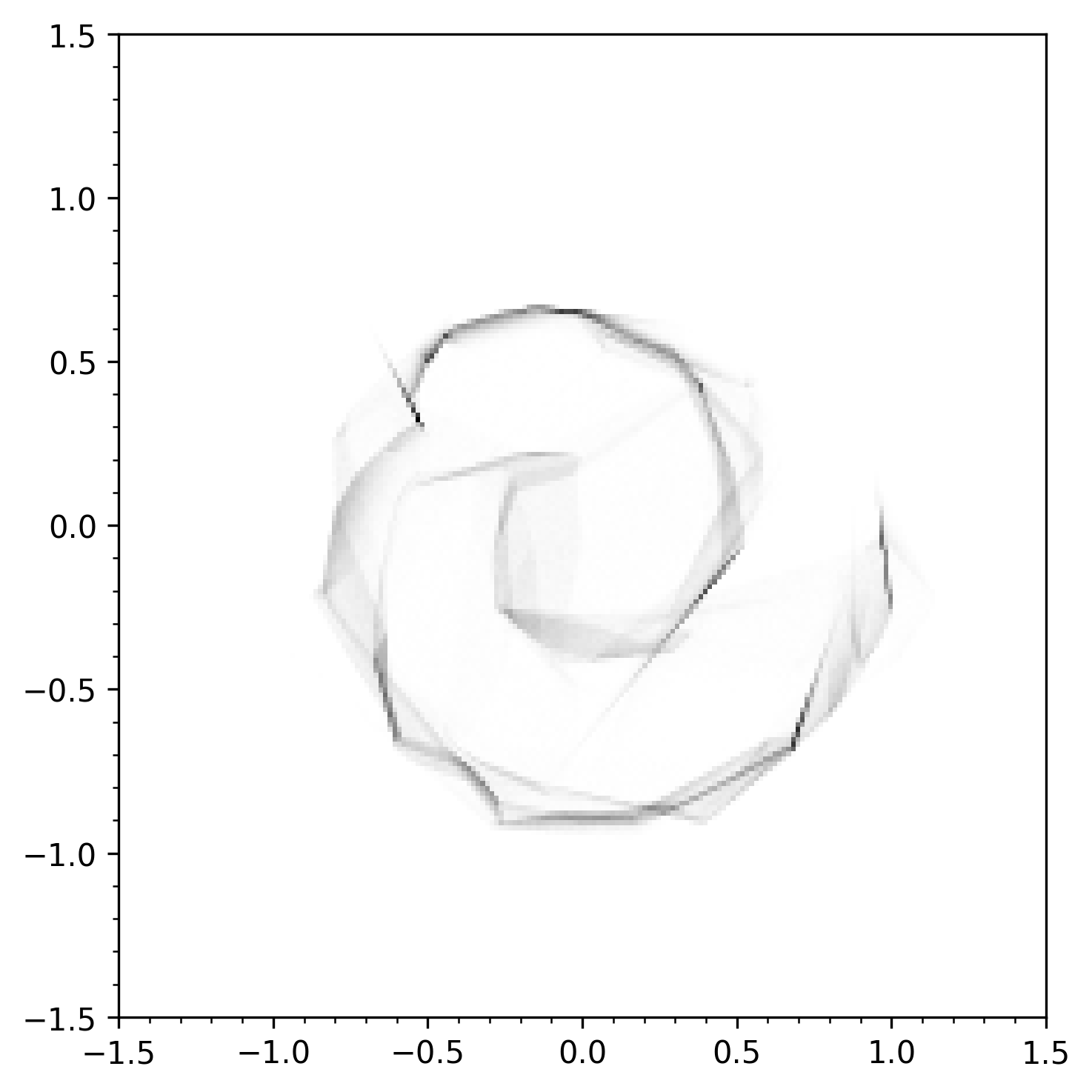}%
    \includegraphics[width=.4\linewidth]{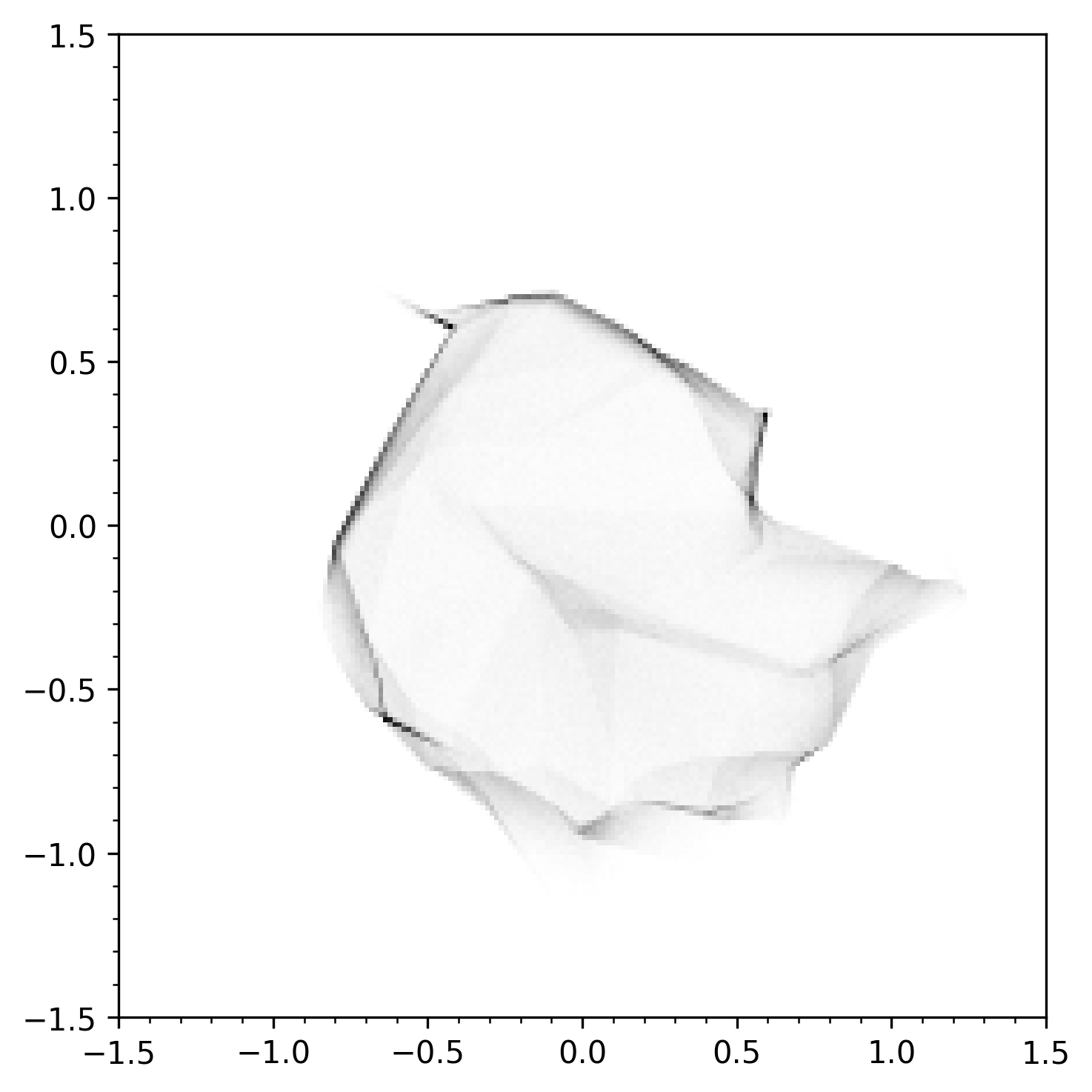}
    \caption{GAN (spiral dataset).
    Left to right, top to bottom: Ground truth, SG, OP, EG, CO, SGA, GNI, LA, LOLA, EDA, BRF, BRE. SLA excluded due to divergence.}
    \label{fig:gan_images_spiral}
\end{figure}

\begin{figure}
    \centering
    \includegraphics[width=.4\linewidth]{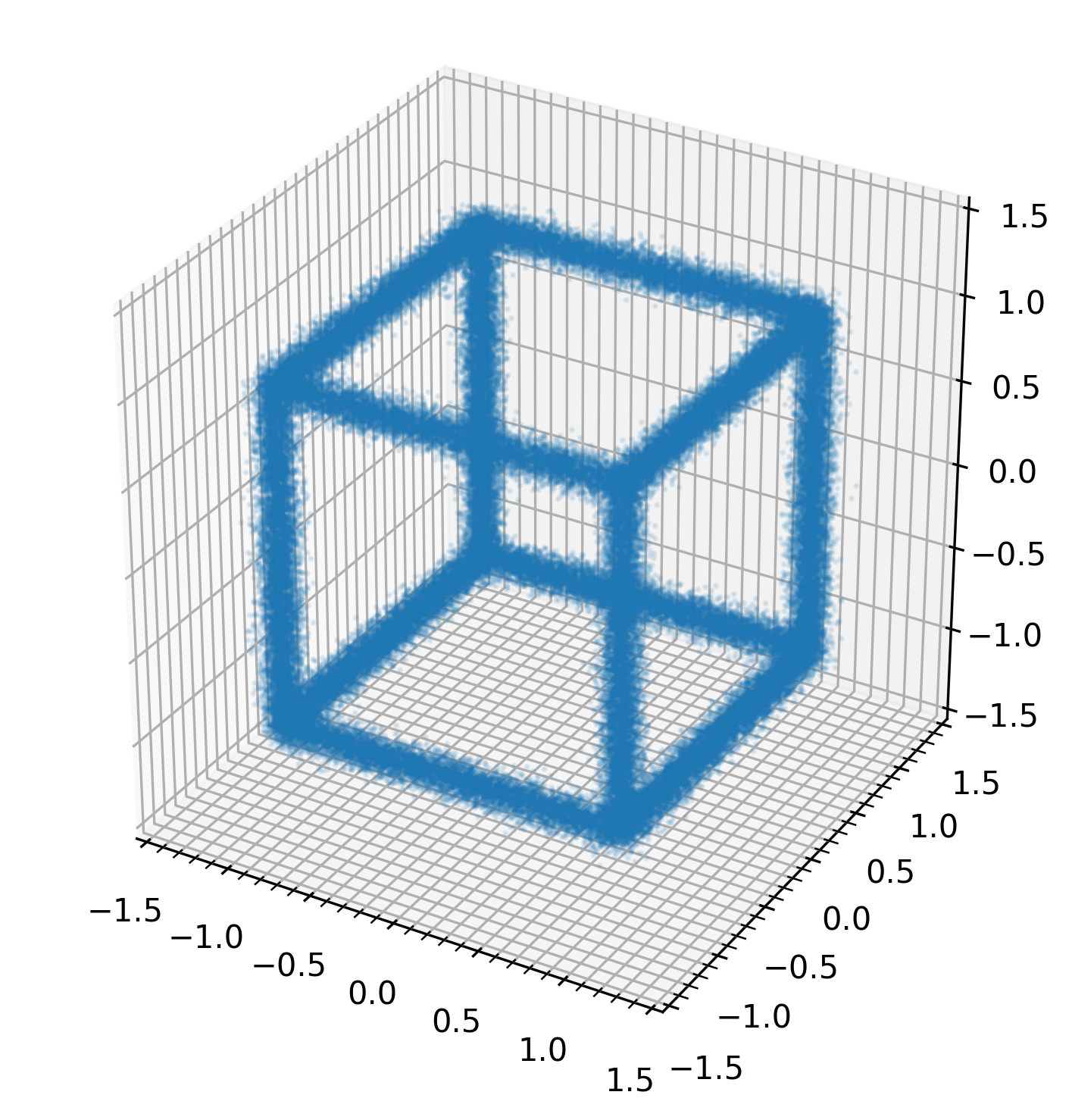}%
    \includegraphics[width=.4\linewidth]{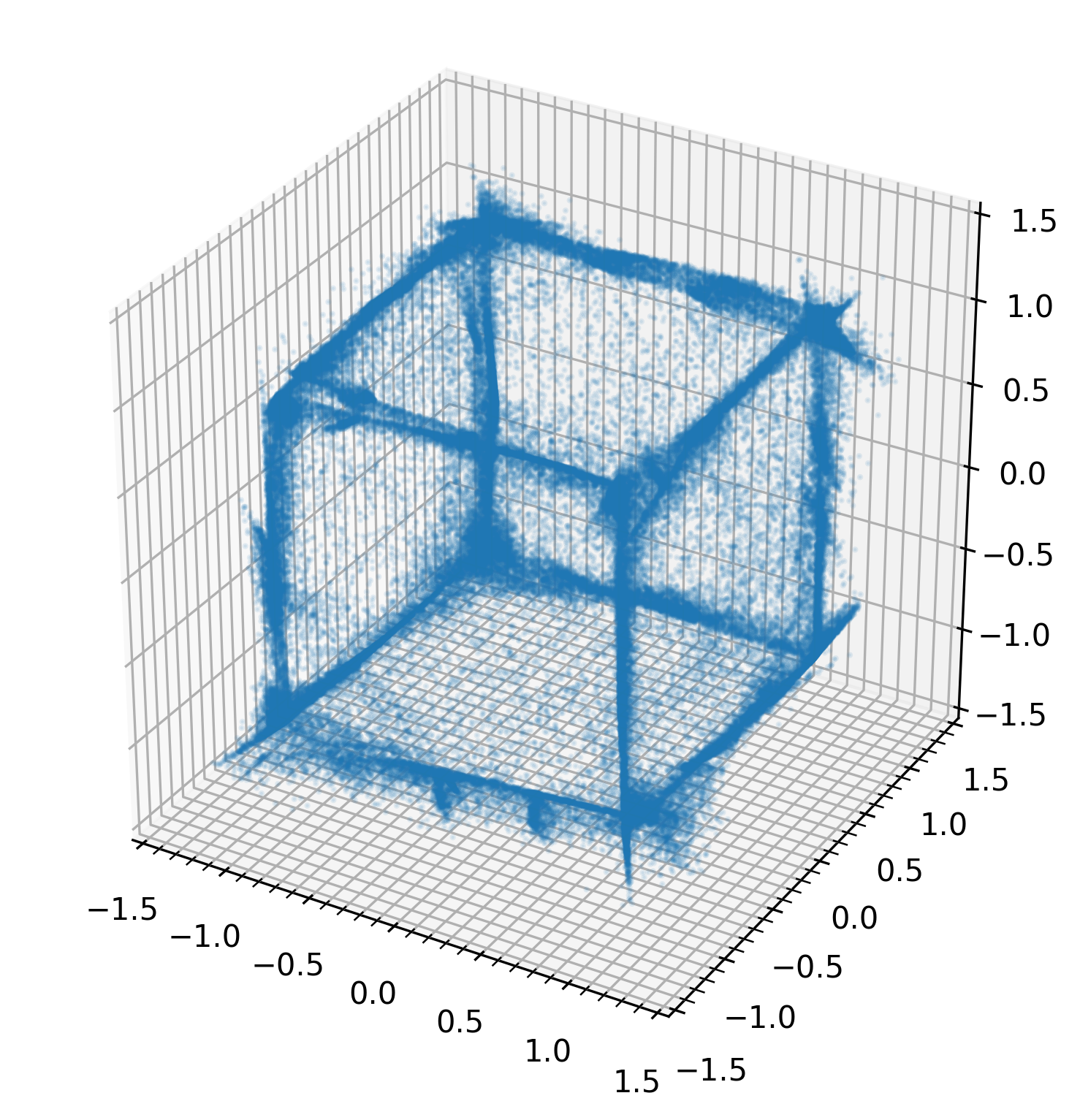}
    \includegraphics[width=.4\linewidth]{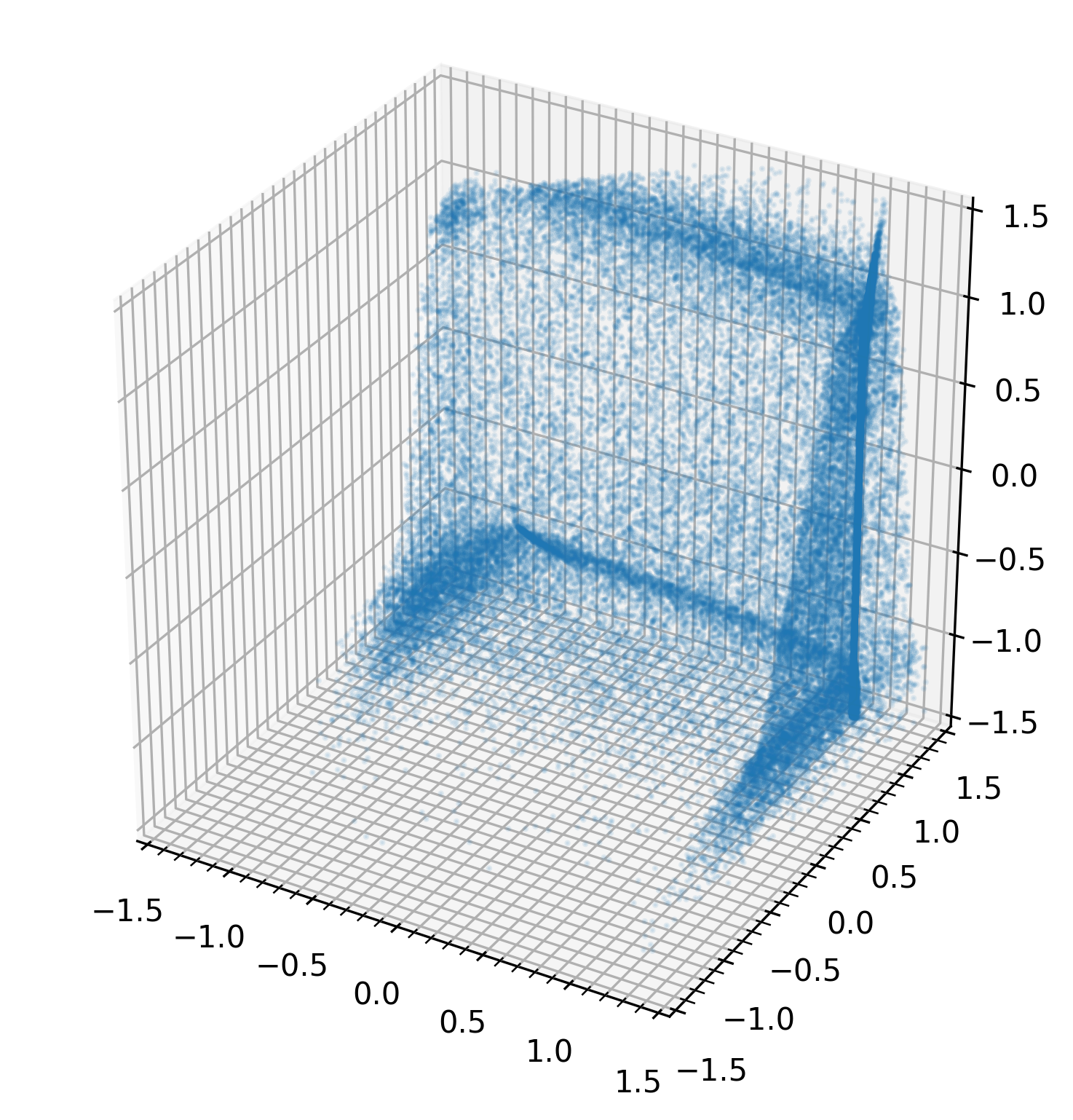}%
    \includegraphics[width=.4\linewidth]{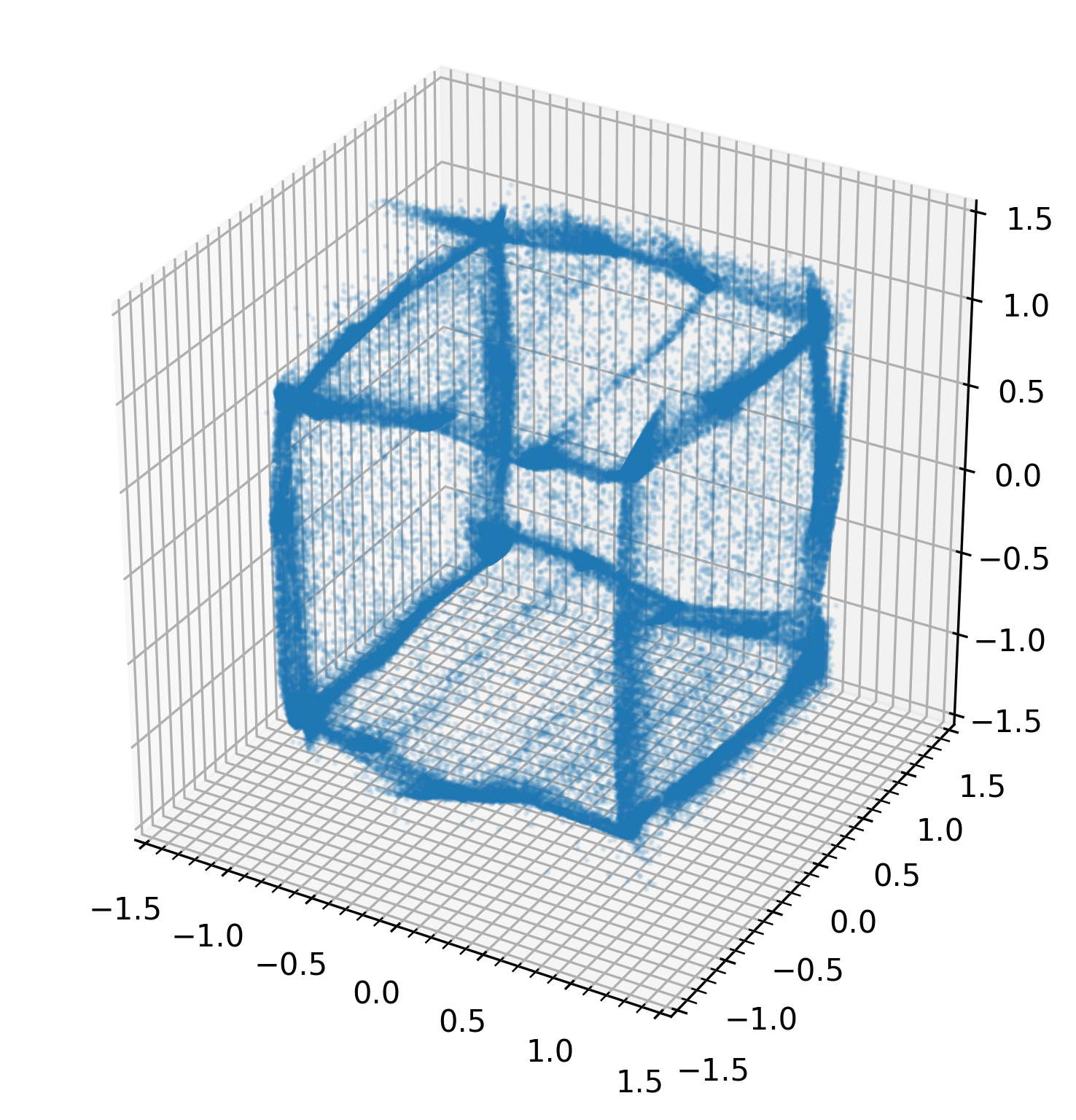}
    \includegraphics[width=.4\linewidth]{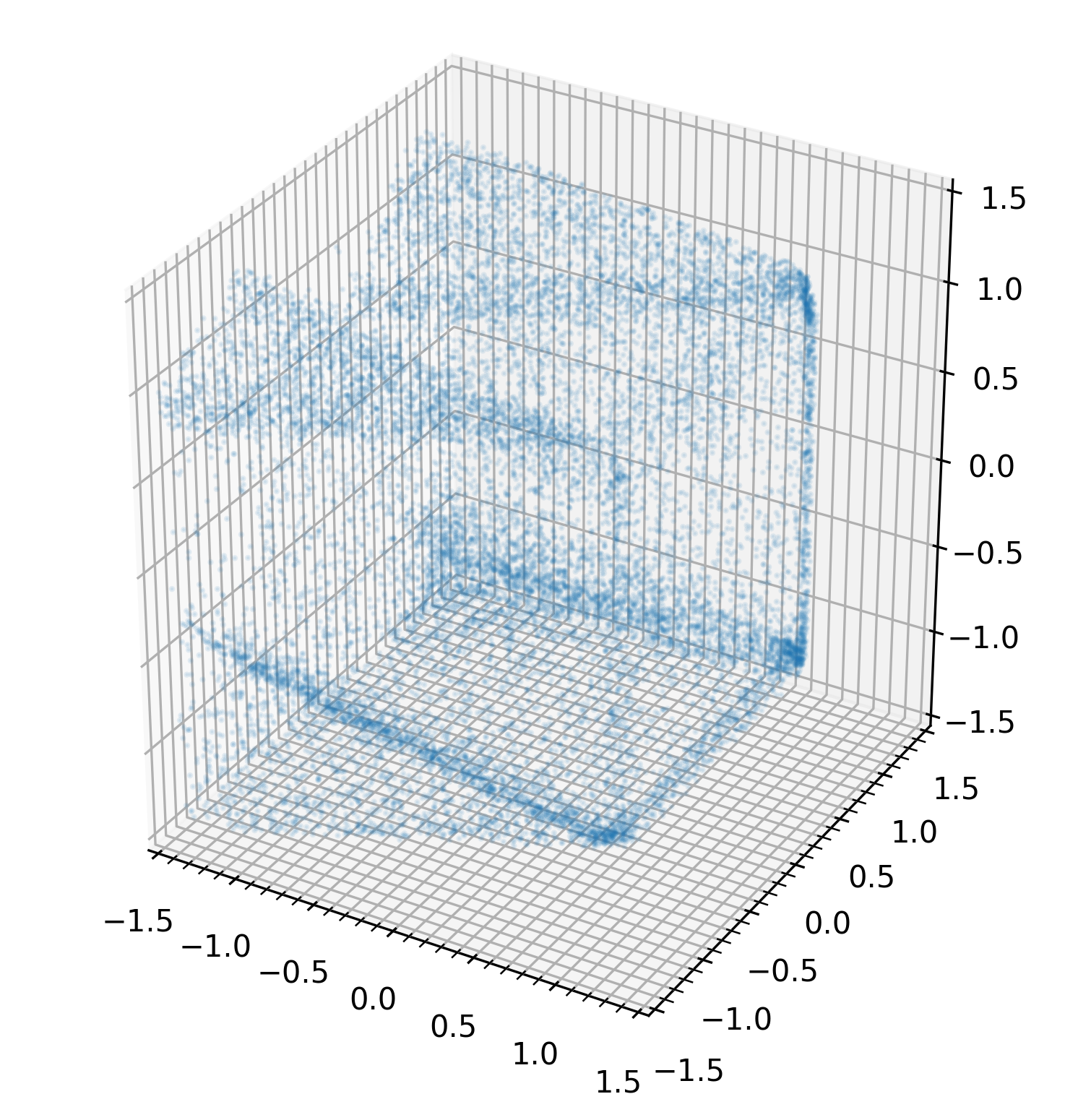}%
    \includegraphics[width=.4\linewidth]{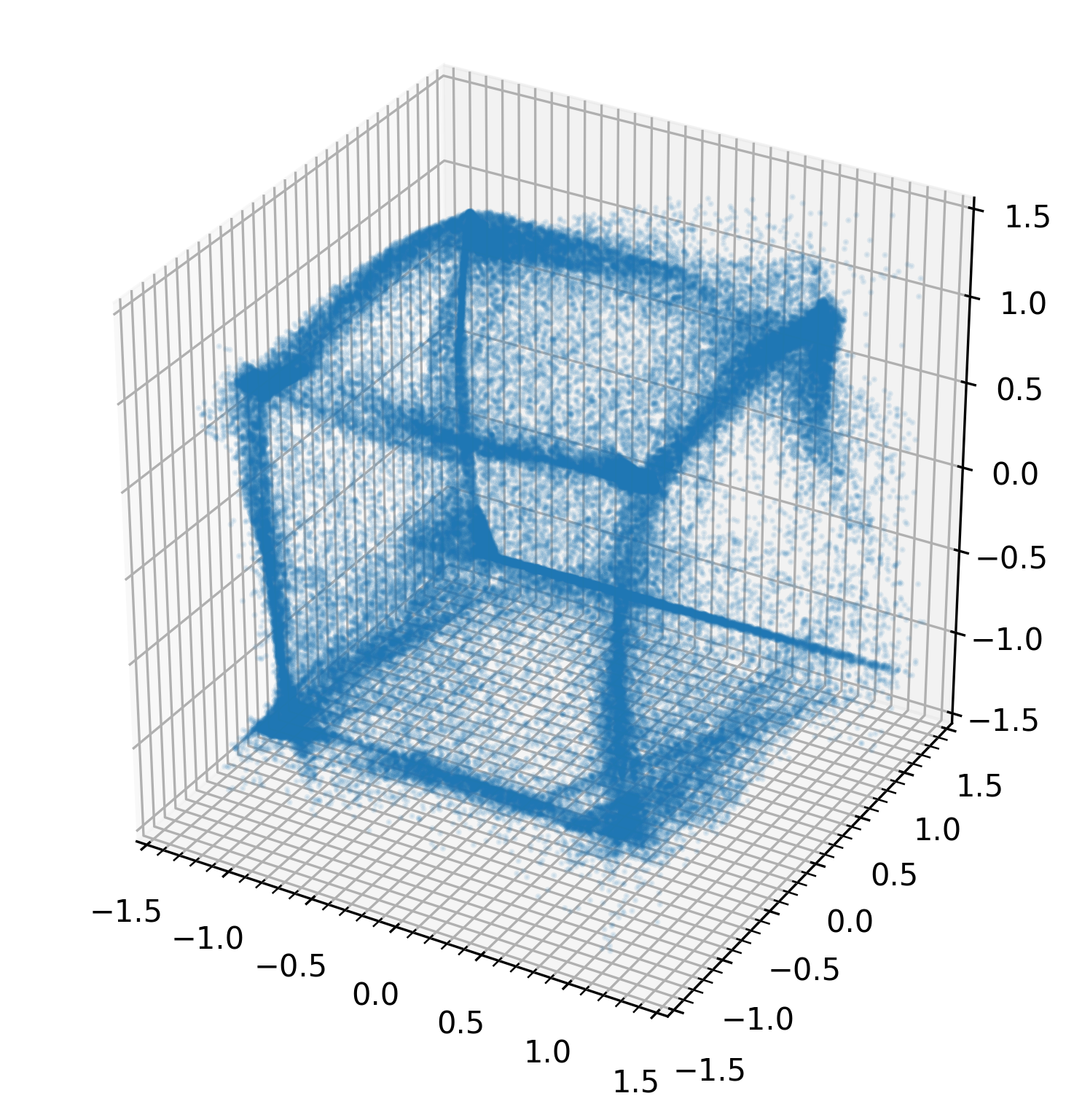}
    \includegraphics[width=.4\linewidth]{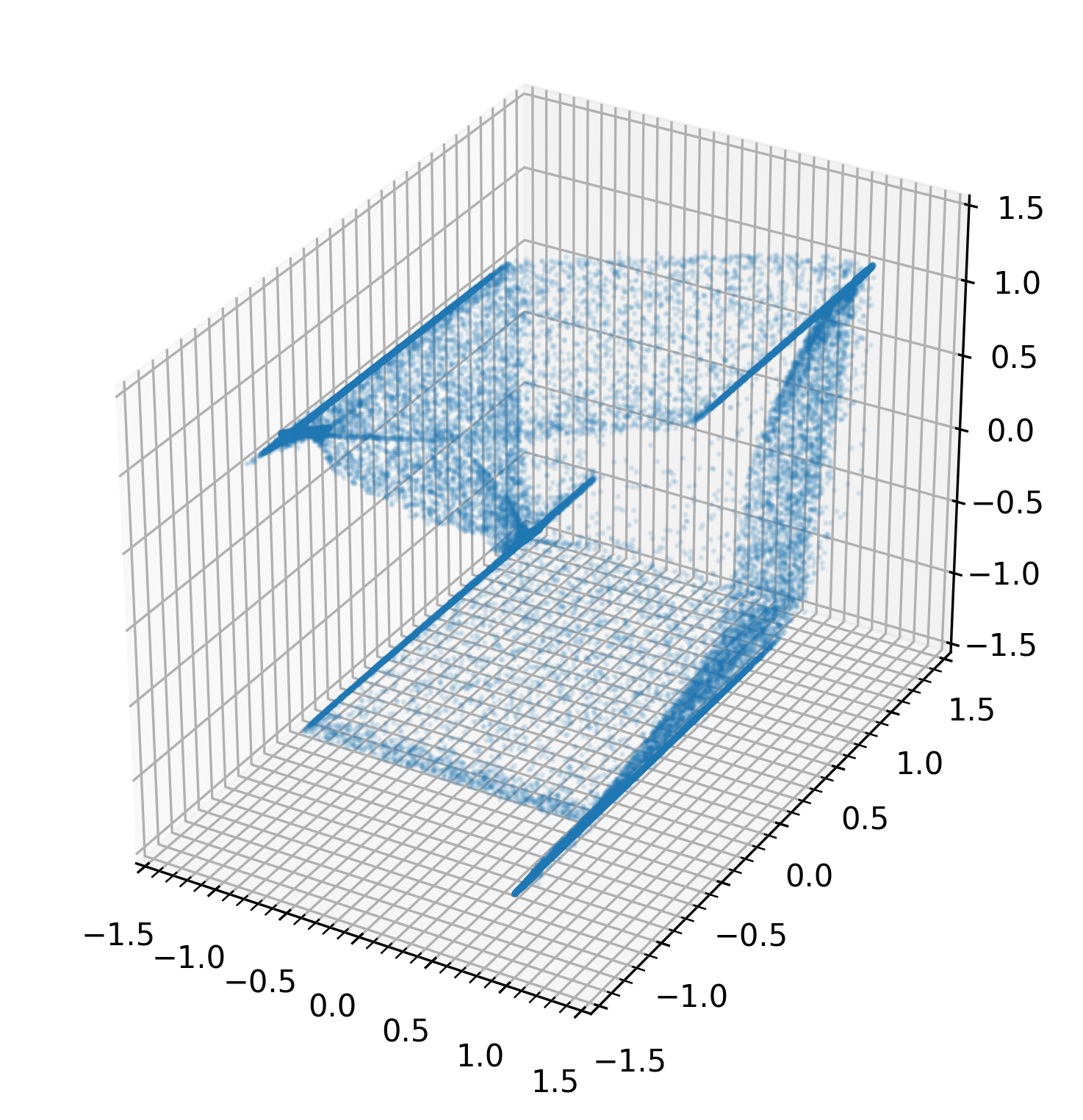}%
    \includegraphics[width=.4\linewidth]{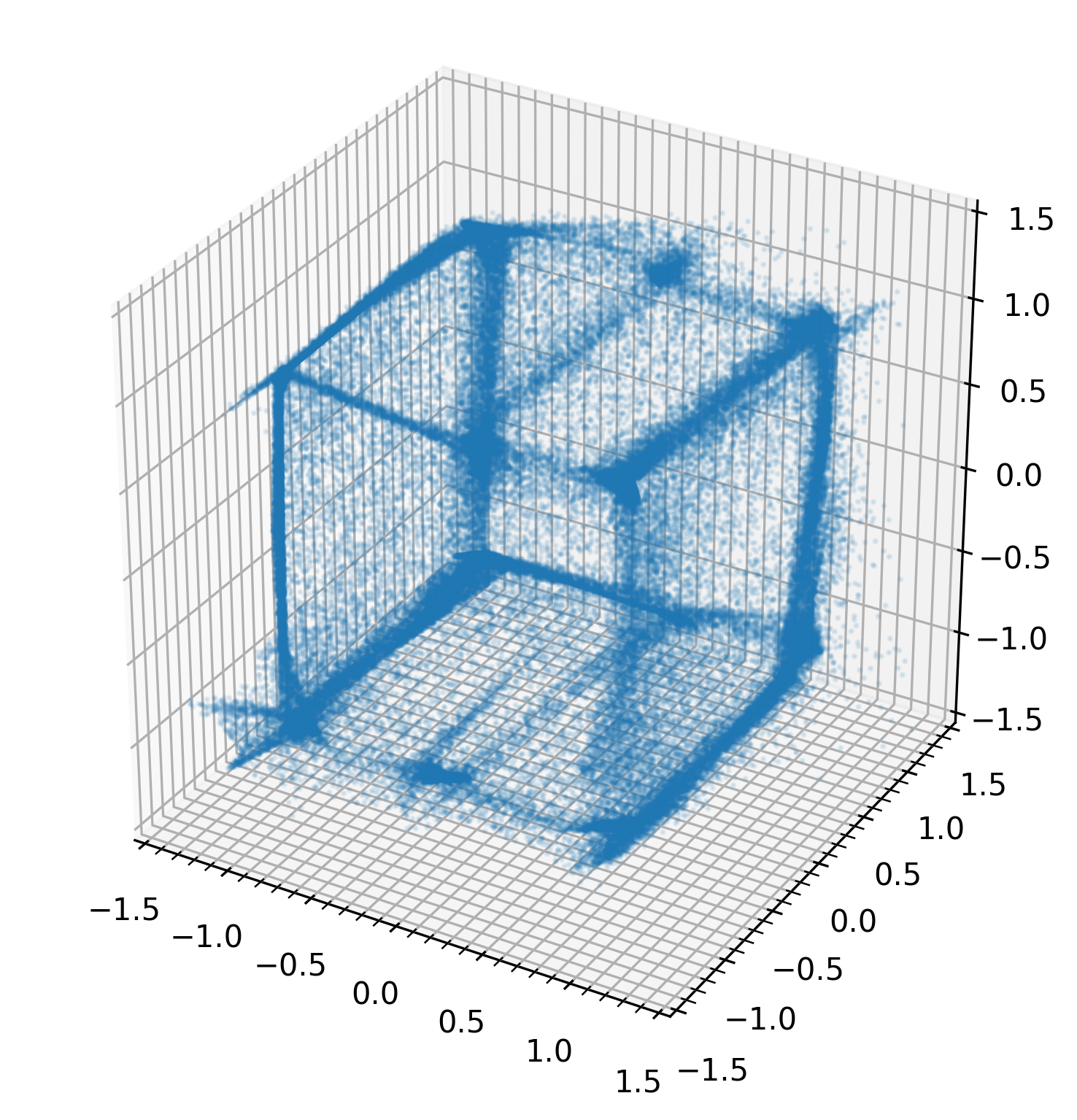}
    \includegraphics[width=.4\linewidth]{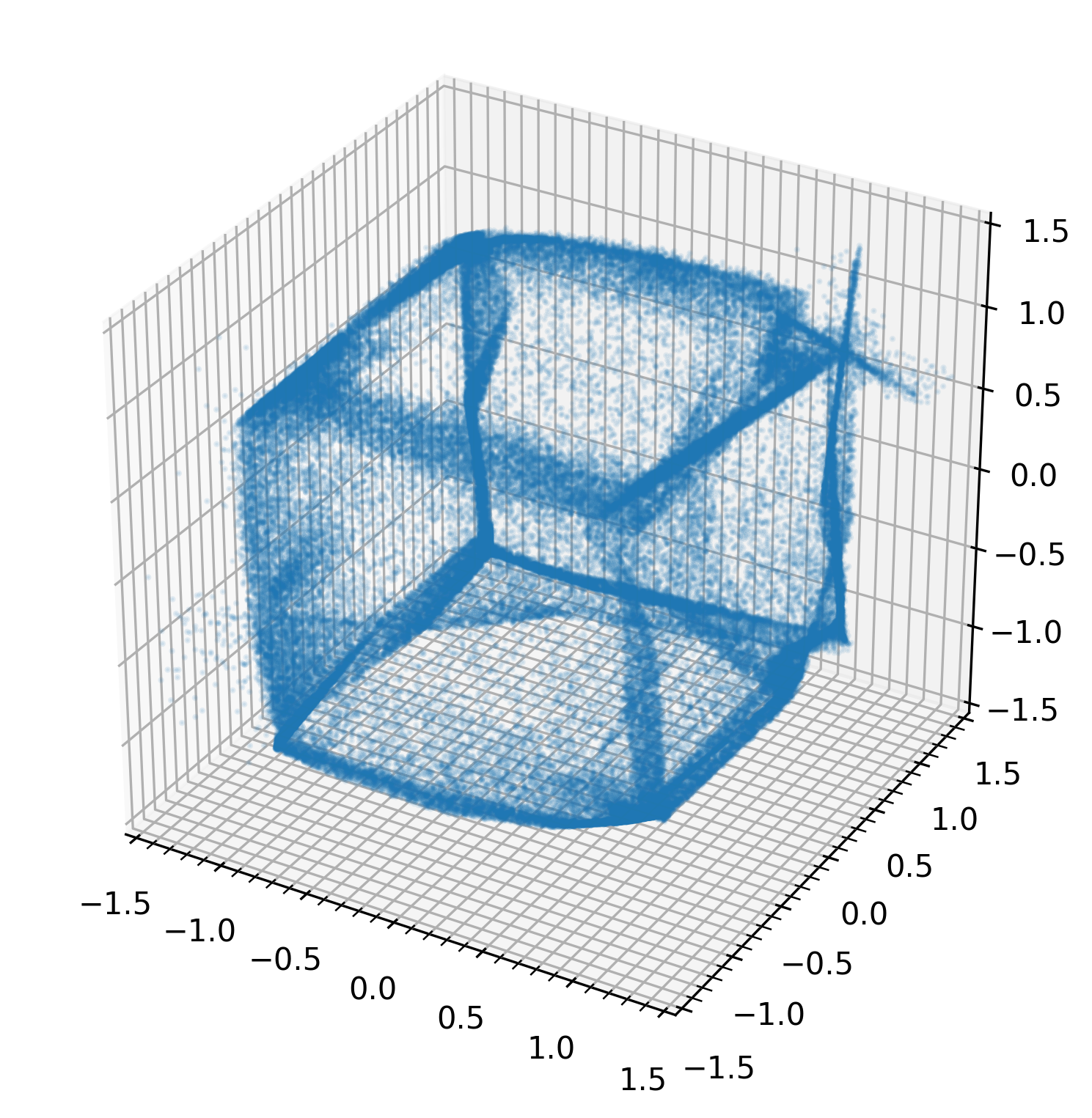}%
    \includegraphics[width=.4\linewidth]{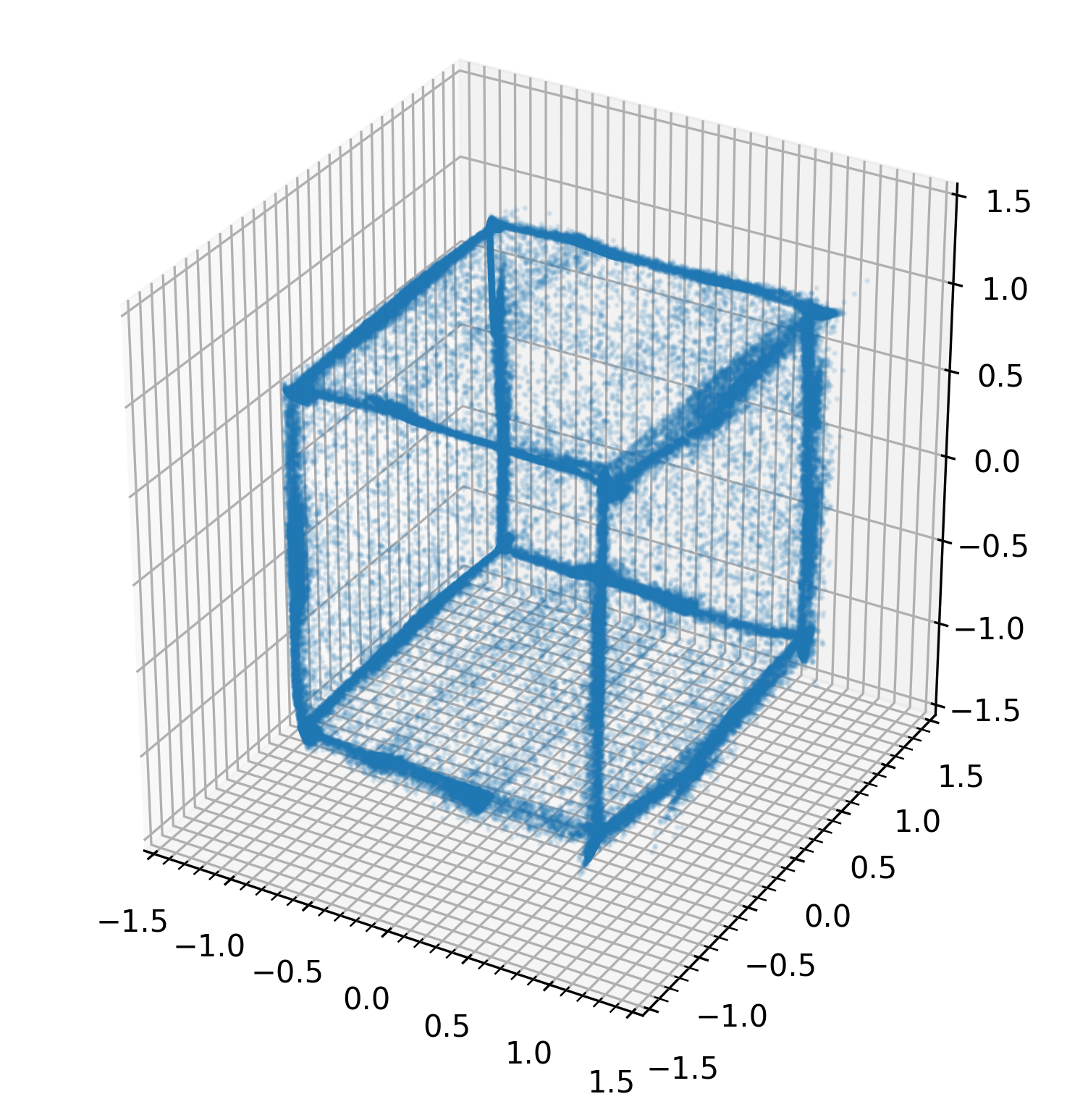}
    \includegraphics[width=.4\linewidth]{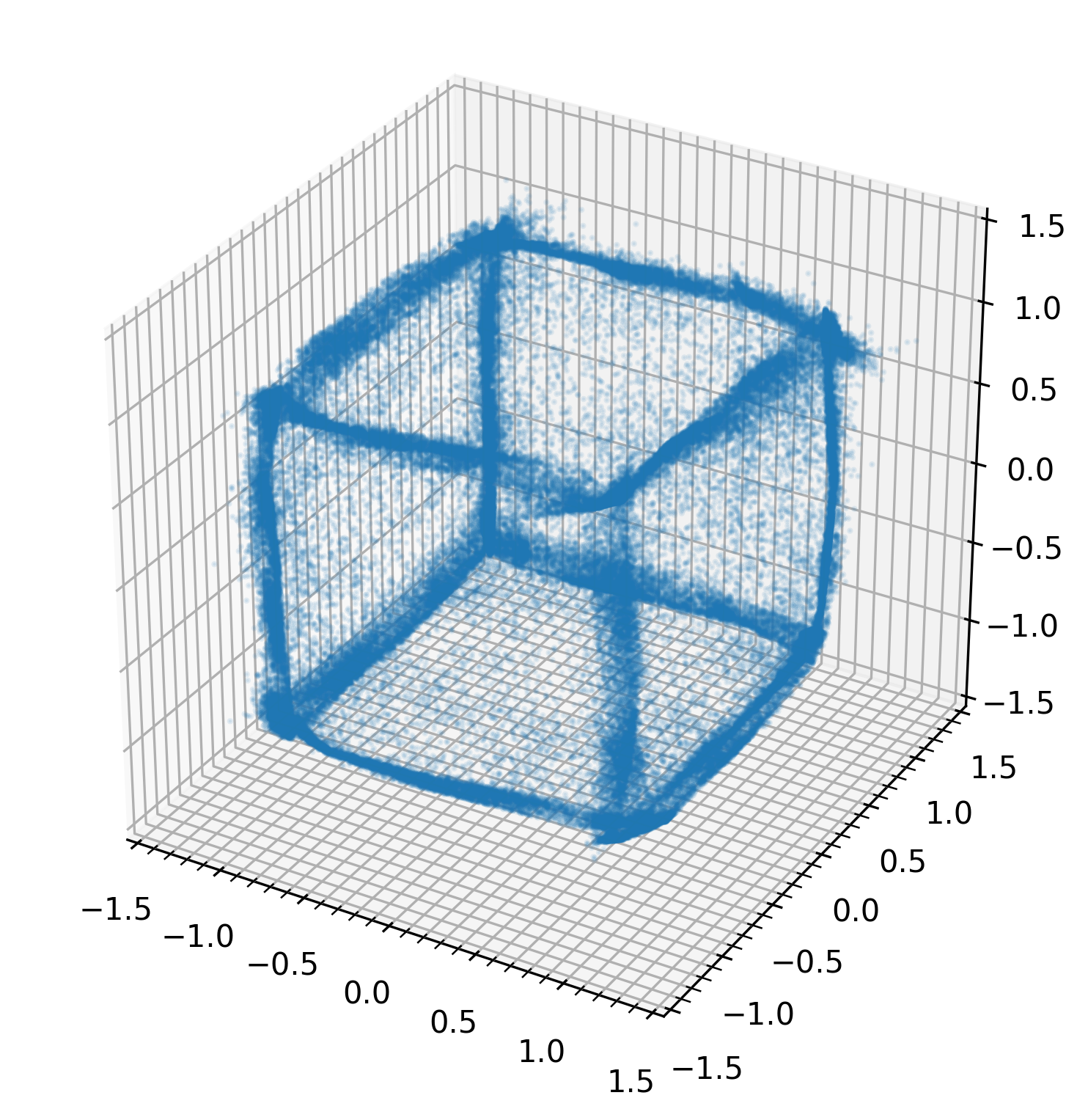}%
    \includegraphics[width=.4\linewidth]{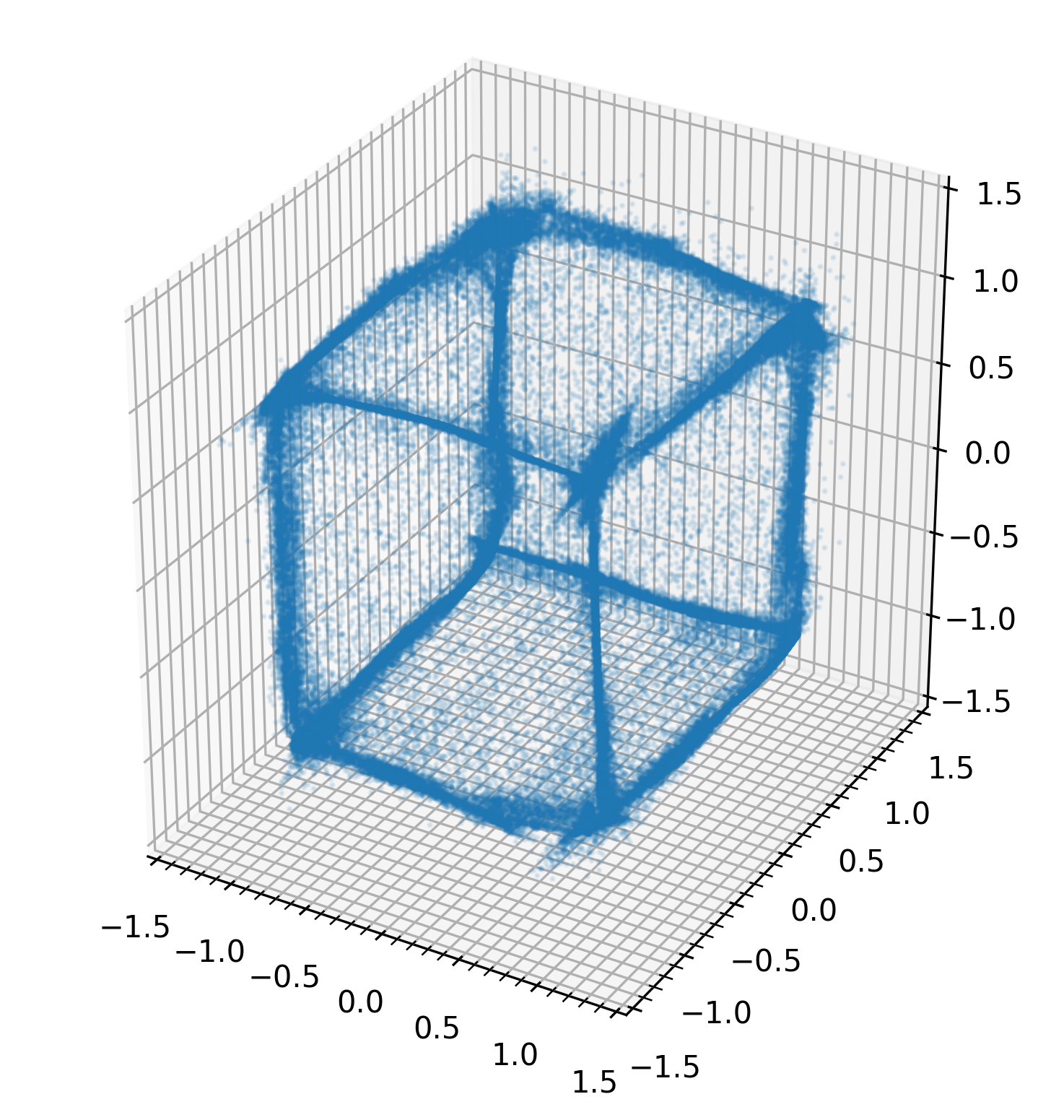}
    \caption{GAN (cube dataset).
    Left to right, top to bottom: Ground truth, SG, OP, EG, CO, SGA, GNI, LA, LOLA, EDA, BRF, BRE. SLA excluded due to divergence.}
    \label{fig:gan_images_cube}
\end{figure}

\begin{figure}
    \centering
    \includegraphics[width=.5\linewidth]{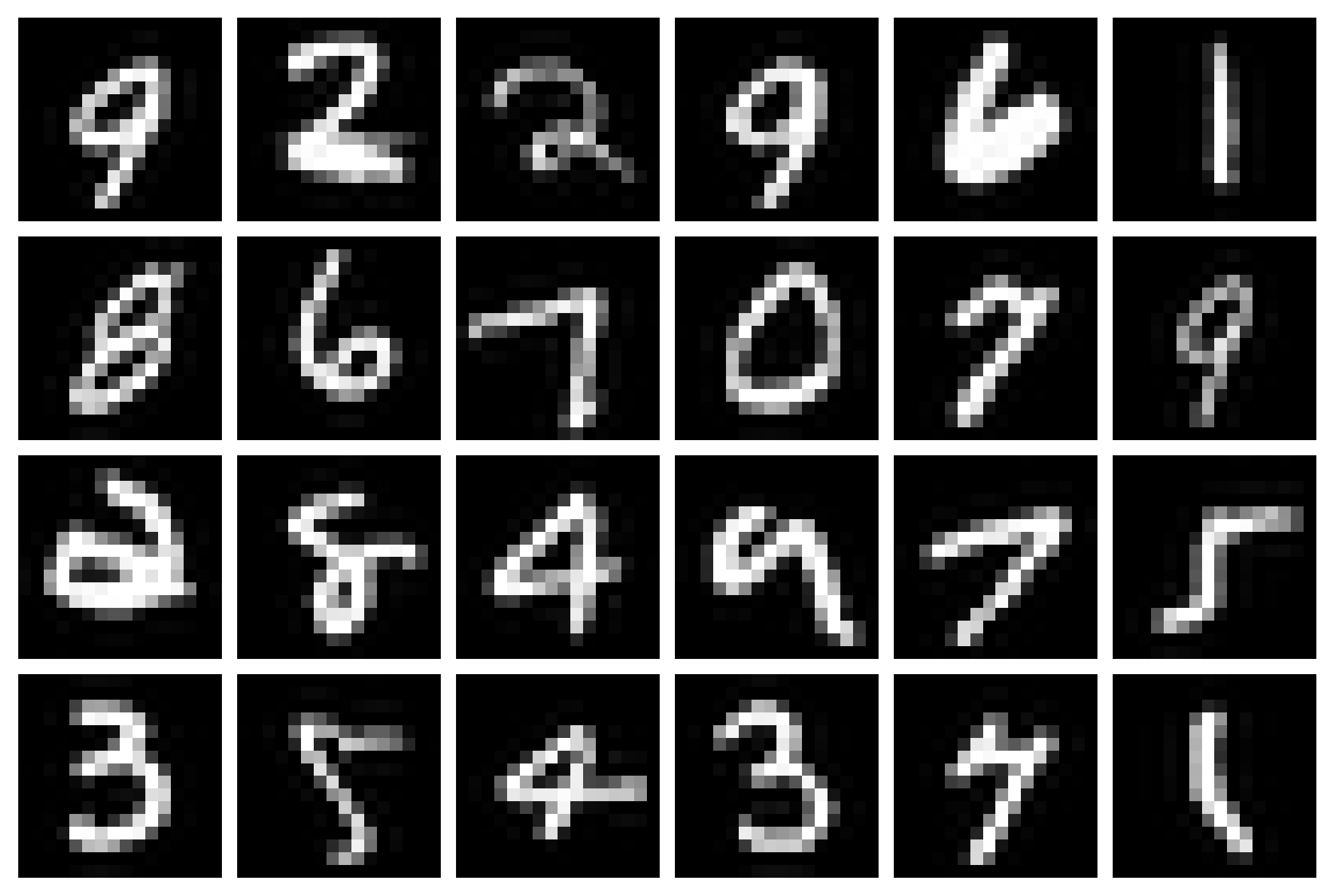}%
    \includegraphics[width=.5\linewidth]{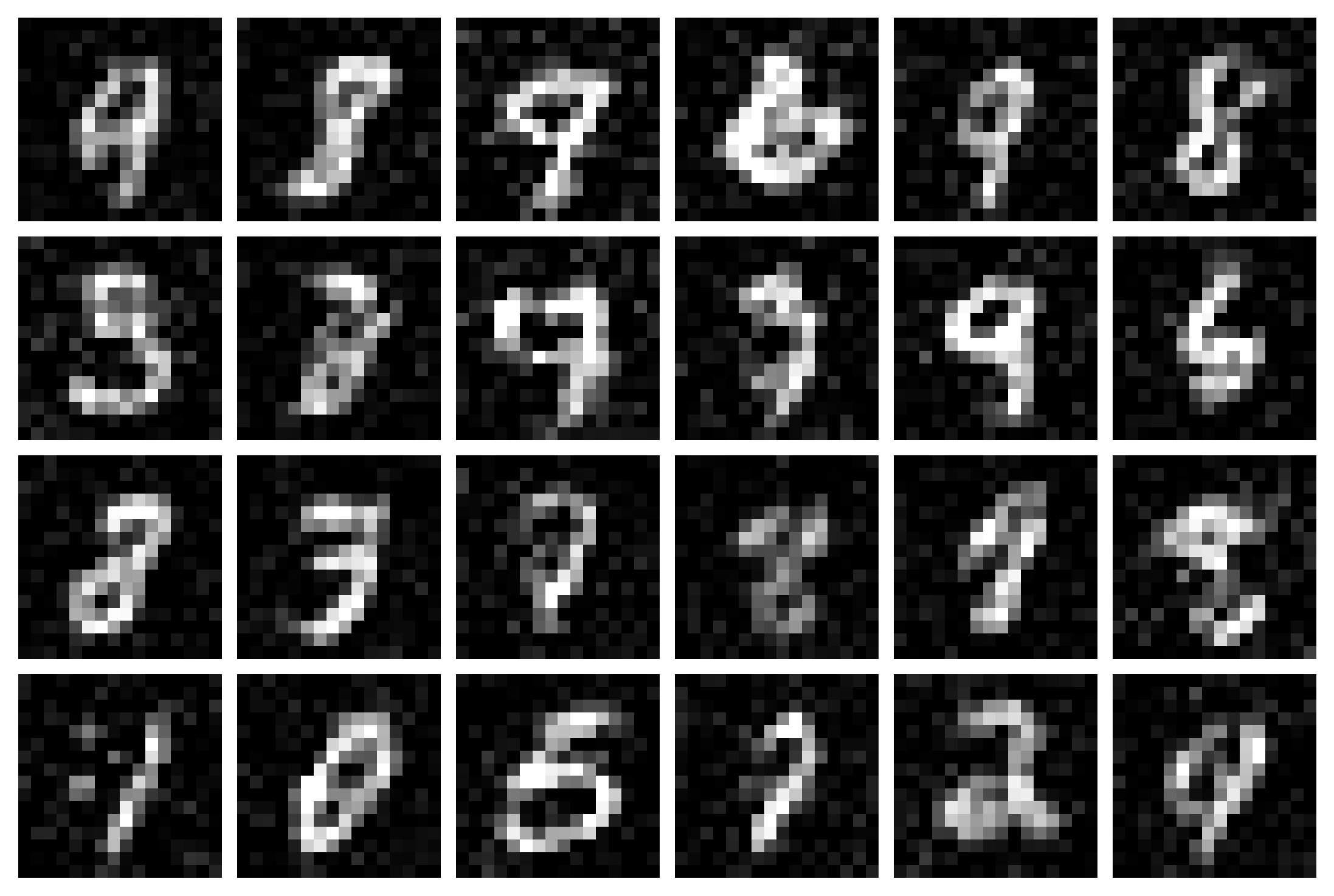}
    \includegraphics[width=.5\linewidth]{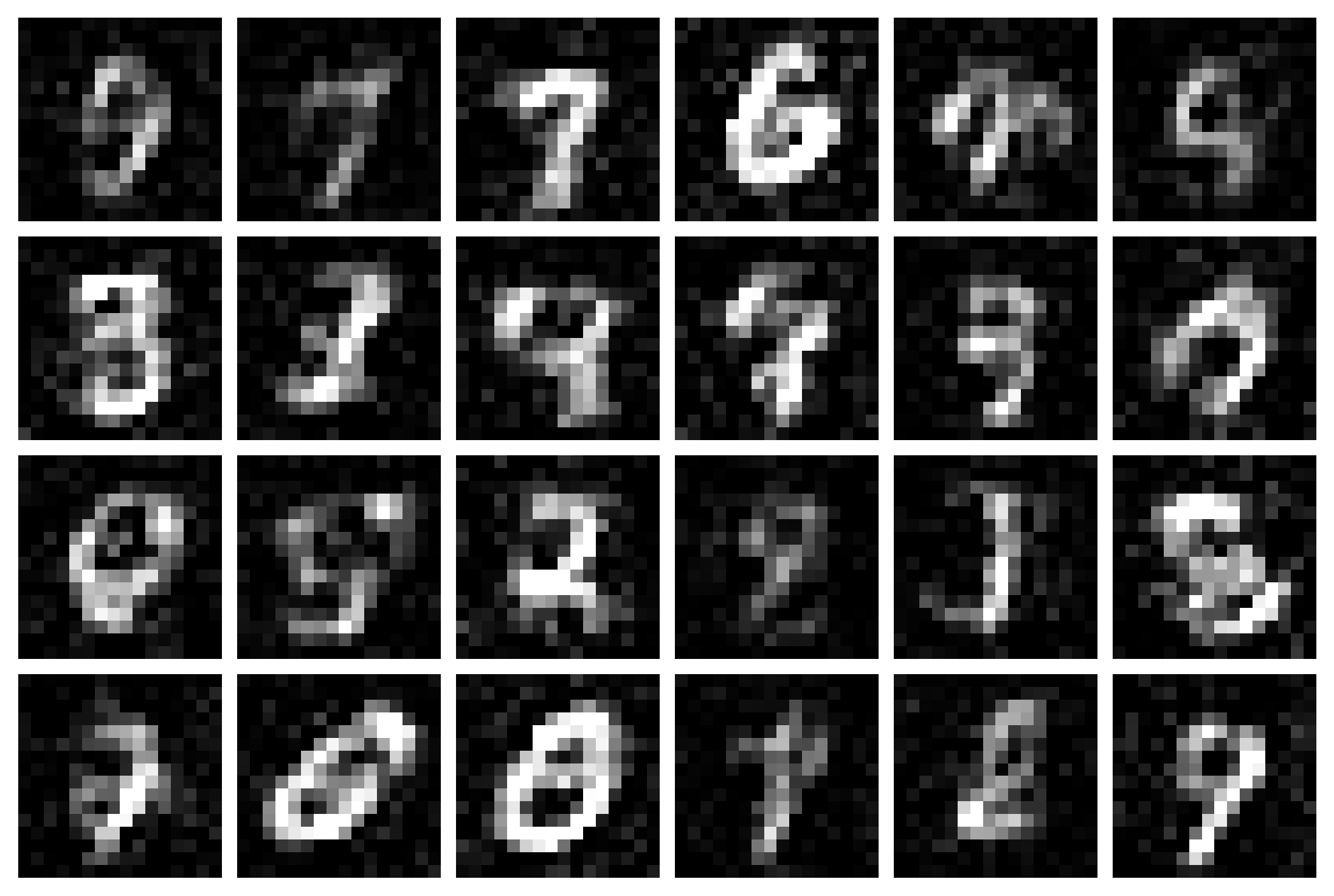}%
    \includegraphics[width=.5\linewidth]{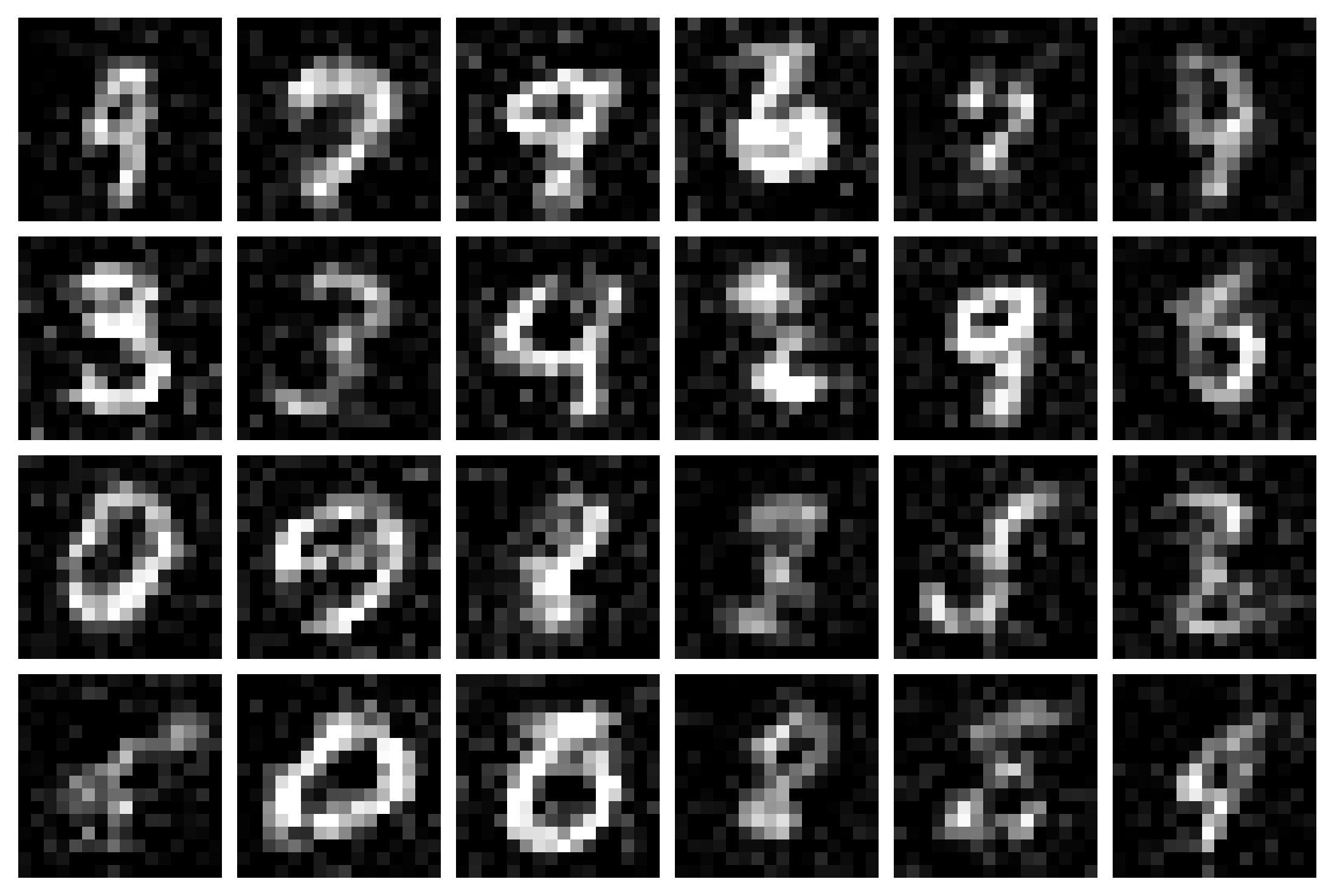}
    \includegraphics[width=.5\linewidth]{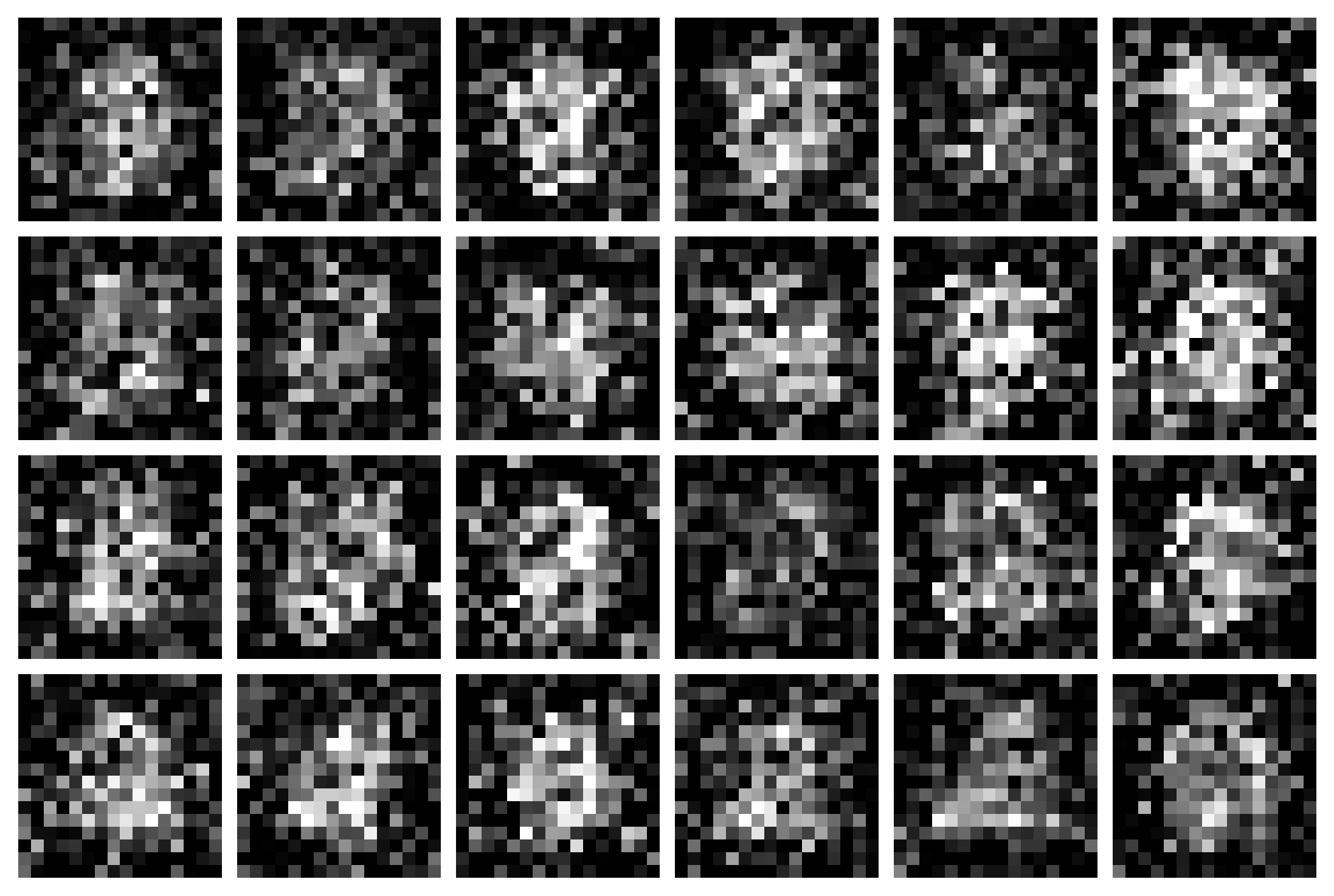}%
    \includegraphics[width=.5\linewidth]{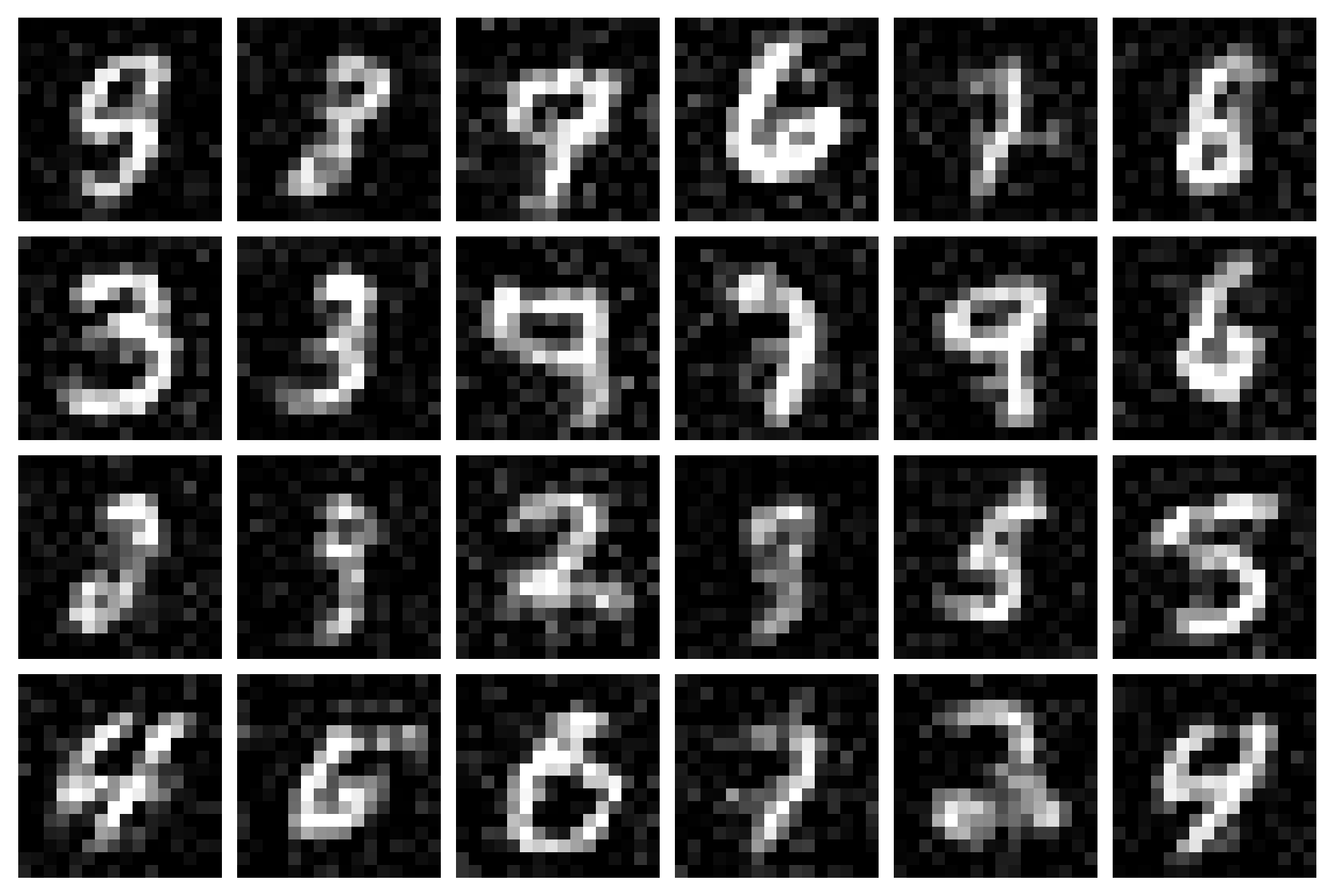}
    \includegraphics[width=.5\linewidth]{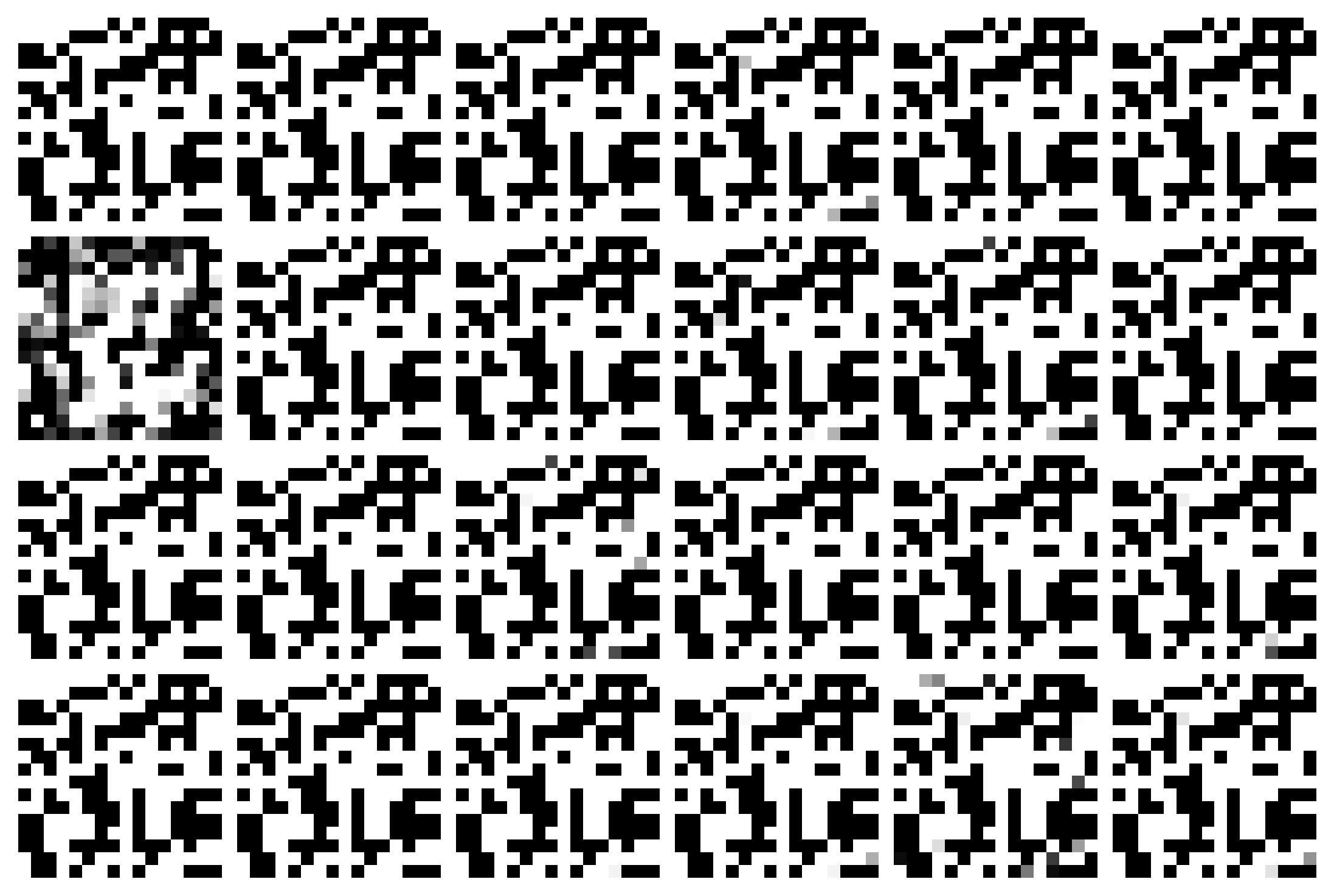}%
    \includegraphics[width=.5\linewidth]{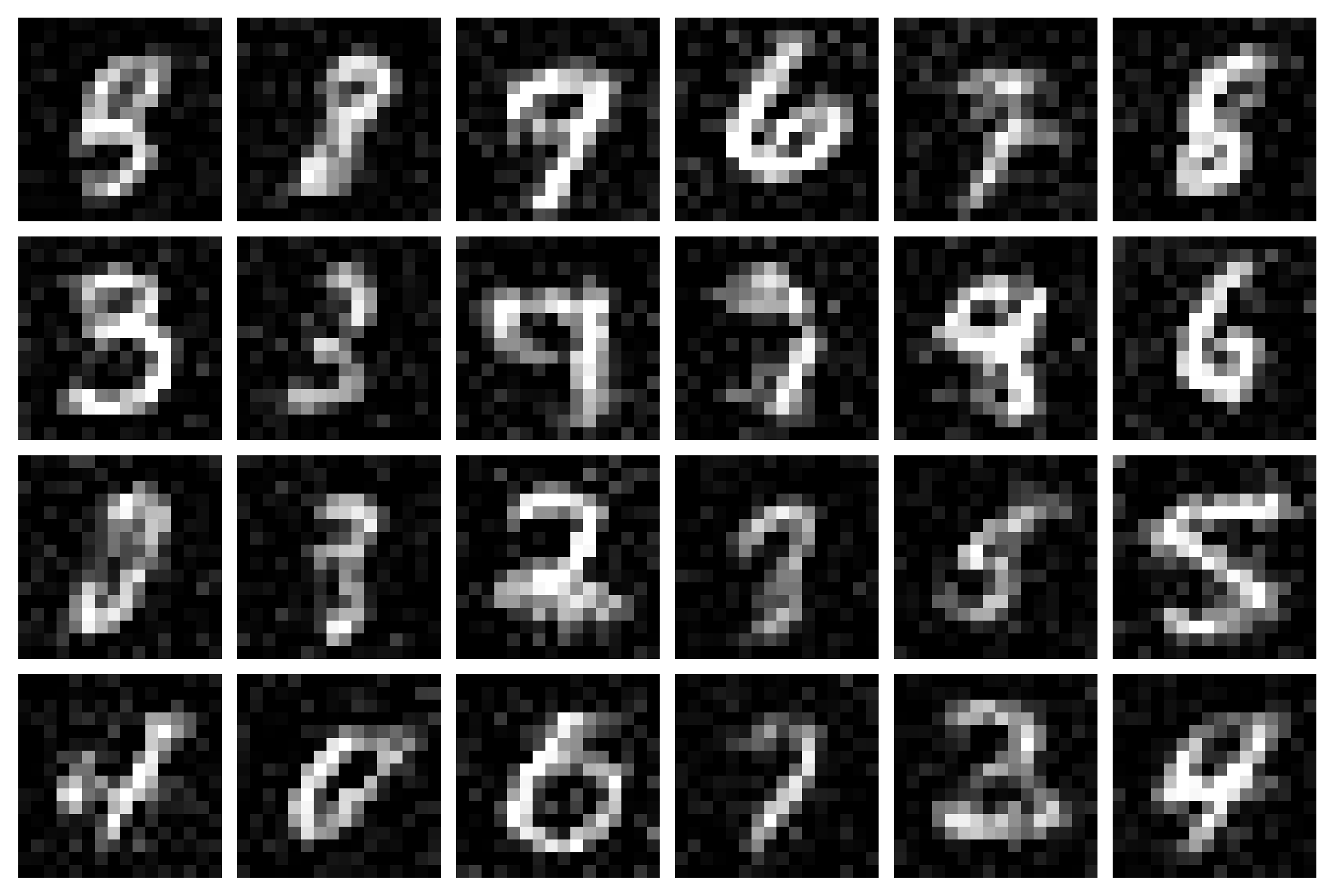}
    \includegraphics[width=.5\linewidth]{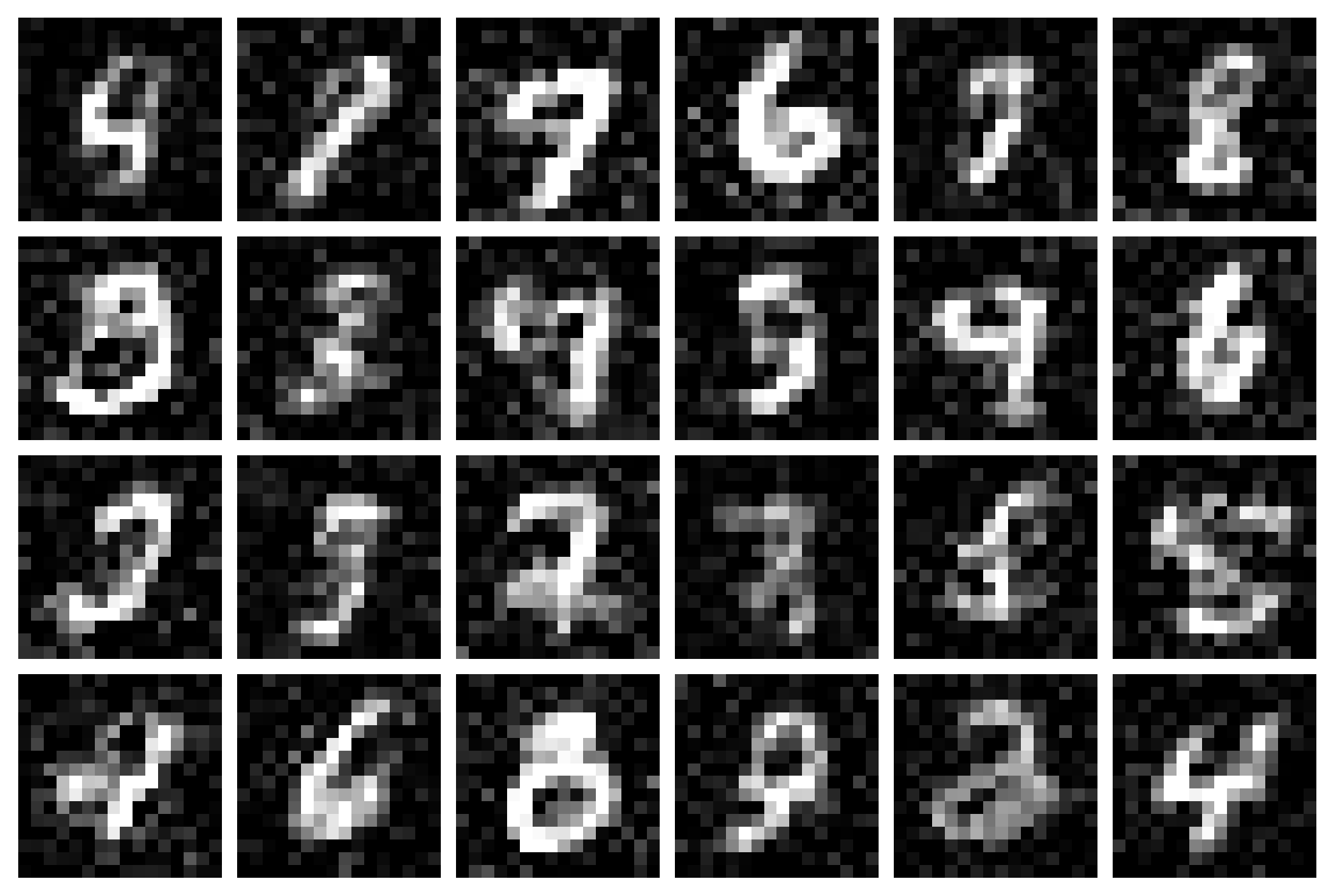}%
    \includegraphics[width=.5\linewidth]{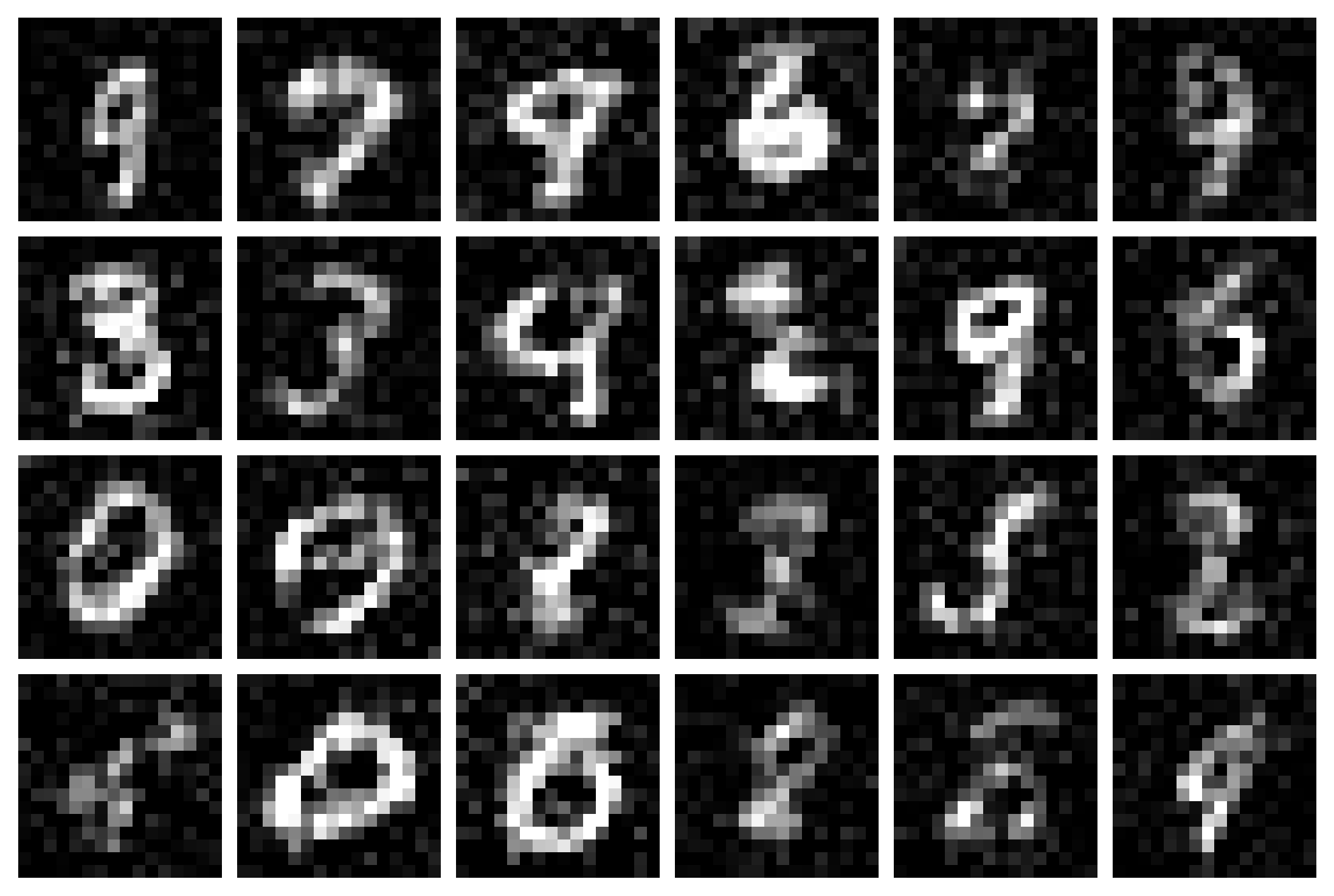}
    \includegraphics[width=.5\linewidth]{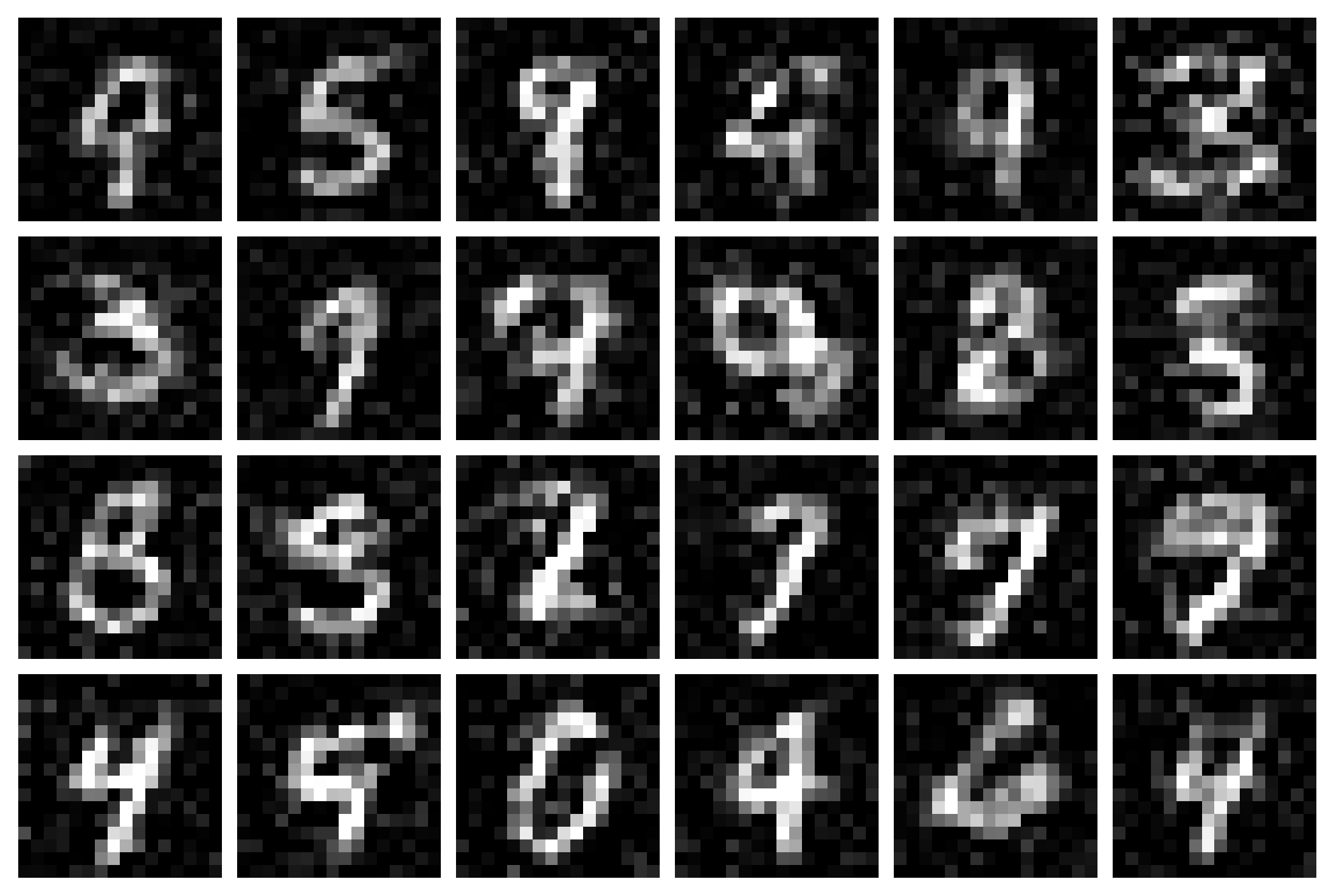}%
    \includegraphics[width=.5\linewidth]{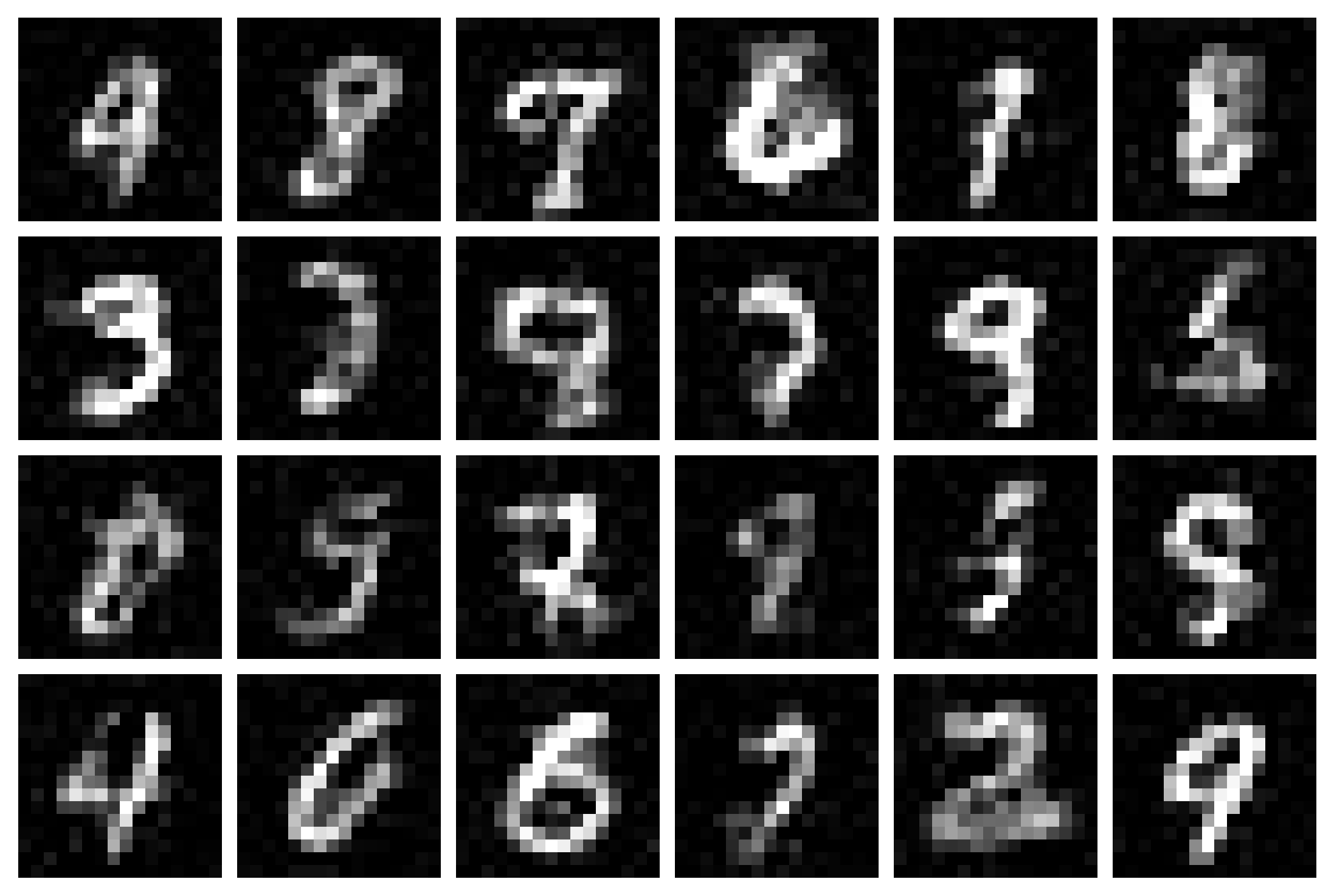}
    \caption{GAN (MNIST dataset). Left to right, top to bottom: Ground truth, SG, OP, EG, CO, SGA, GNI, LA, LOLA, EDA, BRF, BRE. SLA excluded due to divergence.}
    \label{fig:gan_images_mnist}
\end{figure}

\section{Theoretical analysis}
\label{sec:theoretical}

In this section, we present a theoretical analysis of our approach.
A general convergence proof is beyond the scope of this paper, though a potentially interesting question for future research. 
It is not uncommon in this field for methods to be introduced before theoretical guarantees are obtained.
Indeed, the latter is often difficult enough to stand alone as a research contribution. 
There have been many purely theoretical papers in the field of game solving on addressing such open theoretical questions with guarantees or lower bounds.
Furthermore, most of the great breakthroughs in AI game-playing lack theoretical guarantees for the technique that is actually used in practice, especially when they employ neural networks or abstraction techniques.
Theoretical analyses in the literature often make assumptions---such as linearity, (quasi)convexity, \emph{etc.}---that are not always satisfied in practice.

\citet{Lockhart_2019} analyzed exploitability descent in two-player, zero-sum, extensive-form games with finite action spaces.
As stated by \citet{goktas2022exploitability}, minimizing exploitability is a logical approach to computing NE, especially in cases where exploitability is convex, such as in the pseudo-games that result from replacing each player's utility function in a monotone pseudo-game by a second-order Taylor approximation \citep{flam1994noncooperative}.

\subsection{Convex exploitability}

In this subsection, we prove that the exploitability function is convex for certain classes of games.

\begin{definition}
Call a game ``regular'' if it has convex strategy sets and its utility function is of the form
\begin{align}
    u_i(x) = f_i(x_i) + \sum_{j \neq i} g(x_i, x_j)
\end{align}
where \(g_{ij}\) is convex in its second argument.
\end{definition}

\begin{theorem}
\label{thm:regular}
A constant-sum regular game has convex exploitability.
\end{theorem}
\begin{proof}
A sum of convex functions is convex.
A supremum of convex functions is convex.
Therefore,
\begin{align}
    \Phi(x)
    &= \sum_{i \in \mathcal{I}} \mleft( \sup_{y_i \in \mathcal{S}_i} u_i(y_i, x_{-i}) - u_i(x) \mright) \\
    &= \sum_{i \in \mathcal{I}} \sup_{y_i \in \mathcal{S}_i} u_i(y_i, x_{-i}) - \sum_{i \in \mathcal{I}} u_i(x) \\
    &= \sum_{i \in \mathcal{I}} \sup_{y_i \in \mathcal{S}_i} u_i(y_i, x_{-i}) + \text{const} \\
    &= \sum_{i \in \mathcal{I}} \sup_{y_i \in \mathcal{S}_i} \mleft( f_i(y_i) + \sum_{j \neq i} g(y_i, x_j) \mright) + \text{const} \\
    &= \sum_{i \in \mathcal{I}} \sup_{y_i \in \mathcal{S}_i} \mleft( \text{const} + \sum_{j \neq i} \text{convex} \mright) + \text{const} \\
    &= \sum_{i \in \mathcal{I}} \sup_{y_i \in \mathcal{S}_i} \text{convex} + \text{const} \\
    &= \sum_{i \in \mathcal{I}} \text{convex} + \text{const} \\
    &= \text{convex}
\end{align}
\end{proof}

\begin{definition}
A polymatrix game is a game with a utility function of the form
\begin{align}
    u_i(x) = \sum_{j \neq i} x_i^\top A_{ij} x_j
\end{align}
\end{definition}

These are graphical games in which each node corresponds to a player and each edge corresponds to a two-player bimatrix game between its endpoints.
Each player chooses a single strategy for all of its bimatrix games and receives the sum of the resulting payoffs.
In a \emph{constant-sum} polymatrix game, the sum of utilities across all players is constant.
As noted by \citet{cai2011minmax}, ``Intuitively, these games can be used to model a broad class of competitive environments where there is a constant amount of wealth (resources) to be split among the players of the game, with no in-flow or out-flow of wealth that may change the total sum of players’ wealth in an outcome of the game.''
They give an example of a ``Wild West'' game in which a set of gold miners need to transport gold by splitting it into wagons that traverse different paths, each of which may be controlled by thieves who could seize it.

\citet{cai2011minmax} prove a generalization of von Neumann's minmax theorem to constant-sum polymatrix games.
Their theorem implies convexity of equilibria, polynomial-time tractability, and convergence of no-regret learning algorithms to Nash equilibria.
\citet{cai2016zero} show that, in such games,  Nash equilibria can be found efficiently with linear programming.
They also show that the set of coarse correlated equilibria (CCE) collapses to the set of Nash equilibria.

We prove the following result.

\begin{theorem}
A constant-sum polymatrix game has convex exploitability.
\end{theorem}
\begin{proof}
A polymatrix game is a regular game where \(f_i(x_i) = 0\) and \(g_{ij}(x_i, x_j) = x_i^\top A_{ij} x_j\).
The latter is linear, and therefore convex, in its second argument.
Therefore, if the game is constant-sum, by Theorem \ref{thm:regular}, the exploitability is convex.
\end{proof}

\begin{corollary}
A pairwise constant-sum polymatrix game has convex exploitability.
\end{corollary}

\begin{corollary}
A two-player constant-sum matrix game has convex exploitability.
\end{corollary}

\begin{theorem}
A two-player constant-sum concave-convex game has convex exploitability.
\end{theorem}
\begin{proof}
In a two-player constant-sum game, the exploitability reduces to the so-called \textit{duality gap} \citep{Grnarova_2021}.
\begin{align}
    \Phi(x) = \sup_{x_1' \in \mathcal{S}_1} u_1(x_1', x_2) - \inf_{x_2' \in \mathcal{S}_2} u_1(x_1, x_2') + C
\end{align}
Here, \(u_1(x_1', x_2)\) is convex in \(x\).
Thus \(\sup_{x_1' \in \mathcal{S}_1} u_1(x_1', x_2)\) is convex in \(x\).
Also, \(u_1(x_1, x_2')\) is concave in \(x\).
Thus \(\inf_{x_2' \in \mathcal{S}_2} u_1(x_1, x_2')\) is concave in \(x\).
Thus \(\Phi(x) = \text{convex} - \text{concave} + \text{constant}\) in \(x\), which is convex in \(x\).
\end{proof}

\subsection{Subgradient descent}

We seek to minimize exploitability by performing subgradient descent.
This raises the question of when this process attains a global minimum.
\citet{kiwiel2004convergence} analyze the convergence of \emph{approximate subgradient methods} for convex optimization, and prove the following theorems.
Let \(\mathcal{S} \subseteq \mathbb{R}^n\) be a nonempty closed convex set,
\(f : \mathcal{S} \to \mathbb{R}\) be a closed proper convex function,
and \(\mathcal{S}_* = \operatorname{argmin} f\).
Let \(x_{t+1} = P_\mathcal{S}(x_t - \nu_t g_t)\) where \(P_\mathcal{S}\) is the projector onto \(\mathcal{S}\) (\(P_\mathcal{S}(x) \in \argmin_{y \in \mathcal{S}} \|x - y\|\)), \(\nu_t \geq 0\) is a stepsize, \(\varepsilon_t \geq 0\) is an error tolerance, and \(g_t \in \partial_{\varepsilon_t} f(x_t)\) is an \(\varepsilon_t\)-\emph{approximate subgradient} of \(f\) at \(x_t\), that is, \(f(x) \geq f(x_t) + \langle g_t, x - x_t \rangle - \varepsilon_t\) for all \(x\).

\begin{theorem}
\citep[Theorem 3.4]{kiwiel2004convergence}
Suppose \(\mathcal{S}_* \neq \varnothing\), \(\sum_{t \in \mathbb{N}} \nu_t = \infty\), and \(\sum_{t \in \mathbb{N}} \nu_t (\tfrac{1}{2} \|g_t\|^2 \nu_t + \varepsilon_t) < \infty\).
Then \(\{x_t\}_{t \in \mathbb{N}}\) converges to some \(x_\infty \in \mathcal{S}_*\).
\end{theorem}

\begin{theorem}
\citep[Theorem 3.6]{kiwiel2004convergence}
Suppose \(\mathcal{S}_* \neq \varnothing\), \(\sum_{t \in \mathbb{N}} \nu_t = \infty\), \(\sum_{t \in \mathbb{N}} \nu_t^2 < \infty\), \(\sum_{t \in \mathbb{N}} \nu_t \varepsilon_t < \infty\), and the subgradients do not grow too fast: \(\exists c < \infty . \forall t \in \mathbb{N} . \|g_t\|^2 \leq c (1 + \|x_t\|^2)\) (\emph{e.g.}, they are bounded).
Then \(\{x_t\}_{t \in \mathbb{N}}\) converges to some \(x_\infty \in \mathcal{S}_*\).
\end{theorem}

Convergence results are also known for subgradient methods on \emph{quasi}convex functions \citep{hu2015inexact}.

In our paper, we are trying to approximately minimize the exploitability function \(\Phi : \mathcal{S} \to \mathbb{R}\), \(\Phi(x) = \sup_{y \in \mathcal{S}} \phi(x, y)\), where \(\phi\) is the Nikaido--Isoda function and \(\mathcal{S}\) is the set of possible strategy profiles.
Specifically, we use subgradients of \(\tilde{\Phi}_t(x) = \phi(x, \tilde{y}_t)\), where \(\tilde{y}_t\) is the response profile output by the learned best-response ensembles (BRE) or best-response function (BRF) at \(t\).
Thus \(\tilde{\Phi}_t(x) \geq \tilde{\Phi}_t(x_t) + \langle g_t, x - x_t \rangle\).
Unconditionally, \(\Phi \geq \tilde{\Phi}_t\) (since the former maximizes over all possible \(y\)).
Thus \(\Phi(x) \geq \tilde{\Phi}_t(x_t) + \langle g_t, x - x_t \rangle\).
Now, suppose we can guarantee that \(\tilde{\Phi}(x_t) \geq \Phi(x_t) - \varepsilon_t\) for an error tolerance \(\varepsilon_t \geq 0\); that is, the responses output by the BRE/BRF at \(t\) do not perform \emph{too} badly (in the limit) compared to the true best responses.
Then \(\Phi(x) \geq \Phi(x_t) - \varepsilon_t + \langle g_t, x - x_t \rangle\).

Therefore, when the assumptions of the above theorems hold, \(\{x_t\}_{t \in \mathbb{N}}\) converges to a global minimizer of the exploitability function, which is an NE (if an NE exists at all).

\section{Additional related work}

\citet{McMahan_2003} introduced the double oracle algorithm for normal-form games and proved its convergence.~\citet{Adam_2021} extended it to two-player zero-sum continuous games.~\citet{Kroupa_2021} extended it to \(n\)-player continuous games.
Their algorithm maintains finite strategy sets for each player and iteratively extends them with best responses to an equilibrium of the induced finite sub-game.
This ``converges fast when the dimension of strategy spaces is small, and the generated subgames are not large.''
For example, in the two-player zero-sum case, ``The best responses were computed by selecting the best point of a uniform discretization for the one-dimensional problems and by using a mixed-integer linear programming reformulation for the Colonel Blotto games.''

Our best-response ensembles method has some resemblance to double oracle algorithms, including ones for continuous games~\citep{Adam_2021,Kroupa_2021} as well as PSRO~\citep{Lanctot17:Unified}, XDO~\citep{McAleer_2021}, Anytime PSRO~\citep{McAleer_2022}, and Self-Play PSRO~\citep{McAleer_2022b}.
Double oracle algorithms maintain a set of strategies that is expanded on each iteration with approximate best responses to the meta-strategies of the other players.
These strategies are \emph{static}, that is, they do not change after they are added.
In contrast, our algorithm \emph{dynamically} improves the elements of a best-response ensemble during training.
Thus the ensemble does not need to be grown with each iteration, but improves autonomously over time.

\citet{Ganzfried_2021} introduced an algorithm for approximating equilibria in continuous games called ``redundant fictitious play'' and applied it to a continuous Colonel Blotto game.
\citet{Kamra_2019} presented DeepFP, an approximate extension of fictitious play~\citep{Brown51:Iterative,berger_2007} to continuous action spaces.
They demonstrate stable convergence to equilibrium on several classic games and a large forest security domain.
DeepFP represents players' approximate best responses via generative neural networks, which are highly expressive implicit density approximators.
The authors state that, because implicit density models cannot be trained directly, they employ a game-model network that is a differentiable approximation of the players' utilities given their actions, and train these networks end-to-end in a model-based learning regime.
This allows working in the absence of gradients for players.

\citet{Li_2021} extended the double oracle approach to \(n\)-player general-sum continuous Bayesian games.
They represent agents as neural networks and optimize them using \emph{natural evolution strategies (NES)}~\citep{Wierstra_2008, Wierstra_2014}.
For pure equilibrium computation, they formulate the problem as a bi-level optimization and employ NES to implement both inner-loop best-response optimization and outer-loop regret minimization.~\citet{Bichler_2021} presented a learning method that represents strategies as neural networks and applies simultaneous gradient dynamics to provably learn local equilibria.~\citet{Fichtl_2022} compute distributional strategies on a discretized version of the game via online convex optimization, specifically \emph{simultaneous online dual averaging (SODA)}, and show that the equilibrium of the discretized game approximates an equilibrium in the continuous game.

In a \emph{generative adversarial network (GAN)}~\citep{goodfellow2020generative}, a generator learns to generate fake data while a discriminator learns to distinguish it from real data.
\cite{Metz_2016} introduced a method to stabilize GANs by defining the generator objective with respect to an unrolled optimization of the discriminator.
They show how this technique solves the common problem of mode collapse, stabilizes training of GANs with complex recurrent generators, and increases diversity and coverage of the data distribution by the generator.
\citet{Grnarova_2019} proposed using an approximation of the game-theoretic \emph{duality gap} as a performance measure for GANs.~\citet{Grnarova_2021} proposed using this measure as the objective, proving some convergence guarantees.

\citet{Lockhart_2019} presented \emph{exploitability descent}, which computes approximate equilibria in two-player zero-sum extensive-form games by direct strategy optimization against worst-case opponents.
They prove that the exploitability of a player's strategy converges asymptotically to zero.
Hence, when both players employ this optimization, the strategy profile converges to an equilibrium.
Unlike extensive-form fictitious play~\citep{Heinrich_2015} and counterfactual regret minimization~\citep{Zinkevich07:Regret}, their convergence pertains to the strategies being optimized rather than the time-average strategies.~\citet{Timbers_2022} introduced approximate exploitability, which uses an approximate best response computed through search and reinforcement learning.
This is a generalization of \emph{local best response}, a domain-specific evaluation metric used in poker~\citep{Lisy17:Equilibrium}.

\citet{Fiez_2022} consider minimax optimization \(\min_x \max_y f(x, y)\) in the context of two-player zero-sum games, where the min-player (controlling \(x\)) tries to minimize \(f\) assuming the max-player (controlling \(y\)) then tries to maximize it.
In their framework, the min-player plays against \emph{smooth algorithms} deployed by the max-player (instead of full maximization, which is generally NP-hard).
Their algorithm is guaranteed to make monotonic progress, avoiding limit cycles or diverging behavior, and finds an appropriate ``stationary point'' in a polynomial number of iterations.

This work has important differences to ours.
First, our work tackles multi-player general-sum games, a more general class of games than two-player zero-sum games.
Second, their work does not use learned best-response functions, but instead runs a multi-step optimization procedure for the opponent on every iteration, with the opponent parameters re-initialized from scratch.
This can be expensive for complex games, which may require many iterations to learn a good opponent strategy.
It also does not reuse information from previous iterations to recommend a good response.
Our learned best-response functions can retain information from previous iterations and do not require a potentially expensive optimization procedure on each iteration.

\citet{Gemp_2022} proposed an approach called \emph{average deviation incentive descent with adaptive sampling} that iteratively improves an approximation to an NE through joint play by tracing a homotopy that defines a continuum of equilibria for the game regularized with decaying levels of entropy.
To encourage iterates to remain near this path, they minimize average deviation incentive via stochastic gradient descent.

\citet{Ganzfried10:Computing} presented a procedure for solving large imperfect information games by solving an infinite approximation of the original game and mapping the equilibrium to a strategy profile in the original game.
Counterintuitively, it is often the case that the infinite approximation can be solved much more easily than the finite game.
The algorithm exploits some qualitative model of equilibrium structure as an additional input in order to find an equilibrium in continuous games.

\citet{Mazumdar_2020} analyze the limiting behavior of competitive gradient-based learning algorithms using dynamical systems theory.
They characterize a non-negligible subset of the local NE that will be avoided if each agent employs a gradient-based learning algorithm.

\citet{Mertikopoulos_2019} examined the convergence of no-regret learning in games with continuous action sets, focusing on learning via ``dual averaging'', a widely used class of no-regret learning schemes where players take small steps along their individual utility gradients and then ``mirror'' the output back to their action sets.
They introduce the notion of variational stability, and show that stable equilibria are locally attracting with high probability whereas globally stable equilibria are globally attracting with probability 1.

\citet{Fiez_2019} investigated the convergence of learning dynamics in Stackelberg games with continuous action spaces.
They characterize conditions under which attracting critical points of simultaneous gradient descent are Stackelberg equilibria in zero-sum games.
They develop a gradient-based update for the leader while the follower employs a best response strategy for which each stable critical point is guaranteed to be a Stackelberg equilibrium in zero-sum games.
As a result, the learning rule provably converges to a Stackelberg equilibria given an initialization in the region of attraction of a stable critical point.
They then consider a follower employing a gradient-play update rule instead of a best response strategy and propose a two-timescale algorithm with similar asymptotic convergence guarantees.

While most previous work on minimax optimization focused on classical notions of equilibria from simultaneous games, where the min-player and the max-player act simultaneously, \citet{Jin_2020} proposed a mathematical definition of local optimality in sequential game settings, which include GANs and adversarial training.
Due to the nonconvex-nonconcave nature of the problems, minimax is in general not equal to maximin, so the order in which players act is crucial.

\citet{Wang_2020} proposed an algorithm for two-player zero-sum sequential games called \emph{Follow-the-Ridge (FR)} that provably converges to and only converges to local minimax, addressing the rotational behaviour of ordinary gradient dynamics.

\citet{Tsaknakis_2021} proposed an algorithm for finding the FNEs of a two-player zero-sum game, in which the local cost functions can be non-convex, and the players only have access to local stochastic gradients.
The proposed approach is based on reformulating the problem of interest as minimizing the \emph{Regularized Nikaido-Isoda (RNI)} function.
Unlike ours, this work tackles two-player zero-sum games only.
Furthermore, it requires nontrivial subroutines.
For example, in the description of the algorithm, the authors state ``We
assume that these subproblems are solved to a given accuracy using known methods, such as the projected gradient descent method.''

\citet{Willi_2022} showed that the original formulation of the LOLA method (and follow-up work) is inconsistent in that it models other agents as naive learners rather than LOLA agents.
In previous work, this inconsistency was suggested as a cause of LOLA's failure to preserve stable fixed points (SFPs).
They formalize consistency and show that \emph{higher-order LOLA (HOLA)} solves LOLA's inconsistency problem if it converges.
They also proposed a new method called \emph{consistent LOLA (COLA)}, which learns update functions that are consistent under mutual opponent shaping.
It requires no more than second-order derivatives and learns consistent update functions even when HOLA fails to converge.

\citet{Perolat_2022} introduced \emph{DeepNash}, an autonomous agent capable of learning to play the imperfect information game Stratego from scratch, up to a human expert level.
DeepNash uses a game-theoretic, model-free deep reinforcement learning method, without search, that learns to master Stratego via self-play.
The \emph{Regularised Nash Dynamics (R-NaD)} algorithm, a key component of DeepNash, converges to an approximate NE, instead of ``cycling'' around it, by directly modifying the underlying multiagent learning dynamics.
\citet{Qin_2022} proposed a no-regret style reinforcement learning algorithm \emph{PORL} for continuous action tasks, proving that it has a last-iterate convergence guarantee.

\citet{bao2022finding} proposed \emph{double Follow-the-Ridge (double-FTR)} an algorithm with local convergence guarantee to differential NE in general-sum two-player differential games.
Whereas they focus on two-player games, we are interested in methods that tackle general \(n\)-player games.

\citet{goktas2022gradient} studied min-max games with dependent strategy sets, where the strategy of the first player constrains the behavior of the second.
They introduced two variants of gradient descent ascent (GDA) that assume access to a solution oracle for the optimal Karush Kuhn Tucker (KKT) multipliers of the games' constraints, and proved a convergence guarantee.

\subsection{Existence and uniqueness of Nash equilibria}

Every finite game has a mixed strategy NE.
This is seminal result in game theory proven by~\citet{Nash50:Non}.
Beyond this theorem, the following theorems apply to games with infinite strategy spaces \(\mathcal{X}_i\):
If for all \(i\), \(\mathcal{X}_i\) is nonempty and compact, and \(u(x)_i\) is continuous in \(x\), a mixed strategy NE exists~\citep{Glicksberg52:Further}.
If for all \(i\), \(\mathcal{X}_i\) is nonempty, compact, and convex, and \(u(x)_i\) is continuous in \(x\) and quasi-concave in \(x_i\), a pure-strategy NE exists~\citep[p. 34]{Fudenberg91:Gamea}.
Other results include the existence of a mixed-strategy NE for games with discontinuous utilities under some mild semicontinuity conditions on the utility functions~\citep{Dasgupta86:Existence}, and the uniqueness of a pure-strategy NE for continuous games under diagonal strict concavity assumptions~\citep{Rosen_1965}.

\subsection{League training}

\citet{Vinyals19:Grandmaster} tackle StarCraft II, a real-time strategy game that has become a popular benchmark for artificial intelligence.
To address the game-theoretic challenges, they introduce \emph{league training}, an algorithm for multiagent reinforcement learning.
Self-play algorithms learn rapidly but may chase cycles indefinitely without making progress.
\emph{Fictitious self-play (FSP)}~\citep{Heinrich_2015} avoids cycles by computing a best response against a uniform mixture of all previous policies; the mixture converges to an NE in 2-player zero-sum games.
They extend this approach to compute a best response against a non-uniform mixture of opponents.
This league of potential opponents includes a diverse range of agents, as well as their policies from both current and previous iterations.
At each iteration, each agent plays games against opponents sampled from a mixture policy specific to that agent.
The parameters of the agent are updated from the outcomes of those games by an actor–critic reinforcement learning procedure.

The league consists of three distinct types of agent, differing primarily in their mechanism for selecting the opponent mixture.
First, the main agents utilize a \emph{prioritized fictitious self-play (PFSP)} mechanism that adapts the mixture probabilities proportionally to the win rate of each opponent against the agent; this provides our agent with more opportunities to overcome the most problematic opponents.
With fixed probability, a main agent is selected as an opponent; this recovers the rapid learning of self-play.
Second, main exploiter agents play only against the current iteration of main agents.
Their purpose is to identify potential exploits in the main agents; the main agents are thereby encouraged to address their weaknesses.
Third, league exploiter agents use a similar PFSP mechanism to the main agents, but are not targeted by main exploiter agents.
Their purpose is to find systemic weaknesses of the entire league.
Both main exploiters and league exploiters are periodically reinitialized to encourage more diversity and may rapidly discover specialist strategies that are not necessarily robust against exploitation.

\subsection{Ensemble GANs}

\citet{Moghadam_2021} survey various GAN variants that use multiple generators and/or discriminators.
The variants with one generator and multiple discriminators are
GMAN~\citep{Durugkar_2017},
D2GAN~\citep{Nguyen_2017},
FakeGAN~\citep{Aghakhani_2018},
MD-GAN~\citep{Hardy_2019},
DDL-GAN~\citep{Jin_2020},
and Microbatch-GAN~\citep{Mordido_2020}.
The variants with multiple generators and one discriminators are
MPM-GAN~\citep{Ghosh_2016},
MAD-GAN~\citep{Ghosh_2018},
M-GAN~\citep{Hoang_2018},
Stackelberg-GAN~\citep{Zhang_2018},
and MADGAN~\citep{Ke_2022}.
The variants with multiple generators and multiple discriminators are
MIX+GAN~\citep{Arora_2017},
FedGAN~\citep{Rasouli_2020},
and another version of MADGAN~\citep{Ke_2022}.

For GMAN, the generator G trains using feedback aggregated over multiple discriminators under a function F.
If F = max, G trains against the best discriminator.
If F = mean, G trains against a uniform ensemble.
The authors note that training against a far superior discriminator can impede the generator's learning.
This is because the generator is unlikely to generate any samples considered ``realistic'' by the discriminator's standards, and so the generator will receive uniformly negative feedback.
For this reason, they explore alternative functions that soften the max operator, including one parameterized by \(\beta\) in such a way that \(\beta = 0\) yields the mean and \(\beta \to \infty\) yields the max: \(f_\beta(a) = \sum_i a_i w_i\) where \(w = \operatorname{softmax}(\beta a)\).
At the beginning of training, one can set \(\beta\) closer to zero to use the mean, increasing the odds of providing constructive feedback to the generator.
In addition, the discriminators have the added benefit of functioning as an ensemble, reducing the variance of the feedback presented to the generator, which is especially important when the discriminators are far from optimal and are still learning a reasonable decision boundary.
As training progresses and the discriminators improve, one can increase \(\beta\) to become more critical of the generator for more refined training.

The authors also explore an approach that enables the generator to automatically temper the performance of the discriminator when necessary, but still encourages the generator to challenge itself against more accurate adversaries.
This is done by augmenting the generator objective in a way that incentivizes it to increase \(\beta\) to reduce its objective, at the expense of competing against the best available adversary.

\section{Future research}
\label{sec:future}

Here, we discuss some possible extensions of our methods and directions for future research.

\paragraph{Analyzing strategies via samples}

If a player's strategy takes the form of a complex deep generative model, it might be difficult for the best-response function to ``make sense of'' its parameters as input.
However, what really matters from the perspective of best response computation is the distribution the model \emph{generates} on the action space, that is, its \emph{extrinsic} behavior.
Thus, we could instead feed a \emph{batch of action samples} from the model to the best-response function, letting the latter ``analyze'' the former in some fashion.
If the batch size is large enough, the best-response function could identify the features of the distribution that are relevant to recommending a good best response.
The batch of samples can be concatenated together and fed as input into a neural network.
This is the approach taken by PacGAN~\citep{Lin_2020} in the context of GANs, under the name of ``packing''.
A similar approach is also described in \citet[\S 3.2]{salimans2016improved} under the name of ``minibatch discrimination''.
It is also possible to use a more sophisticated neural architecture that exploits the permutation symmetry of the batch of samples~\citep{xu2018how}.

\paragraph{Learning to query strategies}

If players' strategies take additional inputs, such as observations or information sets, that the best-response function is not privy to as a player during a round of play, the best-response function can still analyze the strategies purely through their input-output behavior.
It would have to \emph{learn to query} the players' strategy with hypothetical observations or information sets, given its own.
The best-response function could also exploit the special structure of an extensive-form game by performing some form of \emph{local game tree analysis} of the strategy profile in the neighborhood of an information set.
\citet{Timbers_2022} apply a similar concept in the extensive-form setting to compute approximate best responses.

\paragraph{Dynamically-sized ensembles}

One possible extension of our method would be to dynamically adjust the size of the ensembles as training progresses.
One could start with ensembles of size 1 and, when the exploitability appears to stabilize for some period of time, add a new strategy to each ensemble, \textit{etc}.
One can either retain the previously-trained ensemble (a warm start) or reset it and start from scratch (a cold start).
Though the latter may seem inferior, it may be well-suited to games for which \(k\)-support best responses are drastically different from \(k+1\)-support best responses (where \(k\)-support means being supported on at most \(k\) pure strategies or actions).
Our experiments on the generalized rock paper scissors games showed that increasing the ensemble size helped up to a certain point, depending on the number of pure actions.

\paragraph{Black-box games}

In this paper, we assumed that the utility functions of the players are differentiable with respect to the strategy profile parameters, and that we have access to those gradients.
If these assumptions do not hold, we can use black-box smoothed gradient estimators~\citep{Duchi12:Randomized,Duchi_2015,Nesterov_2017,Shamir_2017,Berahas_2022}, such as natural evolution strategies~\citep{Wierstra_2008,Sun_2009,Wierstra_2014,Salimans_2017}.

\section{Code}

We use Python 3.12.2 with the following libraries:
\begin{itemize}
    \item \texttt{jax} 0.4.28 \citep{jax2018github}: \url{https://github.com/google/jax}
    \item \texttt{flax} 0.8.3 \citep{flax2020github}: \url{https://github.com/google/flax}
    \item \texttt{optax} 0.2.2 \citep{deepmind2020jax}: \url{https://github.com/google-deepmind/optax}
    \item \texttt{matplotlib} 3.8.4 \citep{Hunter:2007}: \url{https://github.com/matplotlib/matplotlib}
\end{itemize}
An implementation of our methods is shown below.

\begin{lstlisting}[
language=python,
breaklines=true,
basicstyle=\ttfamily\scriptsize,
postbreak=\mbox{\textcolor{red}{\(\hookrightarrow\)}\space},
breakindent=0em,
breakatwhitespace=true,
]
from functools import partial

import jax
import optax
from flax import linen as nn
from jax import numpy as jnp
from jax.flatten_util import ravel_pytree


def replace(seq, index, value):
  return seq[:index] + type(seq)([value]) + seq[index + 1 :]


def nikaido_isoda(u, x, y):
  response_utilities = (u(replace(x, i, yi))[i] for i, yi in enumerate(y))
  return sum(response_utilities) - sum(u(x))


def mix(x, axis=None, keepdims=False):
  ranks = 1 + x.argsort(axis)
  weights = ranks / ranks.max(axis, keepdims=True)
  return (x * weights).sum(axis, keepdims=keepdims)


class Responder(nn.Module):
  hidden_dim: int

  @nn.compact
  def __call__(self, x):
    x, unravel = ravel_pytree(x)
    x_skip = x
    x = nn.Dense(self.hidden)(x)
    x = nn.tanh(x)
    x = nn.Dense(x_skip.size)(x)
    x = unravel(x)
    return x


class StochasticGradientAscent:
  def __init__(self, utility_fn, optimizer):
    self.utility_fn = utility_fn
    self.optimizer = optimizer

  def init(self, params, key):
    opt_state = self.optimizer.init(params)
    return params, opt_state

  def grads(self, params, key):
    return jax.grad(self.utility_fn)(params, key)

  def step(self, state, key):
    params, opt_state = state
    grads = self.grads(params, key)
    grads = jax.tree.map(jnp.negative, grads)
    updates, opt_state = self.optimizer.update(grads, opt_state, params)
    params = optax.apply_updates(params, updates)
    return (params, opt_state), state

  def get_params(self, state):
    params, opt_state = state
    return params


class BestResponseFunction(StochasticGradientAscent):
  def __init__(self, utility_fn, optimizer, responder):
    self.utility_fn = utility_fn
    self.optimizer = optimizer
    self.responder = responder

  def init(self, params, key):
    x = params
    w = self.responder.init(key, params)
    opt_state = self.optimizer.init((x, w))
    return (x, w), opt_state

  def grads(self, params, key):
    def exploit_fn(params, key):
      x, w = params
      y = self.responder.apply(w, x)
      return nikaido_isoda(partial(self.utility_fn, key=key), x, y)

    gx, gw = jax.grad(exploit_fn)(params, key)
    return jax.tree.map(jnp.negative, gx), gw

  def get_params(self, state):
    (x, w), opt_state = state
    return x


class BestResponseEnsembles(StochasticGradientAscent):
  def __init__(self, utility_fn, optimizer, size):
    self.utility_fn = utility_fn
    self.optimizer = optimizer
    self.size = size

  def init(self, params, key):
    x = params
    w = jax.tree.map(lambda x: x[None].repeat(self.size, 0), x)
    opt_state = self.optimizer.init((x, w))
    return (x, w), opt_state

  def grads(self, params, key):
    x, w = params
    u = partial(self.utility_fn, key=key)

    def exploit_fn(x, w, reduce_fn):
      response_utilities = (
        jax.vmap(lambda yi, i=i: u(replace(x, i, yi))[i])(
          ensemble
        ).max()
        for i, ensemble in enumerate(w)
      )
      return sum(response_utilities) - sum(u(x))

    gx = jax.grad(exploit_fn, 0)(x, w, jnp.max)
    gw = jax.grad(exploit_fn, 1)(x, w, mix)
    return jax.tree.map(jnp.negative, gx), gw

  def get_params(self, state):
    (x, w), opt_state = state
    return x
\end{lstlisting}

\end{document}